%% file: main.tex
\documentclass[acmsmall]{acmart}

\usepackage[ruled,vlined]{algorithm2e}
\usepackage{algpseudocode}
\usepackage{setspace}
\usepackage{mathtools}
\usepackage{enumitem}
\usepackage{amsmath, bm}
\newtheorem{myDef}{Definition}
\newtheorem{myExa}{Example}
\usepackage{multirow}
\usepackage{stackengine}

\newcommand{\name}{$\mathtt{Starling}$}
\newcommand{\ID}[1]{$\boldsymbol{\mathtt{#1}}$}

\newcommand{\squishlist}{
	\begin{list}{$\bullet$}
		{ \setlength{\itemsep}{1pt}
			\setlength{\parsep}{1pt}
			\setlength{\topsep}{2.5pt}
			\setlength{\partopsep}{0.5pt}
			\setlength{\leftmargin}{1em}
			\setlength{\labelwidth}{1em}
			\setlength{\labelsep}{0.6em}
		}
	}
	\newcommand{\squishend}{
	\end{list}
}

\makeatletter
  \newcommand\figcaption{\def\@captype{figure}\caption}
  \newcommand\tabcaption{\def\@captype{table}\caption}
\makeatother
\AtBeginDocument{%
  \providecommand\BibTeX{{%
    \normalfont B\kern-0.5em{\scshape i\kern-0.25em b}\kern-0.8em\TeX}}}

\renewcommand\footnotetextcopyrightpermission[1]{} %
\newif\ifarxiv
\arxivtrue     %

\newif\ifconf
\conffalse   %

\setcopyright{none}
\begin{document}

\title{\texttt{Starling}: An I/O-Efficient Disk-Resident Graph Index Framework for High-Dimensional Vector Similarity Search on Data Segment}
\titlenote{This is the full version of the work being accepted for publication in \emph{Proc. of the ACM on Management of Data} (SIGMOD 2024). \url{https://doi.org/10.1145/3639269}}

% \affiliation{%

\author{Mengzhao Wang}
\authornote{Work done while working with Zilliz.}
\authornote{Both authors contributed equally to this research.}
\affiliation{%
  \institution{Zhejiang University}
  \country{China}}
\email{wmzssy@zju.edu.cn}
\orcid{0000-0003-3806-1012}

\author{Weizhi Xu}
\authornotemark[3]
\affiliation{%
  \institution{Zilliz}
  \country{USA}}
\email{Weizhi.Xu@zilliz.com}
\orcid{0009-0006-7384-7212}

\author{Xiaomeng Yi}
\affiliation{%
 \institution{Zhejiang Lab}
 \country{China}}
\email{xiaomeng.yi@zhejianglab.com}
\orcid{0000-0001-5792-9994}

\author{Songlin Wu}
\authornotemark[2]
\affiliation{%
  \institution{Tongji University}
  \country{China}}
\email{wusonglin@tongji.edu.cn}
\orcid{0009-0005-3560-436X}

\author{Zhangyang Peng}
\affiliation{%
  \institution{Hangzhou Dianzi University}
  \country{China}}
\email{pengzhangyang@hdu.edu.cn}
\orcid{0009-0002-5383-9750}

\author{Xiangyu Ke}
\affiliation{%
  \institution{Zhejiang University}
  \country{China}}
\email{xiangyu.ke@zju.edu.cn}
\orcid{0000-0001-8082-7398}

\author{Yunjun Gao}
\affiliation{%
  \institution{Zhejiang University}
  \country{China}}
\email{gaoyj@zju.edu.cn}
\orcid{0000-0003-3816-8450}

\author{Xiaoliang Xu}
\affiliation{%
  \institution{Hangzhou Dianzi University}
  \country{China}}
\email{xxl@hdu.edu.cn}
\orcid{0000-0001-8040-6809}

\author{Rentong Guo}
\affiliation{%
  \institution{Zilliz}
  \country{USA}}
\email{Rentong.Guo@zilliz.com}

\author{Charles Xie}
\affiliation{%
  \institution{Zilliz}
  \country{USA}}
\email{charles.xie@zilliz.com}

\renewcommand{\shortauthors}{Mengzhao Wang, et al.}
\renewcommand{\shorttitle}{\texttt{Starling}: An I/O-Efficient Disk-Resident Graph Index Framework}

\begin{abstract}
  High-dimensional vector similarity search (HVSS) is gaining prominence as a powerful tool for various data science and AI applications. As vector data scales up, in-memory indexes pose a significant challenge due to the substantial increase in main memory requirements. 
  A potential solution involves leveraging disk-based implementation, which stores and searches vector data on high-performance devices like NVMe SSDs. However, implementing HVSS for data segments proves to be intricate in vector databases where a single machine comprises multiple segments for system scalability. In this context, each segment operates with limited memory and disk space, necessitating a delicate balance between accuracy, efficiency, and space cost. Existing disk-based methods fall short as they do not holistically address all these requirements simultaneously.
  
  In this paper, we present {\name}, an I/O-efficient disk-resident graph index framework that optimizes data layout and search strategy within the segment. It has two primary components: 
  {\bf (1)} a data layout incorporating an in-memory navigation graph and a reordered disk-based graph with enhanced locality, reducing the search path length and minimizing disk bandwidth wastage; 
  and {\bf (2)} a block search strategy designed to minimize costly disk I/O operations during vector query execution. Through extensive experiments, we validate the effectiveness, efficiency, and scalability of {\name}. 
  On a data segment with 2GB memory and 10GB disk capacity, {\name} can accommodate up to 33 million vectors in 128 dimensions, offering HVSS with over 0.9 average precision and top-10 recall rate, and latency under 1 millisecond. The results showcase {\name}'s superior performance, exhibiting 43.9$\times$ higher throughput with 98\% lower query latency compared to state-of-the-art methods while maintaining the same level of accuracy.
\end{abstract}




\received{July 2023}
\received[revised]{October 2023}
\received[accepted]{November 2023}

\maketitle

\input{sections/1_intro}
\input{sections/2_pre}
\input{sections/3_index}
\input{sections/4_search}
\input{sections/5_exp}
\input{sections/6_discussion}
\input{sections/7_conclusion}

\begin{acks}
This work was supported in part by the NSFC under Grants No. (62025206, U23A20296), the Primary R\&D Plan of Zhejiang (2023C03198 and 2021C03156). Xiaomeng Yi and Xiaoliang Xu are the corresponding authors of the work.
\end{acks}

\bibliographystyle{ACM-Reference-Format}
\bibliography{myref}

\input{sections/appendix}


\end{document}

%% file: sections/1_intro.tex
\section{Introduction}\label{sec: intro}
The management of unstructured data, including video, image, and text, has become an urgent requirement \cite{Milvus_sigmod2021}. A notable advancement in this domain has been the widespread adoption of learning-based embedding models, which leverage high-dimensional vector representations to enable effective and efficient analysis and search of unstructured data \cite{le2014distributed, salvador2017learning}. High-dimensional Vector Similarity Search (HVSS) is a critical challenge in many domains, such as databases~\cite{graph_survey_vldb2021,NSG}, information retrieval~\cite{grbovic2018real,huang2020embedding}, recommendation systems~\cite{covington2016deep,okura2017embedding}, scientific computing~\cite{nasr2010hashing,zhu2020benchmark}, and large language models (LLMs)~\cite{li2023skillgpt,openai_hvss,retrieval-lm-tutorial}. The computational complexity associated with exact query answering in HVSS has spurred recent research efforts toward developing approximate search methods \cite{HNSW, NSG, graph_survey_vldb2021}.
While various compact index structures and intelligent search algorithms have been proposed to achieve a balance between efficiency and accuracy \cite{DPG, graph_survey_vldb2021, li2022deep}, the majority of these approaches assume that both the vector dataset and its corresponding index can be accommodated in the main memory (DRAM). 
However, in practical scenarios, the sheer volume of vector data often surpasses the capacity of the main memory, necessitating substantial expansion of main memory resources \cite{trilliondata, SPANN, li2022efficient} to support state-of-the-art in-memory HVSS algorithms like NSG \cite{NSG} and HNSW \cite{HNSW}. To illustrate, constructing an HNSW index based on one billion floating-point vectors in 96 dimensions would require more than 350GB of memory, resulting in out-of-memory failures \cite{HM_ANN} when executed on a typical single server \cite{bigann}.
Consequently, there is a growing demand for disk-based methods that leverage solid-state disks (SSDs) as a storage and retrieval medium for vector data, thereby circumventing the constraints imposed by main memory \cite{SPANN, DiskANN, zhang2019grip}. This paradigm shift allows for handling massive-scale vector datasets on a single machine, accommodating billions of vectors with improved efficiency and scalability.

However, maintaining a single large vector index on a solitary machine (Fig. \ref{fig: intro_segment}(a)) proves impractical for vector databases, as it constrains many essential system features for industrial applications \cite{Manu_zilliz}. For example, constructing a sole DiskANN index for a dataset at a billion-scale may demand over five days and peak memory usage of 1,100GB \cite{DiskANN}, significantly impacting system availability. Additionally, a large index exhibits poor migration capabilities, essential for scaling, load balancing, and fault tolerance \cite{Manu_zilliz}. Consequently, \textbf{vector databases consistently partition large-scale data into \textit{multiple} segments and assign an appropriate number of segments to a single machine} (Fig. \ref{fig: intro_segment}(b)) \cite{Milvus_sigmod2021,Manu_zilliz,Jingdong_paper}. Recent trends show that data fragmentation has become commonplace in mainstream vector databases \cite{Manu_zilliz,lanns,Jingdong_paper}. Each segment operates with limited storage capacity and computing resources \cite{trilliondata,Manu_zilliz}, and an independent medium-sized index is constructed on each segment for autonomous searching. Leveraging specific data fragmentation strategies and a query coordinator allows us to search only \textit{a few} segments of a machine during vector query execution \cite{Manu_zilliz,zhang2022leqat,deng2019pyramid}. This poses an open problem of efficient and accurate HVSS within a data segment in vector databases \cite{Manu_zilliz,milvus_archi}. In general, we may manage tens of millions of vectors on a typical segment with only several gigabytes of space \cite{trilliondata,Manu_zilliz}. The space left for index storage is very limited. Therefore, HVSS within the segment necessitates a meticulous equilibrium between \textit{search performance} (encompassing accuracy and efficiency) and \textit{space cost}.

\begin{figure}[!tb]
  \centering
  \setlength{\abovecaptionskip}{0.1cm}
  \includegraphics[width=0.6\linewidth]{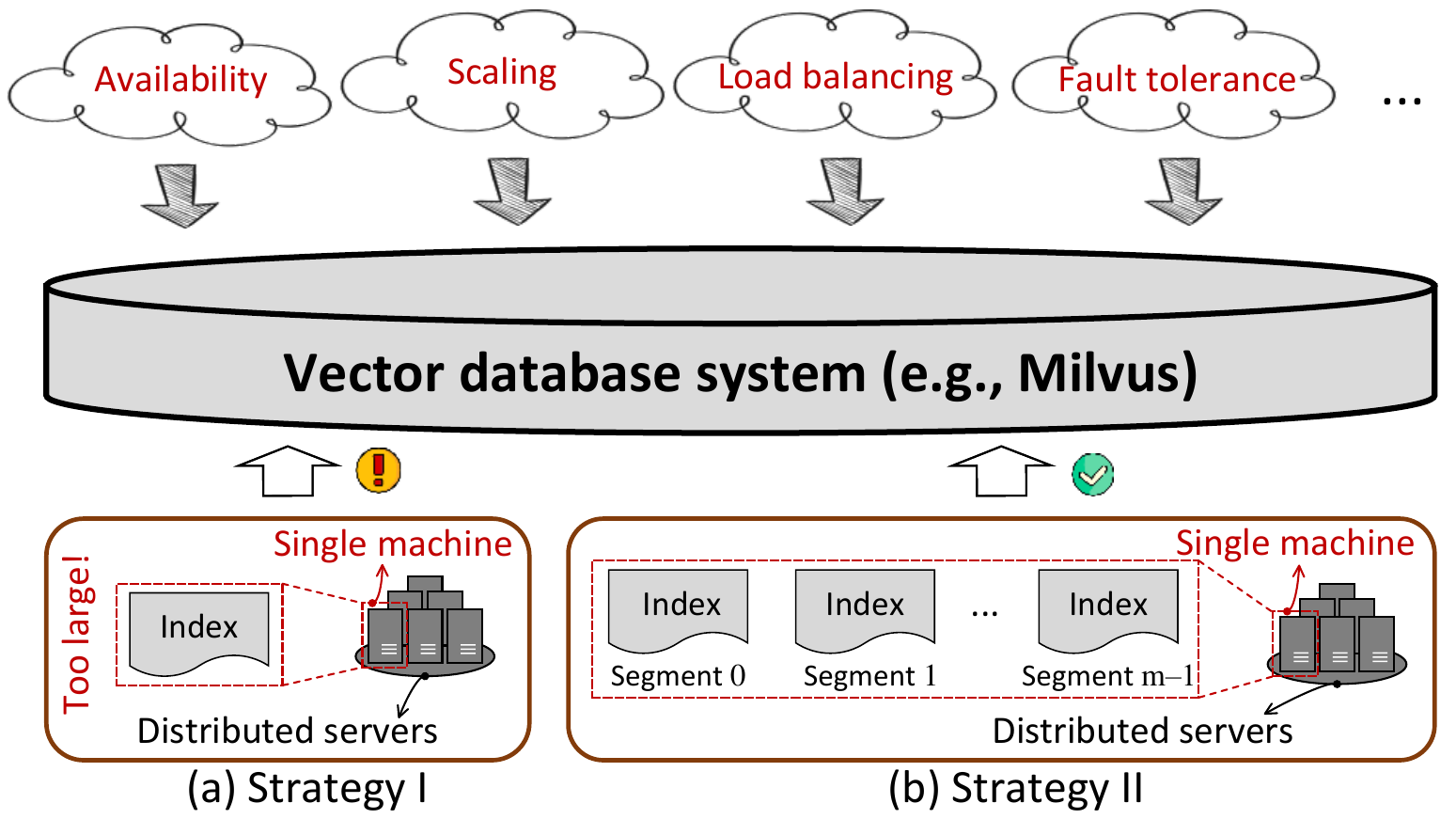}
  \caption{Two indexing strategies on a single machine for vector database system. In a distributed setting, different machines share the same strategy \cite{Manu_zilliz}.}
%   \Description{Native hybrid query framework.}
  \label{fig: intro_segment}
  \vspace{-0.6cm}
\end{figure}

No existing work has achieved effective HVSS on the data segment, considering both search performance and space cost. An immediate attempt is to use existing disk-based methods on the segment to handle large-scale vector data. However, as reported in the literature \cite{HM_ANN,SPANN} and also empirically verified in our experiments (\textbf{\S \ref{subsec: intra-seg}}), these techniques are in a quandary in this scenario. For example, two state-of-the-art disk-based schemes—SPANN \cite{SPANN} and DiskANN \cite{DiskANN}—lie in either of the two following extremes. At one extreme, SPANN \cite{SPANN} requires huge disk space to support efficient vector queries in the segment since each vector is copied multiple times (up to eight times \cite{SPANN}). Moreover, in vector databases, a segment may have multiple replicas that are distributed on different machines for fault tolerance \cite{Manu_zilliz}. This will further increase the disk space overhead for SPANN. At the other extreme, the vector index of DiskANN \cite{DiskANN} can fit the segment's disk capacity but DiskANN suffers from high latency due to the large number of random disk I/Os during the search process \cite{DiskANN,SPANN}. Therefore, it is particularly challenging to achieve a balance between search performance and space cost for HVSS on data segments.

We find that current disk-based methods for HVSS still follow the paradigm of memory-based solutions in their data layout and search strategy. This leads to numerous disk I/Os, which limit search performance under a given space cost. We illustrate this with DiskANN \cite{DiskANN}, the state-of-the-art disk-resident graph index algorithm\footnote{We exclude SPANN as it duplicates each vector up to eight times, which far exceeds the capacity of the data segment when serving tens of millions of vectors.} (cf. \textbf{\S \ref{subsec: disk_graph}}). \underline{First}, each hop along the search path corresponds to a disk I/O request. The search procedure on the graph index executes a sequence of disk I/O requests, which results in high search latency. \underline{Second}, the disk bandwidth is underutilized due to poor data locality. For each disk I/O in the search procedure, only one vertex’s information (a vector and its neighbor IDs) is \textit{relevant}, which typically takes a few hundred bytes. However, the smallest disk I/O unit is a block, which is usually 4KB in size. Therefore, most of the data read from the disk is wasted. According to our evaluation on the BIGANN dataset \cite{bigann}, DiskANN’s disk I/O operation takes up to 92.5\% of the time in the search procedure, and up to 94\% of the vertices (they are \textit{irrelevant}) in each loaded block are wasted (i.e., the vertex utilization ratio is low). Thus, there is much potential to achieve an I/O-efficient HVSS using graph index on the data segment.

In this paper, we introduce an I/O-efficient di\underline{\textbf{s}}k-residen\underline{\textbf{t}} gr\underline{\textbf{a}}ph index f\underline{\textbf{r}}amework tailored for high-dimensiona\underline{\textbf{l}} vector s\underline{\textbf{i}}milarity search o\underline{\textbf{n}} a data se\underline{\textbf{g}}ment, called~\name{}. {\name} optimizes data layout and search strategy to improve the search performance without additional space cost. The data layout consists of a sampled in-memory navigation graph and a reordered disk-based graph, which help find query-aware dynamic entry points and improve data locality, respectively. Based on this data layout, we design a search strategy that reduces disk I/Os by shortening the search path and increasing the vertex utilization ratio. Specifically, \name{} builds a navigation graph in memory by sampling a small portion of vectors. This allows the search procedure to quickly find some vertices that are close to the query vector without disk I/O, and then start searching on the disk-based graph from these vertices. To improve data locality on disk, \name{} reorganizes the graph index layout by shuffling the vertices and storing them with their neighbors in the same block (i.e., \textit{block shuffling}). This way, one disk I/O can load information from multiple relevant vertices in the search path. Based on this, we design a block search strategy that checks all the data from disk I/O and adds the ones that are potentially useful to the search sequence. We further improve the strategy by three computation-specific optimizations. Finally, we also propose efficient algorithms for approximate nearest neighbor search (ANNS) and range search (RS) queries based on the optimized search procedure.

To the best of our knowledge, this is the first work that tackles the fundamental data locality problem of disk-based graph index for the data segment by effective block shuffling algorithms. Our evaluation shows that {\name} achieves 43.9$\times$ higher throughput with 98\% lower query latency than state-of-the-art methods under the same accuracy. The main contributions of this work are:

\squishlist

\item We present~\name{}, an I/O-efficient disk-resident graph index framework with excellent universality that enhances the search efficiency of various graph index algorithms (e.g., Vamana \cite{DiskANN}, NSG \cite{NSG}, and HNSW \cite{HNSW}) on a data segment (\textbf{\S \ref{sec: overview}}, \textbf{\S \ref{subsec: scalability}}). 

\item We design a data layout for the data segment that consists of two components: an in-memory navigation graph (\textbf{\S \ref{subsec: navigation_graph}}) and a reordered disk-based graph (\textbf{\S \ref{subsec: graph_reorder}}). This data layout reduces the search path length and increases the vertex utilization ratio of each loaded block when searching on the disk-based graph index.

\item We prove that the block shuffling problem is NP-hard and has no polynomial time approximation algorithm with a finite approximation factor unless P=NP (\textbf{\S \ref{subsec: graph_reorder}}). Hence, we devise three effective heuristic shuffling algorithms that improve the locality of the disk-based graph index (\textbf{\S \ref{subsec: graph_reorder}}).

\item We propose a block search strategy that exploits the data locality of the reorganized index layout to reduce disk I/Os (\textbf{\S \ref{subsec: basic_page_search}}). Based on this strategy, we also provide three computation-specific optimizations (\textbf{\S \ref{subsec: basic_page_search}}) and develop search algorithms for two major types of queries in vector databases (\textbf{\S \ref{subsec: knn}} and \textbf{\S \ref{subsec: rs}}).

\item We implement~\name{} and evaluate it on four real-world datasets to verify its effectiveness, efficiency, and scalability (\textbf{\S \ref{sec: exp}}). We show that search efficiency can be improved significantly by simply adjusting the index layout on the disk.
\squishend

%% file: sections/2_pre.tex
\section{Preliminaries}\label{sec: pre}
In this section, we formulate two critical queries in vector databases. Then we explain why current methods are inefficient on the data segment and state our optimization goal in this paper.
\subsection{Two Query Types for HVSS}
\textbf{K Nearest Neighbor Search (KNNS).}
Given a vector dataset $X=\{x_1, x_2, \ldots, x_n\} \subset \mathbb{R}^D$, a query vector $q \in \mathbb{R}^D$, a distance function $dist(\mathbb{R}^D,\mathbb{R}^D)\rightarrow \mathbb{R}$, and a positive integer $k$ $(0<k<n)$, the KNNS problem finds a set $R_{knn}$ of $k$ vectors from $X$ that are closest to $q$ such that 
for any $x_i \in R_{knn}$ and $x_j \in X \setminus R_{knn}$, 
\begin{align}
    dist(x_i,q)\leq dist(x_j,q)\quad.
\end{align}
The distance function can be Euclidean distance or others.

The KNNS problem requires scanning all the vectors in the dataset, which is impractical for large vector datasets. Hence, most studies try to find approximate results that are close to the exact KNNS results. This is the Approximate Nearest Neighbor Search (ANNS) problem~\cite{Faiss,HNSW,NSG,NSSG}. ANNS uses $Recall \in [0,1]$ to measure the accuracy of approximate results. Let $R^\prime_{knn}$ be the approximate results, then $Recall$ is
\begin{equation}
\label{equ: recall}
    Recall = \frac{|R_{knn} \cap R^\prime_{knn}|}{k}\quad.
\end{equation}
A higher recall means a more accurate result. Many studies have designed different index structures for vector data to improve the trade-off between search efficiency and accuracy.

\vspace{0.2em}
\noindent\textbf{Range Search (RS).}
Unlike the KNNS problem that returns a fixed number of results, the RS problem retrieves all the vectors that are within a given distance $r$ from the query vector $q$. Average precision $(AP)\in [0,1]$ is a metric to measure the accuracy of approximate RS results. Given the exact result $R_{range}$ and an approximate solution that ensures all the returned results $R^\prime_{range}$ are within $r$, then $AP$ can be computed as follows:
\begin{equation}
\label{equ: ap}
    AP = \frac{|R^\prime_{range}|}{|R_{range}|}
\end{equation}

\subsection{HVSS on Data Segment.} \label{subsec: disk_graph}
Vector databases divide large-scale data into multiple segments and distribute them across servers \cite{lanns,Manu_zilliz}. Each server may manage numerous segments, processing vector queries either in parallel or serially through a query coordinator. Typically, each segment operates within strict memory and disk space constraints \cite{trilliondata,Manu_zilliz}, often possessing less than 2GB of memory and under 10GB of disk capacity. Data segment facilitates system scalability \cite{Manu_zilliz}, yet it also poses challenges for HVSS within a data segment. Within this context, our objective is to achieve high search accuracy and efficiency within these space limitations. This necessitates a delicate equilibrium between search performance and the space cost for HVSS. For instance, on the BIGANN dataset \cite{bigann}, each segment might accommodate 33 million vectors, requiring 4GB of storage. Consequently, the remaining space for index storage becomes exceedingly limited. Note that we consistently construct a substantial index to achieve a better accuracy and efficiency trade-off. In the following, we delve into the reasons why current HVSS advancements struggle within the context of data segments.

Mainstream in-memory algorithms such as HNSW \cite{HNSW}, NSG \cite{NSG}, and HVS \cite{HVS} encounter challenges when handling tens of millions of vectors on a segment due to their approach of loading all vectors and index into memory, exceeding memory limitations. While some vector compression methods like LSH \cite{HuangFZFN15} and PQ \cite{PQ} store compressed vectors in memory, the compression process introduces significant errors that notably degrade search accuracy. For instance, the top-1 recall rate of the leading compression method seldom surpasses 0.5 \cite{DiskANN}. Disk-based solutions present more promise as they require less memory and achieve high search accuracy. Traditional disk-based methods based on trees are excluded due to their susceptibility to the ``curse of dimensionality'' when addressing HVSS \cite{DPG}. Instead, the focus is directed towards two state-of-the-art disk-based methods for HVSS: SPANN \cite{SPANN} and DiskANN \cite{DiskANN}. SPANN achieves high search accuracy and efficiency but at the cost of substantial storage space. It may replicate each vector up to eight times, leading to extensive disk overhead that far surpasses the capacity of the data segment when managing tens of millions of vectors. On the other hand, DiskANN follows the graph index algorithm and achieves superior search performance with less disk capacity. It has served as a baseline for HVSS on disk \cite{bigann}. However, its data layout and search strategy still align with the paradigm of memory-based solutions, necessitating disk I/O for each vertex access. In summary, devising an efficient and accurate solution for HVSS on a data segment remains challenging, involving the management of tens of millions of high-dimensional vectors within limited space.

\subsection{Our Optimization Objective} \label{subsec: optimization_objective}
We draw two conclusions from our observations. \underline{First}, the disk-resident graph index proves to be the optimal choice for HVSS on the data segment, offering improved search performance with minimal space overhead \cite{HM_ANN}. \underline{Second}, the disk-based graph index can be further optimized for I/O-efficiency in HVSS. Consequently, the HVSS problem on the data segment can be framed as enhancing the I/O-efficiency of disk-resident graph index. Let $T_{I/O}$, $T_{comp}$, and $T_{other}$ represent the I/O time, computation time, and other times (such as data structure maintenance) when conducting searches on the disk-based graph, respectively. The total search time can be denoted as 
\begin{equation}
\label{equ: search time}
    T_{total}=T_{I/O}+T_{comp}+T_{other}.
\end{equation}
In Eq. \ref{equ: search time}, $T_{I/O}$ emerges as the dominant factor influencing the search efficiency. For instance, upon visualizing the search time costs of DiskANN on BIGANN (Fig. \ref{fig: block search optimization}(d)), it becomes evident that $T_{I/O}$ constitutes up to 92.5\% of $T_{total}$, while $T_{comp}$ and $T_{other}$ collectively occupy less than 7.5\%. This highlights the efficiency bottleneck as the cost of disk I/Os for HVSS on the disk-resident graph index. Therefore, our objective is to improve search efficiency by minimizing $T_{I/O}$. Note that we must not introduce additional space overhead (including memory and disk), compared to the original graph index, in order to adhere to the space capacity limitation of the data segment.

\begin{figure*}[!tb]
  \centering
  \setlength{\abovecaptionskip}{0cm}
  \setlength{\belowcaptionskip}{-0.4cm}
  \includegraphics[width=\linewidth]{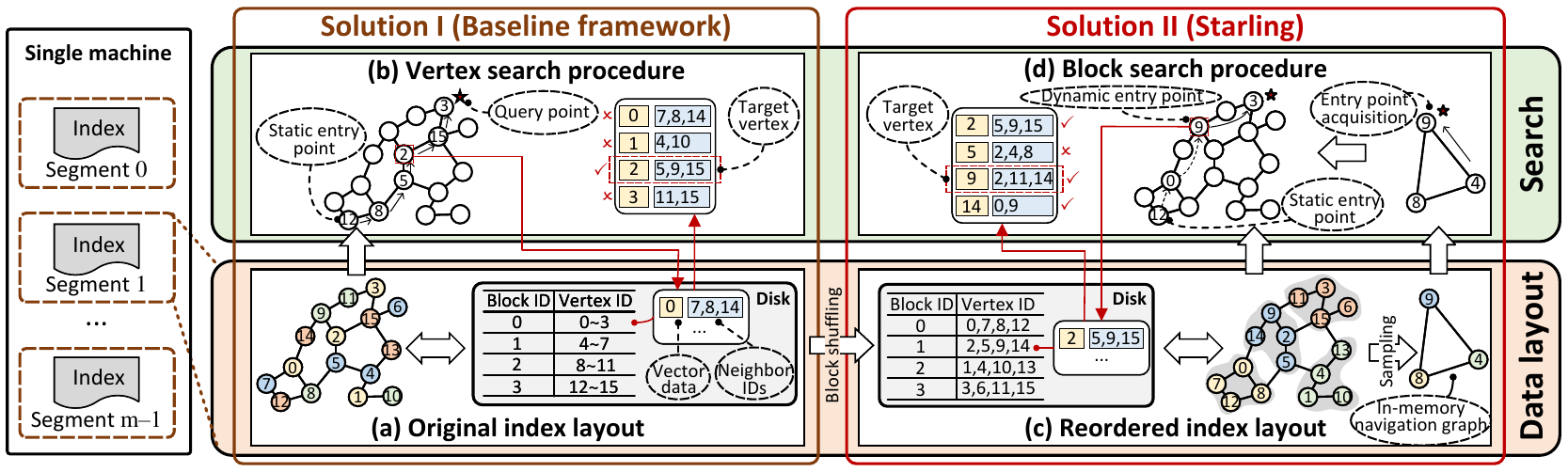}
  \caption{{Illustration of the data layouts and search strategies for the baseline and {\name}, respectively.}}
%   \Description{Native hybrid query framework.}
  \label{fig: overview_graph_layout}
%   \vspace{-0.1cm}
\end{figure*}

\section{Design Philosophy}\label{sec: overview}
We analyze two key factors that affect I/O time for existing graph index methods. Then we give an overview of~\name{}, illustrating its data layout and search strategy on the data segment.

\subsection{I/O-efficiency Analysis}
\label{subsec: problem_analysis}
When conducting searches on the disk-based graph, I/O-efficiency depends on two primary factors: the vertex utilization ratio in each disk I/O and the length of the search path. The former indicates the extent to which useful vertices are loaded in a given block, while the latter denotes the number of hops required from the entry point to the result point. Both of these factors influence the number of disk I/Os and thus determine the I/O time.
A higher vertex utilization ratio signifies that more relevant information is loaded in each disk I/O, leading to more effective utilization of disk bandwidth and requiring fewer disk I/Os for each query. A longer search path implies an increased number of disk I/Os during the search process. We delve into an analysis of the vertex utilization ratio and the search path length within the \textit{baseline framework}\footnote{Unless specifically stated, we refer to DiskANN as the baseline framework.} and reveal two underlying issues.
 
\vspace{0.2em}
\noindent\textbf{Problem 1: Poor data locality.} Fig.~\ref{fig: overview_graph_layout}(a) shows that the baseline assigns ID-consecutive vertices (a vertex contains the vector data and neighbor IDs) to the same blocks. For example, block \ID{0} stores vertices \ID{0}$\sim$\ID{3}. Since a block is the smallest disk I/O unit, reading vertex \ID{2} also reads vertices \ID{0}, \ID{1}, and \ID{3}. This wastes disk bandwidth as the other three vertices are irrelevant to vertex \ID{2}. A naive way to avoid the wastage is to check all the vertices in a block. However, a segment has up to $10^7\times$ more vertices than a block in a real-world scenario. The probability of finding a near vertex (w.r.t the target vertex) among the non-target vertices\footnote{We load a block according to a target vertex. In this example, the target vertex is \ID{2}.} in the block is too low to improve efficiency effectively. Therefore, the baseline only checks the target vertex and discards the rest. This means a low {vertex utilization ratio} in each disk I/O. According to our evaluation on BIGANN (Tab. \ref{tab: vertex utilization and search path}), up to 94\% of data read from the disk is wasted. Obviously, ID-consecutive vertices in a block do not imply spatial proximity (e.g., vertices of the same color are scattered on the graph topology in Fig. \ref{fig: overview_graph_layout}(a)). Therefore, the baseline has poor data locality.

\vspace{0.2em}
\noindent\textbf{Problem 2: Long search path.} The I/O complexity of searching on a disk-based graph index is proportional to the {search path length}. The baseline framework uses a fixed or random vertex as the entry point for the search. However, a segment may have tens of millions of vertices, so the entry point may be far from the query. In this case, only the last few vertices in the path are near the query and likely to be in the final result. However, we need to read each vertex along the path from the disk. Fig. \ref{fig: overview_graph_layout}(b) shows an extreme example where the entry point is the farthest vertex from the query point. It takes five hops to search from \ID{12} to \ID{3}. If we start from \ID{9}, we can reduce the hops to two. For a dataset of tens of millions of vectors, even searching for only the top-10 nearest neighbors may generate a path of hundreds of hops to achieve high accuracy. Therefore, the baseline has a long search path.

\subsection{Framework Overview}
We present~\name{}, a framework designed to enhance the I/O-efficiency of disk-based graph index for HVSS on a data segment. {While a single machine may encompass multiple segments, they all share the same index layout and search strategy. Therefore, our primary focus is on one specific segment}. As depicted in Fig. \ref{fig: overview_graph_layout}(c), \name{} preserves the graph's topology while optimizing the data layout to augment the vertex utilization ratio and diminish the search path length. Moreover, it employs a block search strategy aimed at reducing disk I/Os. Below, we provide an overview of the data layout and search strategy employed by \name{}:

\vspace{0.2em}
\noindent\textbf{Data layout on disk.} \name{} improves data locality by shuffling the data blocks in accordance with the graph topology (Fig. \ref{fig: overview_graph_layout}(c) left). {This aligns the data layout with the search procedure, which tends to visit neighboring vertices \cite{graph_survey_vldb2021}.} Specifically, \name{} endeavors to store a vertex and its neighbors in the same block, enabling a single disk read to fetch not only the target vertex but also other likely candidates. This approach increases the {vertex utilization ratio}, as each disk I/O operation brings multiple relevant vertices. Furthermore, the search procedure can more readily jump to a closer vertex from the query, potentially reducing the number of required disk I/Os.

\vspace{0.2em}
\noindent\textbf{Data layout in memory.} \name{} identifies query-aware entry points for the disk-based graph using an in-memory navigation graph (Fig. \ref{fig: overview_graph_layout}(c) right). {This approach incorporates the concept of multi-layered graphs~\cite{HVS,HNSW,HM_ANN} and involves sampling a small fraction of vectors from the disk-based graph to construct the navigation graph. During this process, any in-memory graph algorithm~\cite{NSG,NSSG} can be utilized.} Given a query, \name{} initially explores the in-memory navigation graph to obtain query-close vertices as entry points. Subsequently, it initiates the disk-based graph search from these identified points. This methodology effectively reduces the {search path length} on the disk-based graph.

\vspace{0.2em}
\noindent\textbf{Search strategy.} \name{} uses a block search strategy to exploit data locality (Fig. \ref{fig: overview_graph_layout}(d)). Unlike the baseline, which explores data on a vertex basis, this strategy processes data by blocks. This way, it benefits from the optimized data layout. For each loaded block, {\name} computes the distance to the query for all vertices. Then, it selects the vertices that are close to the query and checks their neighbors for new search candidates. This strategy lowers the disk operation cost by exploring more relevant data per block but increases the computation cost. We further optimize the performance with three computation-specific optimizations (\textbf{\S \ref{subsec: basic_page_search}}).

Example~\ref{exa: basic_page_search} shows how~\name{} reduces disk I/Os.

% \vspace{-0.2cm}
\begin{myExa}
  \label{exa: basic_page_search}
  \rm Fig.~\ref{fig: overview_graph_layout}(b) shows a vertex search strategy that starts from vertex \ID{12} and reaches the result vertex \ID{3} in five hops for a given query point. This strategy needs at least six disk I/Os to access the vertices in the path and their neighbors. This is inefficient, especially for large-scale scenarios, because the probability of finding a vertex close to the query among the non-target vertices in a loaded block is very low. Therefore, exploring data on a block basis with the original index layout is more harmful than helpful. In contrast, \name{} reduces the disk I/Os to three with a reordered index layout, as shown in Fig. \ref{fig: overview_graph_layout}(d). The block search procedure works as follows: (1) It loads block \ID{0} and visits vertices \ID{12}, \ID{0}, \ID{7}, and \ID{8} in the block. Then it computes their distance to the query and selects \ID{0} to visit its neighbors. (2) It loads block \ID{1} with vertices $\{$\ID{2}, \ID{5}, \ID{9}, \ID{14}$\}$ and chooses \ID{9} as the next hop. Vertex \ID{3} as the search result will be found when~\name{} loads the next block according to vertex \ID{11}. With a navigation graph built on the vectors of $\{$\ID{4}, \ID{8}, \ID{9}$\}$ in memory, it obtains \ID{9} as the entry point for disk-based graph search. Since \ID{9} is close to the query, it only takes two disk I/Os (one to visit the entry point and one to visit its neighbors) to get the query result.
\end{myExa}

%% file: sections/3_index.tex
\section{Data Layout}
\label{sec: index}
Fig. \ref{fig: overview_graph_layout}(c) shows how {\name} organizes data on a segment. \underline{First}, it builds a disk-based graph index on the full dataset and rearranges it by block shuffling to improve data locality. We can use different methods to construct {\name}'s disk-based graph, such as NSG \cite{NSG}, HNSW \cite{HNSW}, and Vamana \cite{DiskANN}. For more details on these methods, please refer to the original papers. We do not focus on developing a specific graph index algorithm since existing ones are well-studied but not suitable for disk deployment. Instead, we address the block shuffling problem on the disk-based graph (\textbf{\S \ref{subsec: graph_reorder}}). \underline{Second}, {\name} samples some data points from the full dataset and builds an in-memory navigation graph (\textbf{\S \ref{subsec: navigation_graph}}). This structure allows searching on the disk-based graph to begin from some query-aware entry points, which reduces the search path length.

\subsection{Block Shuffling on the Disk}\label{subsec: graph_reorder}
\textbf{Notations.} Let $G=(V,E)$ represent a disk-based graph index, where $V$ and $E$ denote the sets of vertices and edges, respectively. The Edges are directed and are stored as adjacency lists of vertices. Each vertex necessitates $\gamma$ KB of storage to house its vector data, neighbor count $\lambda$, and a list of neighbor IDs (with a maximum length of $\Lambda$)\footnote{Each vertex is allocated $\Lambda$ ID spaces for alignment, and padding is added when $\lambda<\Lambda$. For simplicity, the neighbor count $\lambda$ is omitted in Fig. \ref{fig: overview_graph_layout} and \ref{fig: reordering}.}. The size of a disk block is $\eta$ KB. Since we do not split the data of a vertex into two blocks, each block can accommodate at most $\varepsilon=\lfloor \eta/\gamma \rfloor$ vertices. Therefore, $\rho = \lceil |V|/\varepsilon \rceil$ blocks are required to store the graph index.

Next, we give the definition of block-level graph layout:

\begin{myDef}
  \label{def: page_based_graph_layout}
  \textbf{Block-Level Graph Layout.} The block-level graph layout of a graph index is a scheme that assigns $|V|$ vertices to $\rho$ blocks.
\end{myDef}

\begin{myExa}
  \label{exa: diskann_layout}
  \rm Fig. \ref{fig: overview_graph_layout}(a) shows a graph index $G=(V,E)$ with 16 vertices ($|V|=16$), four vertices per block ($\varepsilon=4$), and four blocks in total ($\rho=4$). For DiskANN \cite{DiskANN_code} on 33 million BIGANN dataset ($|V|=3.3 \times 10^7$) \cite{bigann}, each vector is 128-dimensional and one byte per value. If $\Lambda$ is 31 and $\eta$ is 4 KB, then $\gamma=(128+4+31\times 4) /1,024$ KB (ID is unsigned integer type), $\varepsilon = 16$, and $\rho=2,062,500$. Thus, DiskANN puts 16 ID-contiguous vertices in a block, such as $0 \sim 15$.
\end{myExa}

Fig. \ref{fig: overview_graph_layout}(a) and (c) show two different ways of organizing the graph index on disk block level. As discussed in \textbf{\S \ref{subsec: problem_analysis}} (Problem 1), the graph layout affects disk I/Os during searching. Specifically, a layout without data locality lowers the vertex utilization ratio and increases random disk I/Os. We use the overlap ratio $OR(G)$ of the graph index $G$ to measure the locality of the graph layout. For any vertex $u \in V$, $OR(G)$ is the average of $OR(u)$ over $V$, where $OR(u)$ is the proportion of vertices that are $u$'s neighbors among all the vertices except $u$ in the block. We can compute $OR(u)$ by
\begin{equation}
  \label{equ: overlap_ratio}
  OR(u) = \left\{
  \begin{matrix}
  \frac{|B(u)\cap N(u))|}{|B(u)|-1} && |B(u)|>1 \\
  0 && |B(u)|\leq 1
  \end{matrix}
  \right.\quad,
\end{equation}
where $B(u)$ is the set of vertices in the block that contains $u$ and $N(u)$ is the set of $u$'s neighbors ($|N(u)|\leq\Lambda$). In real-world datasets, the maximal vertex count $\varepsilon$ is about 10 in each block and there is at most one block with a different $\varepsilon$ (i.e., $|V|$ is not divisible by $\varepsilon$) in graph layout. We use $OR(B(u))=\sum_{v\in B(u)} ({OR(v)}/{|B(u)|})$ to denote the overlap ratio of the block $B(u)$. $OR(G)$ is calculated as $\sum_{u\in V} ({OR(u)}/{|V|})$. A higher $OR(G)$ means better data locality and a higher vertex utilization ratio in a loaded block.

\begin{myExa}
  \label{exa: overlap_ratio}
  \rm A graph layout with optimal data locality has $OR(G)$ $= 1$, meaning that every vertex in a block is a neighbor of any other vertex in the block. However, we get $OR(G)$ close to 0 for the DiskANN \cite{DiskANN} graph index on the 33 million BIGANN (Fig. \ref{fig: data_locality}(a)). This indicates the poor data locality of the DiskANN graph layout.
\end{myExa}

To enhance data locality, we use \textit{block shuffling} to adjust the graph layout, which is defined as:
\begin{myDef}
  \label{def: graph_reordering}
  \textbf{Block Shuffling.} Given a graph layout for a disk-based graph index $G$, the block shuffling aims to get a new layout that maximizes the $OR(G)$ while satisfying Def. \ref{def: page_based_graph_layout}.
\end{myDef}

In Def. \ref{def: graph_reordering}, we aim to get a graph layout where every vertex in a block is a neighbor of other vertices within the same block. However, the graph index built on high-dimensional vectors is complex, as the neighborhood relationship encompasses both navigation and similarity aspects \cite{DPG,NSSG}. {A vertex may have neighbors that belong to different clusters \cite{HNSW}, and all vertices exhibit a constant out-degree.} Hence, it is challenging, if not impossible, to ensure that any two vertices in a block are mutual neighbors. We prove that the block shuffling problem is NP-hard in Theorem \ref{theorem: graph reordering np-hard}. Furthermore, it lacks a polynomial time approximation algorithm with a finite approximation factor. This motivates our block shuffling research, as the baseline framework merely populates a block with ID-contiguous vertices.

\begin{theorem}
  \label{theorem: graph reordering np-hard}
  The block shuffling problem is NP-hard and does not have a polynomial time approximation algorithm with a finite approximation factor unless P=NP.
\end{theorem}

\begin{proof}
(Sketch.) We reduce the block shuffling problem to the triple shuffling problem, which is strongly NP-complete \cite{book1980michael,andreev2004balanced}. The triple shuffling problem is defined as follows: given $t=3\cdot \rho$ integers $\alpha_0$, $\alpha_1$, $\cdots$, $\alpha_{t-1}$ and a threshold $\Omega$ such that $\Omega/4 < \alpha_i < \Omega/2$ and 
\begin{equation}
  \label{equ: block reconstruction proof}
  \sum_{i=0}^{t-1} \alpha_i = \rho \cdot \Omega \quad,
\end{equation}
the task is to partition the numbers into $\rho$ triples and DECIDE if these triples can be shuffled by swapping numbers between triples so that each triple sums up to $\Omega$. We construct a graph $G$ with cliques of size $\alpha_i$ for each integer $\alpha_i$. Let $\Omega$ and $\rho$ be the number of vertices in a block and the number of blocks in the graph layout of $G$, respectively. We demonstrate that the block shuffling problem on $G$ has a solution iff the triple shuffling problem can be solved. Consequently, the block shuffling problem is NP-hard. Assuming the existence of a polynomial time approximation algorithm with a finite approximation factor for the block shuffling problem, we could use it to solve the triple shuffling problem. However, we establish that the triple shuffling problem cannot be solved by such an approximation algorithm with a finite approximation factor. This contradicts the assumption. (The detailed proof is available in Appendix \ref{appendix: proof theorem 4.1})
\end{proof}

We design three heuristic algorithms to handle the shuffling problem. We first describe a straightforward solution (\textit{Algorithm I}), and then present two optimized algorithms (\textit{Algorithms II–III}).

\vspace{0.5em}
\noindent\textit{\textbf{Algorithm I: Block Neighbor Padding (BNP).}}
This algorithm fills disk blocks in a one-by-one fashion. To fill blocks, it checks vertices in ascending order of IDs. If a vertex $u$ is not assigned to any block yet, then the algorithm tries to assign the vertex and its neighbors to the current block. Once a block is full the algorithm opens a new block and assigns vertices to it. BNP has a time complexity of $O(|V|)$ and improves the overlap ratio $OR(G)$ by assigning vertices and their neighbors to the same block. However, the improvement is limited as some neighbors of $u$ (denoted as $z, v \in N(u)$) may not be adjacent to each other, which lowers $OR(z)$ or $OR(v)$. Also, some neighbors of $u$ may have been assigned to other blocks earlier, such as $z$ also being a neighbor of $o$ with a smaller ID than $u$. It cannot be stored with $u$ as we store each vertex only once.

\begin{figure}[!tb]
  \centering
  \setlength{\abovecaptionskip}{0.1cm}
  \setlength{\belowcaptionskip}{-0.15cm}
  \includegraphics[width=0.6\linewidth]{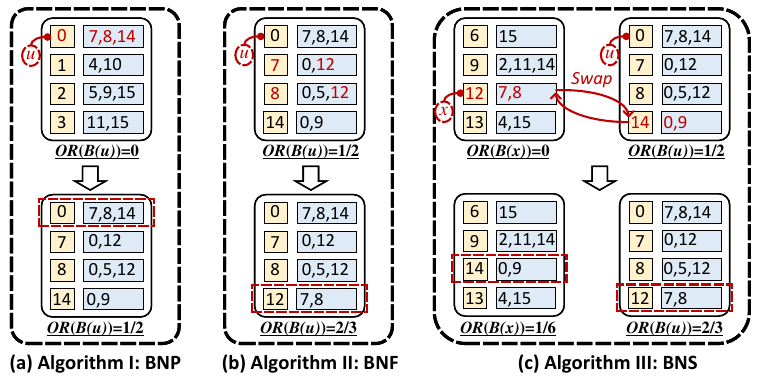}
  \caption{{Block shuffling. Refer to Fig. \ref{fig: overview_graph_layout} for graph topology.}}
  \label{fig: reordering}
\end{figure}

\begin{myExa}
  \label{exa: straightforward_solution}
  \rm In Fig. \ref{fig: reordering}(a), we have $B(u)$ $=$ $\{$ \ID{0}, \ID{1}, \ID{2}, \ID{3} $\}$ and $OR(B(u))$ $=0$ for the original layout (the ID of $u$ is \ID{0}). BNP puts \ID{0} and its neighbors \ID{7}, \ID{8}, \ID{14} in $B(u)$, so that $B(u)=\{$ \ID{0}, \ID{7}, \ID{8}, \ID{14} $\}$ and $OR(B(u)) = 1/2$. But vertex \ID{12} cannot be in the same block with its neighbors \ID{7} and \ID{8}, that is $OR(v)=0$ for vertex $v$ with ID \ID{12}.
\end{myExa}

\noindent\textit{\textbf{Algorithm II: Block Neighbor Frequency (BNF).}} To further optimize $OR(G)$, we propose BNF, which aims to assign a vertex to the block that holds most of its neighbors, i.e., the block with the highest neighbor frequency. The BNF takes the result of BNP as an initial layout and optimizes $OR(G)$ iteratively. As illustrated in Algorithm \hyperref[alg: npf]{1}: (1) It stores the current mapping of vertex IDs to block IDs and clears all the blocks of $G$ (lines 3–5). (2) For each $u$ in $V$, BNF tries to assign it to the blocks that hold its neighbors in the previous iteration (lines 6–14). The algorithm checks blocks in descending order of neighbor count and assigns $u$ to the first block that is not full yet. If all the blocks are full, then BNF puts $u$ in an empty block in $\mathcal{B}$ (lines 13–14). (3) It repeats this process until it reaches the iteration limit $\beta$ or the $OR(G)$ gain between two iterations falls below a given threshold (lines 15–16). Then a new layout of $G$ is returned. {BNF has a time complexity of $O(\beta \cdot o \cdot|V|)$, where $\beta$ is the number of iterations, $o$ is the average out-degree, and $|V|$ is the number of vertices. This is because BNF needs to access the neighbors of all vertices in each iteration.}

\setlength{\textfloatsep}{0pt}
\begin{algorithm}[t]
\label{alg: npf}
  \caption{\textsc{Block Shuffling by BNF}}
  \LinesNumbered
  \KwIn{block-level graph layout of $G=(V,E)$ from BNP, maximum iterations $\beta$, $OR(G)$ gain threshold $\tau$}
  \KwOut{new block-level graph layout of $G$}
  
  $\mathcal{B}$$=$$\{B_0, \cdots, B_{\rho-1} \} \gets$ all blocks \textcolor{blue}{\Comment{\textsf{$\rho$ is the number of blocks}}}
  
  \While{iterations $\leq \beta$}{
    $D$ $\gets$ mapping of vertex IDs to block IDs;
    
    % $OR(G)_s$ $\gets$ current overlap ratio of graph layout of $G$;
    
    \ForAll{$B_i\in \mathcal{B}$}{
      $B_i \gets \emptyset$; \textcolor{blue}{\Comment{\textsf{clear all blocks}}}
    }
    \ForAll{$u \in V$}{
      $H$ $\gets$ $\bigcup_{a\in N(u)}\{D(a)\}$; \textcolor{blue}{\Comment{\textsf{all neighbors' block IDs}}}
    %   $\mathcal{P}_j$ $\gets$ all pages including $j$'s neighbors based on $D$;
    
      \While{$H \neq \emptyset$}{
        $x \gets$ block ID with the most neighbors in $H$;
        
        \If{$B_x$ is not full}{
          $B_x \gets B_x \cup \{u\}$;
          break;
        }
        $H$ = $H \setminus \{x \}$; \textcolor{blue}{\Comment{\textsf{remove the block that is full}}}
      }
      
      \If{$H = \emptyset$}{
        add $u$ to a empty block in $\mathcal{B}$;
      }
    }
    % $OR(G)_e$ $\gets$ overlap ratio of $G$ after an iteration;
    
    \If{$OR(G)$ gain $< \tau$}{
      break;
    }
  }
  
  \textbf{return} new layout of $G$
\end{algorithm}
\setlength{\textfloatsep}{12pt plus 2pt minus 2pt}

\begin{myExa}
  \label{exa: npf}
  \rm As shown in Fig. \ref{fig: reordering}(b), BNF replaces \ID{14} with \ID{12} in $B(u)$, since \ID{12} has two neighbors, \ID{7} and \ID{8}, in the block. Thus, $B(u)$ becomes $\{$ \ID{0}, \ID{7}, \ID{8}, \ID{12} $\}$ and $OR(B(u))=2/3$.
\end{myExa}

\noindent\textit{\textbf{Algorithm III: Block Neighbor Swap (BNS).}}
The design of BNS is inspired by the NN-Descent method \cite{NNDescent} in the literature of graph construction. BNS refines $OR(G)$ from an initial layout, which could be the result of BNP or BNF. In BNS, for a pair of vertices $a$ and $e$ that are neighbors of vertex $u$ and belong to different blocks $B(a)$ and $B(e)$, BNS swaps vertices with the lowest overlap ratio ($OR$) in $B(a)$ and $B(e)$ to increase the sum of $OR(B(a))$ and $OR(B(e))$. {Let $o$ be the average out-degree, then the number of swaps is $o^2$ for each vertex in $V$. In each swap, the time complexity of computing one block's $OR$ is $O(o\cdot \varepsilon)$, where $\varepsilon$ is the number of vertices in a block. BNS operates in an iterative manner, where each iteration checks vertex pairs among the neighbor set of all vertices. Therefore, the time complexity of BNS is $O(\beta \cdot o^3 \cdot \varepsilon \cdot |V|)$, considering the number of iterations $\beta$.} We prove that $OR(G)$ does not decrease with the number of iterations in Lemma \ref{lemma: nps OR non-decreasing}.

\begin{myExa}
  \label{exa: npd}
  \rm In Fig. \ref{fig: reordering}(c), BNS identifies the blocks $B(u)$ and $B(x)$ that contain vertices \ID{0} and \ID{12}, respectively, from the neighbors of vertex \ID{7}. Then, it swaps vertices \ID{12} and \ID{14}, which have the lowest $OR$ in their blocks and can increase $(OR(B(x))+OR(B(u)))$ by swapping. Finally, we get $OR(B(x))=1/6$ and $OR(B(u))=2/3$.
\end{myExa}

\begin{lemma}
\label{lemma: nps OR non-decreasing}
In BNS, the $OR(G)$ is a monotonically non-decreasing function of the number of iterations $\beta$.
\end{lemma}

\begin{proof}
In each iteration, an update is local and affects only two blocks' vertices. For any two neighbors $a$ and $e$ of vertex $u$ ($B(a)\neq B(e)$), we prove that swapping vertices does not lower the sum of $OR(B(a))$ and $OR(B(e))$. Let $OR(B(a))_i$, $OR(B(e))_i$ and $OR(B(a))_j$, $OR(B(e))_j$ be the values of $OR(B(a))$, $OR(B(e))$ before and after an iteration, respectively. We swap two vertices in $B(a)$ and $B(e)$ only if $OR(B(a))_j$ + $OR(B(e))_j$ $>$ $OR(B(a))_i$ + $OR(B(e))_i$. Therefore, the sum of $OR(B(a))$ and $OR(B(e))$ is non-decreasing.
\end{proof}

\vspace{-0.25cm}
\noindent\textbf{Analysis.} Among our three block shuffling algorithms, BNP emerges as the fastest since it scans all vertices just once. On the other hand, BNF's efficiency is contingent on the number of iterations and does not ensure the convergence of $OR(G)$. However, BNF demonstrates proficiency, both in terms of efficiency and effectiveness, particularly on numerous real-world datasets. Meanwhile, BNS, although possessing the highest time complexity, guarantees that $OR(G)$ does not decrease with iterations. In our implementation, we have parallelized BNF and BNS to enhance their speed. All three algorithms notably improve $OR(G)$ in contrast to the original graph layout, with BNS exhibiting the most significant improvement, followed by BNF and then BNP. In our evaluation (\textbf{\S \ref{subsec: ablation_study}}), higher $OR(G)$ improves vertex utilization ratio and search performance. In practice, the choice of algorithm can be tailored to specific requirements. Generally, we recommend BNF due to its adept balance between efficiency and effectiveness.

\vspace{0.2em}
\noindent\textbf{Time cost.} {We analyze the extra time cost caused by block shuffling. A recent study of current graph algorithms \cite{graph_survey_vldb2021} shows that the graph index construction time complexity is always $O(|V|\log(|V|))$. The index construction involves high-dimensional vector distance calculation, which is very time-consuming. However, block shuffling only scans vertices and performs simple statistics, without any vector calculation. According to our evaluation, the block shuffling procedure introduces a relatively low additional time cost compared to the graph index construction process. For example, BNF only occupies 3\%$\sim$10\% of the graph index construction cost (see Fig. \ref{fig: index cost}(a)).}

\vspace{0.2em}
\noindent\textbf{Space cost.} {The disk-based graph index is stored as adjacency lists of vertices. Each vertex contains its vector data, neighbor count, and a list of neighbor IDs. Let $\Lambda$ and $D$ represent the maximum number of neighbor IDs and the vector dimensionality, respectively. The space complexity of the disk-based graph index is $O(|V| \cdot (D + \Lambda))$. In our experiments, we adjust $\Lambda$ to conform to the disk capacity constraint of a segment.
Note that the space cost of the disk-based graph index remains unchanged before and after block shuffling. This is because we only adjust the order of vertices and do not add any extra information.}

\vspace{0.2em}
\noindent\textbf{Remarks.} (1) Our block shuffling methods can work with any block size, not just the default 4KB for general disk. For example, we can extend to 8KB or 16KB blocks by modifying the block size. (2) {Block shuffling is similar to but not identical to graph partitioning. Graph partitioning has been extensively studied for real-world graphs (e.g., Social Networks \cite{pacaci2019experimental,WeiYLL16}) in graph engines \cite{Pregel,FlashGraph,Galois,Ligra}, where the out-degree follows a power-law distribution (neighbors tend to cluster together) \cite{abbas2018streaming}. However, our graph index is based on high-dimensional vectors, where neighbors exhibit similarity and navigation traits (with about 50\% long links \cite{graph_survey_vldb2021}) and the out-degree distribution is uniform (neighbors may scatter across clusters) \cite{graph_survey_vldb2021}, making locality more challenging. Moreover, current graph partitioning methods thrive on real-world graphs with clustering properties but may falter in graph index for vectors. We evaluated some advanced graph partitioning methods for our block shuffling task, but they only gave limited improvement. For example, BNF shows 40\% higher $OR(G)$ than an advanced graph partitioning method—KGGGP \cite{predari2016k} for graph index built on the SSNPP dataset. Our block shuffling strategies are tailored for graph index with long navigation links \cite{HNSW} and prove well-suited under our problem setting.} (3) Some recent works use some graph partitioning methods to reorder graph indexes \cite{coleman2021graph,jaiswal2022ood}, but they only achieve limited improvement. In contrast, we design the block shuffling algorithms based on block-level graph layout on the disk, which leads to better disk-based HVSS performance.

\subsection{In-Memory Navigation Graph}
\label{subsec: navigation_graph}
To reduce the search path length ($\ell$), many in-memory graph-based algorithms use an additional structure, such as multi-level Voronoi diagrams \cite{HVS}. The disk I/O cost depends directly on $\ell$ for searching on the disk-based graph index, so a shorter $\ell$ is even more crucial.

Our in-memory navigation graph reduces $\ell$ by finding better entry points for the disk-based graph search. To obtain such a graph, we employ two steps: (1) \textit{sample data points} and (2) \textit{build a graph index}. \underline{First}, we randomly sample some data points from all the data in a segment, based on the memory limit of a segment. \underline{Second}, we use the same algorithm as the disk-based graph (such as HNSW \cite{HNSW}, NSG \cite{NSG}) to build a navigation graph on the sampled data. The navigation graph is memory-resident and can quickly return the dynamic entry points that are close to a query.

\vspace{0.2em}
\noindent\textbf{Time cost.} {{\name} builds an in-memory graph only on a small data subset $V^{\prime}$, lowering the time complexity to $O(|V^{\prime}|\log(|V^{\prime}|))$, where $|V^{\prime}|$ is less than 10\% of the total number of vertices in a segment. This process is notably faster than constructing the graph on the entire dataset. In our evaluation, the in-memory graph construction demonstrates a low time cost (e.g., only 5.5\% of the total index processing time in Fig. \ref{fig: index cost}(a)).}

\vspace{0.2em}
\noindent\textbf{Space cost.} {The in-memory graph and the disk-based graph share the same storage format. Let $\Lambda^{\prime}$ and $D$ denote the maximum number of neighbor IDs and the vector dimensionality, respectively. The space complexity of the in-memory graph is $O(|V^{\prime}| \cdot (D + \Lambda^{\prime}))$. In our implementation, we adjust $|V^{\prime}|$ and $\Lambda^{\prime}$ to adhere to the memory limitation of a segment.}

%% file: sections/4_search.tex
\section{Search Strategy}
\label{sec: search}
The search strategy of {\name} consists of two components: (1) \textit{vertex search on the in-memory navigation graph} and (2) \textit{block search on the disk-resident graph}. The first component is designed to quickly navigate to the query's neighborhood without requiring disk access. It employs the same vertex search strategy as existing graph algorithms. The vertex search results serve as the entry points for the second component. To exploit the improved data locality, we utilize block search to explore not only the target vertex but also other vertices in the same block for each disk I/O. {Next, we will outline the fundamental block search strategy and its optimizations. Following that, we will show how we build the ANNS and RS algorithms based on the block search.}

\subsection{Block Search}
\label{subsec: basic_page_search}
\name{} {offline} assigns vertices and their neighbors to the same block to improve data locality. This way, one {online} block read from disk provides multiple relevant vertices. The block search strategy efficiently explores all the useful data in a block. It updates the current search results by calculating the distance from each vertex in the block to the query. It also checks the neighbor IDs of the closer vertices for new search candidates. \name{} further enhances this strategy with three computation-specific optimizations.

% \vspace{0.5em}
\vspace{0.3em}
\noindent\textbf{Block pruning.} We aim for a graph layout with $OR(G)=1$, meaning every vertex accessed by the block search is necessary. However, this ideal layout is very hard or impossible to achieve for a graph index built on high-dimensional vectors. Usually, our block shuffling results in $OR(G)$ ranging from 0.3 to 0.6 on many real-world datasets (Fig. \ref{fig: data_locality}(a)). This implies that some vertices in a block are irrelevant. To avoid exploring irrelevant data, {\name} prunes each loaded block. Specifically, it sorts the vertices in a block by their distance to the query vector in ascending order. Then, it only checks the neighbor IDs of the top-($(\varepsilon-1)\cdot \sigma$) vertices for new search candidates, where $(\varepsilon-1)$ is the number of vertices excluding the target vertex in a block and $\sigma$ is pruning ratio ($0<\sigma \leq 1$). This way, the neighbor IDs of distant vertices are discarded early. In our evaluation, $\sigma=0.3$ is always the optimal value for better performance.

% \vspace{0.5em}
\vspace{0.3em}
\noindent\textbf{I/O and computation pipeline.} In the block search stage, disk read (\textsf{DR}) and distance computation (\textsf{DC}) are two main operations. They occupy about 80\% of the block search time (see Fig. \ref{fig: block search optimization}(d)). A naive implementation is to execute \textsf{DR} and \textsf{DC} serially—that is—load a block by \textsf{DR} and then select the next hop by \textsf{DC} (as done in DiskANN \cite{DiskANN}). However, this schedule is inefficient because \textsf{DR} and \textsf{DC} are idle alternately. This wastes disk bandwidth and CPU. To avoid this, {\name} uses an I/O and computation pipeline that performs \textsf{DR} and \textsf{DC} concurrently (cf. \hyperref[alg: knn]{Algorithm 2}). Specifically, {\name} first executes \textsf{DC} for the target vertex $u$ in the current loaded block (lines 6–7). Then, it immediately conducts the next \textsf{DR}, while also performing \textsf{DC} of other vertices in the block that contains $u$ (line 11). Although each current \textsf{DR} decision does not consider the non-target vertices from the previous \textsf{DR}, it is acceptable as {\name} fully utilizes the disk bandwidth and CPU. We will show in {\S \ref{subsec: ablation_study}}, that I/O and computation pipeline reduces query latency significantly.

\vspace{0.3em}
\noindent\textbf{PQ-based approximate distance.} Recall that we need to load the full-precision vectors of all neighbors of the visited vertex to determine the next hop. Block shuffling mitigates such disk access by filling a block with many neighboring vertices. However, the number of neighbor IDs of each vertex is usually several times the number of vertices in a block, so many neighbors still require extra disk I/Os. To address this issue, we use approximate distance instead of exact distance with full-precision vectors. This is based on the observation that the next hop decision can be made by approximate computation with little accuracy damage \cite{chen2022finger}. Specifically, {\name} preprocesses the full dataset by PQ \cite{PQ} (like DiskANN \cite{DiskANN}), a popular compression method. It encodes the full-precision vectors into short codes that reside in the main memory. Indeed, {\name} efficiently obtains the approximate distances between the neighbors (whose full-precision vectors are not in memory) and the query by the short codes without disk I/Os to make the next disk read decision.

\vspace{0.3em}
\noindent\textbf{Time cost analysis.} {Let $\xi$ be the vertex utilization ratio and $\varepsilon$ be the number of vertices in a disk block. Then $\xi\cdot \varepsilon$ ($\geq 1$) represents the number of vertices accessed in each disk I/O. The vertex access complexity of the graph index is $O(o\cdot \ell)$ \cite{graph_survey_vldb2021}, where $\ell$ is the search path length and $o$ is the average out-degree of visited vertices. Therefore, the I/O time complexity of searching on the disk-based graph is $O((o\cdot \ell)/(\xi \cdot \varepsilon))$. In addition, the computation-specific block search optimizations further enhance search efficiency.}

\subsection{Approximate Nearest Neighbor Search}\label{subsec: knn}
The ANNS strategy in {\name} uses two ordered lists to store the candidates and search results: the candidate set and the result set. It starts the search by adding the entry points of the disk-based graph to the candidate set. These entry points are obtained on the in-memory navigation graph. Then it selects the unvisited vertex that is closest to the query from the candidate set and performs a block search. It updates both sets with the block search outcome. The procedure ends when all the vertices in the candidate set are visited. To improve efficiency, the strategy limits the size of the candidate set to prevent too many candidate vertices. The result set has no size limit because it is sorted only when the search terminates.

\setlength{\textfloatsep}{-5pt}
\begin{algorithm}[t]
\label{alg: knn}
  \caption{\textsc{ANNS}}
  \LinesNumbered
  \KwIn{in-memory navigation graph $G_m$, PQ short codes, disk-based graph $G_d$, query $q$, candidate set $C$ with fixed size, result set $R$, pruning ratio $\sigma$}
  \KwOut{top-$k$ results of $q$}
  
  $S$ $\gets$ entry points of $q$ by a vertex search strategy on $G_m$;

  $C$ $\gets$ $S$; \textcolor{blue}{\Comment{\textsf{sort by PQ distance to $q$}}}

  $R$ $\gets$ $S$; \textcolor{blue}{\Comment{\textsf{compute exact distance to $q$}}}

  $u$ $\gets$ top-1 unvisited vertex in $C$; \textcolor{blue}{\Comment{\textsf{target vertex}}}

  $B$ $\gets$ load the block including $u$ from $G_d$; \textcolor{blue}{\Comment{\textsf{DR}}}

  update $C$ according to $u$'s neighbor IDs; \textcolor{blue}{\Comment{\textsf{DC}}}

  add $u$ to $R$ based on $u$'s full-precision vector; \textcolor{blue}{\Comment{\textsf{DC}}}

  $B^{\prime}$ $\gets$ top-($(\varepsilon-1) \cdot \sigma$) vertices in $B \setminus \{u\}$; \textcolor{blue}{\Comment{\textsf{block pruning}}}
  
  \While{$C$ has unvisited vertex}{
    $v$ $\gets$ top-1 unvisited vertex in $C$; \textcolor{blue}{\Comment{\textsf{PQ-based routing}}}
    
    conduct line 5 for $v$, and lines 6–7 for the vertices in $B^{\prime}$ in parallel;
    
    \textcolor{blue}{\Comment{\textsf{I/O and computation pipeline}}}

    conduct lines 6–8 for $v$ and the block including it;
    
  }
  
  \textbf{return} top-$k$ vertices in $R$\textcolor{blue}{\Comment{\textsf{sort by exact distance to $q$}}}
\end{algorithm}
\setlength{\textfloatsep}{12pt plus 2pt minus 2pt}

{\name} performs ANNS for a query $q$ as follows (cf. \hyperref[alg: knn]{Algorithm 2}). (1) It finds the entry points $S$ of $q$ on the in-memory navigation graph $G_m$ (line 1). (2) It initializes a fixed-size candidate set $C$ and a result set $R$ with $S$ (lines 2–3). It sorts $C$ by approximate distance using PQ short codes. (3) It picks the top-1 unvisited vertex $u$ in $C$ as the target vertex and reads the block $B$ containing $u$ from the disk (lines 4–5). (4) It adds $u$'s neighbor IDs to $C$ and $u$ to $R$ (lines 6–7). (5) It prunes some vertices that are far from $q$ from $B$, adding the closer vertices into $B^{\prime}$ (line 8). (6) It executes (3) for the next top-1 unvisited vertex $v$ and (4) for the vertices in $B^{\prime}$ in parallel, then it executes (4)–(5) for $v$ (lines 10–12). (7) It terminates when $C$ has no unvisited vertices. Finally, it sorts $R$ by exact distance using full-precision vectors and returns top-$k$ results in $R$.

\subsection{Range Search}
\label{subsec: rs}
A range search (RS) query returns all the vectors within a search radius $r$. The result length depends on the vector distribution and can vary a lot among queries. For a dataset with tens of millions of vectors, the same $r$ may give zero to thousands of results. So, the RS algorithm should handle different result lengths. A simple RS strategy is to do ANNS repeatedly with different $k$ values to check the result length. However, this is inefficient because it will revisit the same vertices multiple times, causing extra computation and disk I/Os.

{\name} performs RS for a query $q$ based on block search. It uses a candidate set $C$, a result set $R$, and a kicked set $P$ to store candidate vertices, results, and vertices kicked out from $C$, respectively. {\name} changes the length limit of $C$ dynamically to handle different result lengths. The steps of RS are as follows. (1) It obtains the entry points from the in-memory navigation graph and initializes $C$ and $R$ with them. (2) It iteratively explores $C$ and updates $C$ and $R$ (like ANNS). It also adds the unvisited vertices kicked out from $C$ to $P$. (3) When all the vertices in $C$ are visited, it calculates the ratio of $R$'s length to $C$'s length. Given a ratio threshold $\varphi$, if
\begin{equation}
  \label{equ: rs ratio}
  \frac{|R|}{|C|}\geq \varphi \quad,
\end{equation}
it doubles the size of $C$ and restarts the search. This is because a high ratio means that most candidates are results. In this case, exploring more vertices may find new results. (4) In the next search, {\name} adds the closer vertices (w.r.t $q$) in $P$ to $C$ and repeats (2)–(3) for unvisited candidates in $C$. (5) The search stops when Eq. \ref{equ: rs ratio} is not met. After changing the length of $C$, it resumes the search with the previous results in $R$, candidates in $C$, and some closer vertices (w.r.t $q$) in $P$. This avoids extra computation and disk access. Our evaluation shows that $\varphi=0.5$ is optimal.

%% file: sections/5_exp.tex
\section{Experiments}\label{sec: exp}
We evaluate the following aspects of {\name}: (1) search performance (\textbf{\S \ref{subsec: intra-seg}}), (2) I/O-efficiency (\textbf{\S \ref{subsec: i/o efficiency}}), (3) index cost (\textbf{\S \ref{subsec: index cost}}), (4) ablation study (\textbf{\S \ref{subsec: ablation_study}}), (5) parameter sensitivity (\textbf{\S \ref{subsec: param_sensi}}), (6) scalability (\textbf{\S \ref{subsec: scalability}}), {(7) query distribution (\textbf{\S \ref{subsec: differnt query type}}), (8) segment setup (\textbf{\S \ref{subsec: segment setup}}), (9) large-scale search results (\textbf{\S \ref{subsec: big top-k result}}), and (10) billion-scale data (\textbf{\S \ref{subsec: billion-scale evaluation}}).} Our source code and datasets are available at: \url{https://github.com/zilliztech/starling}. Kindly refer to Appendix for additional evaluations.

\subsection{Experimental Setting}

\setlength{\textfloatsep}{0cm}
\setlength{\floatsep}{0cm}
\begin{table}[!th]
  \centering
  \setlength{\abovecaptionskip}{0.05cm}
  \setlength{\belowcaptionskip}{-0.3cm}
  \setstretch{0.8}
  \fontsize{6.5pt}{3.3mm}\selectfont
  \caption{Statistics of experimental datasets on a segment.}
  \label{tab: Dataset}
  \setlength{\tabcolsep}{.006\linewidth}{
  \begin{tabular}{|l|l|l|l|l|l|l|}
    \hline
    \textbf{Dataset} & \textbf{Data type} & \textbf{Dimensions} & \textbf{Distance function} & \textbf{\# Base vector per segment} & \textbf{\# Query} & \textbf{Query type}\\
    \hline
    BIGANN & uint8 & 128 & L2 & 33M & $10^4$ & ANNS/RS \\
    \hline
    DEEP & float & 96 & L2 & 11M & $10^4$ & ANNS/RS \\
    \hline
    SSNPP & uint8 & 256 & L2 & 16M & $10^5$ & RS \\
    \hline
    Text2image & float & 200 & IP & 5M & $10^5$ & ANNS \\
    \hline
  \end{tabular}
  }\vspace{0.2cm}
\end{table}

\noindent\textbf{Datasets.} We use four public real-world datasets \cite{bigann} that vary in data type, dimensions, distance, and query type. Tab. \ref{tab: Dataset} shows their details. Unless specified, we limit the raw vectors per segment to under 4GB and adjust the data scale accordingly for each dataset. {We randomly select vectors from standard datasets for a segment. We perform a brute-force search on the selected vectors to get the ground truth. The query sets are the same for a given dataset.}

\vspace{0.2em}
\noindent\textbf{Query workloads.} {The query workloads' details are provided in Tab. \ref{tab: Dataset}. By default, all query sets are derived from real-world scenarios and are not-in-database, meaning they do not have any intersection with the base data. The queries are executed in a random order, and a batch of queries is served using a pool of threads. Each thread is assigned to handle one query at a time.}

\vspace{0.2em}
\noindent\textbf{Compared methods.} We compare {\name} with two state-of-the-art disk-based HVSS methods that are able to {process up to 4GB vector data within the 2GB memory constraint}. We do not include in-memory methods (e.g., HNSW \cite{HNSW} and IVFPQ \cite{PQ}) in the evaluation, because they either run out of memory or have very low accuracy, as reported in previous works \cite{DiskANN,SPANN,HM_ANN}.

\begin{itemize}[leftmargin=*]
  \item \textbf{DiskANN} \cite{DiskANN} is the state-of-the-art disk-based graph index method, which we use as the baseline framework.
  \item \textbf{SPANN} \cite{SPANN} is a disk-based inverted index method that achieves better performance than DiskANN but requires more disk space.
  \item \boldsymbol{\name} is our framework, which implements all the optimizations proposed in this paper. Unless otherwise specified, we use the Vamana algorithms as the default option in {\name}, denoted as {\name}-Vamana or simply {\name}.
\end{itemize}

\noindent\textbf{Evaluation protocol.} We use \textit{queries per second (QPS)}, \textit{mean latency}, and \textit{mean I/Os} to evaluate the search performance of all competitors. By default, we use eight threads on the server to serve queries and report the \textit{QPS} based on the wall clock time from the query input to the result output. For the latency and I/Os, we record the response time and number of disk I/Os for each query and compute the \textit{mean latency} and \textit{mean I/Os} over all queries. For ANNS, we measure its accuracy by $Recall$ (Eq. \ref{equ: recall}). Unless otherwise stated, we set $k=10$ in \textit{Recall}. For RS, $AP$ (Eq. \ref{equ: ap}) is a more suitable accuracy metric. We fix a search radius $r$ for each dataset following \cite{neurips21_competition_report}. We evaluate two popular similarity metrics, $L2$ and inner product ($IP$).

\vspace{0.2em}
\noindent\textbf{Segment configuration.} We follow the configuration of a mainstream open-source vector database, Milvus \cite{milvus_archi, Manu_zilliz}, with at most 2GB memory space and 10GB disk capacity for each segment by default.

\vspace{0.2em}
\noindent\textbf{Setup.} We execute the C++ codes of all methods using different instances for both index building and search to align with the protocols of current vector databases \cite{Manu_zilliz} and related research \cite{graph_survey_vldb2021,NSG}. For index building, we employ a n2-standard-64 instance (ubuntu-2004-focal-v20221121, 2 TB SSD Persistent Disk) with 64 vCPUs—this facilitates fast construction and segment sharing. In terms of search functionality, vector databases typically leverage multiple smaller instances to improve query performance and promote fine-grained load balancing and scaling. In our experiments, all segments share a n2-standard-8 instance (ubuntu-2004-focal-v20221121, 375 GB NVMe Local SSD Scratch Disk) with 8 vCPUs. {We use the \textsf{o\_direct} option to read data from the disk for all competitors, circumventing operating system caching.} We optimize all approaches' hyper-parameters for the segment's configuration. We report the average results of three trials on the optimum configuration.

\subsection{Search Performance} \label{subsec: intra-seg}

\setlength{\textfloatsep}{0cm}
\setlength{\floatsep}{0cm}
\begin{figure*}[!th]
\setlength{\abovecaptionskip}{0cm}
\setstretch{0.9}
\fontsize{8pt}{4mm}\selectfont
\begin{minipage}{0.575\textwidth}
  \setlength{\abovecaptionskip}{0cm}
  \setlength{\belowcaptionskip}{0.1cm}
  \centering
  \footnotesize
  \stackunder[0.8pt]{\includegraphics[scale=0.14]{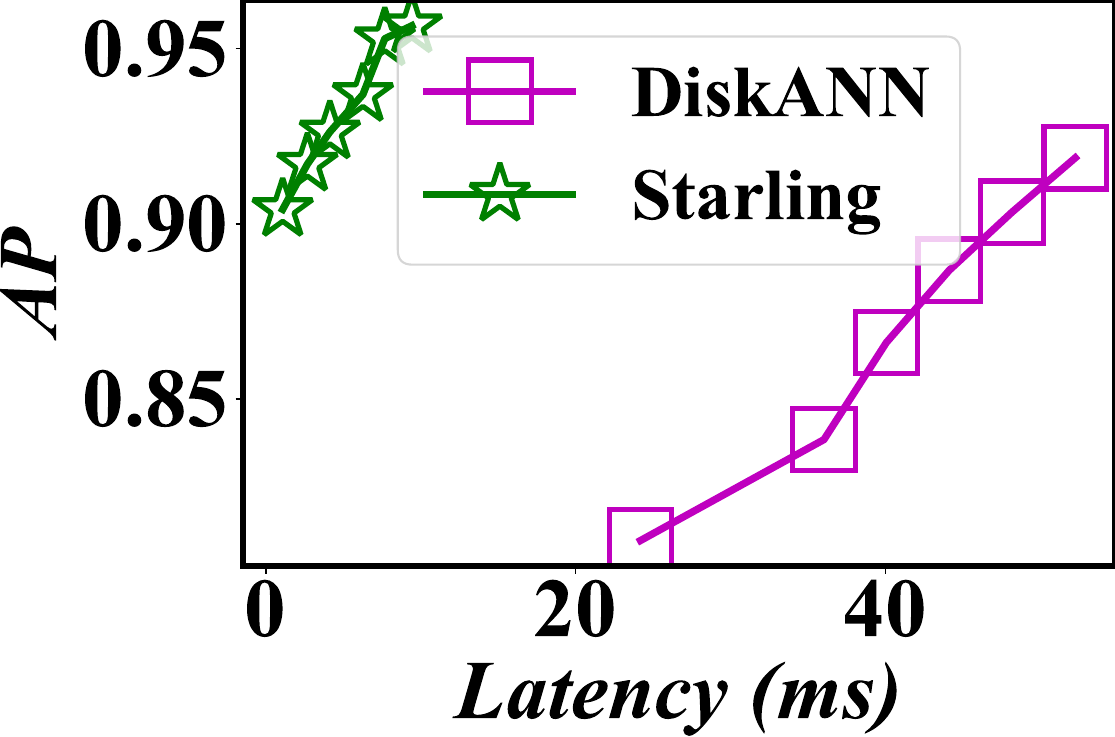}}{(a) BIGANN}
  \hspace{-0.1cm}
  \stackunder[0.75pt]{\includegraphics[scale=0.14]{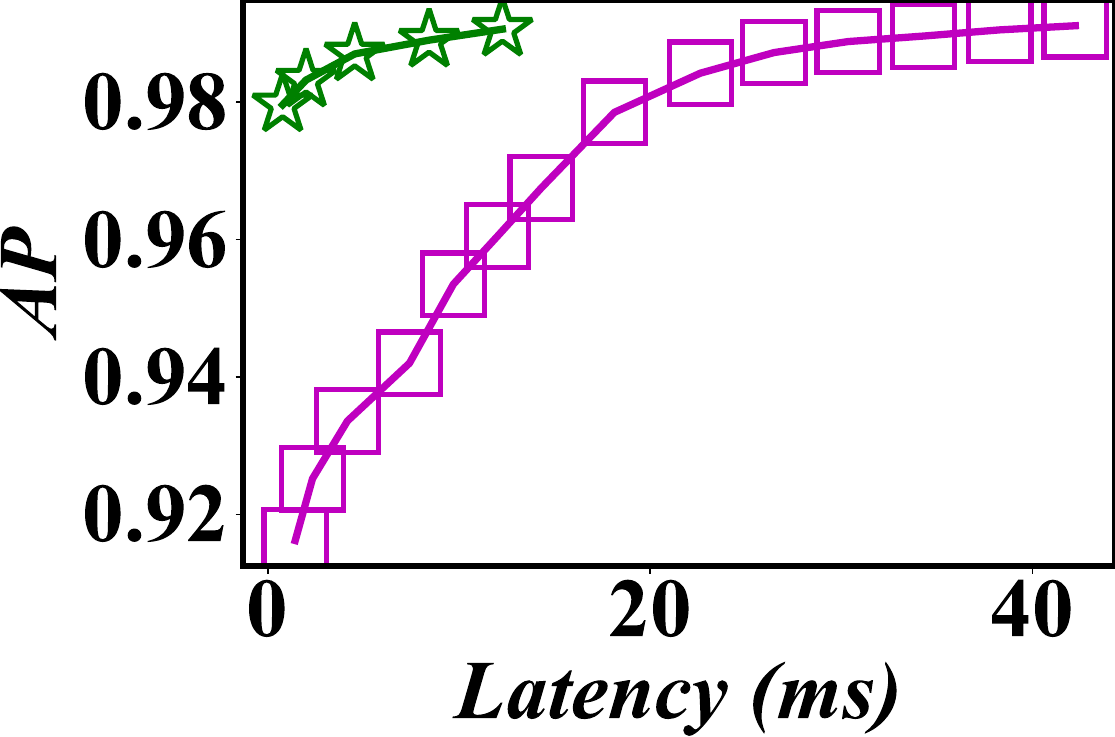}}{(b) DEEP}
  \hspace{-0.1cm}
  \stackunder[0.75pt]{\includegraphics[scale=0.14]{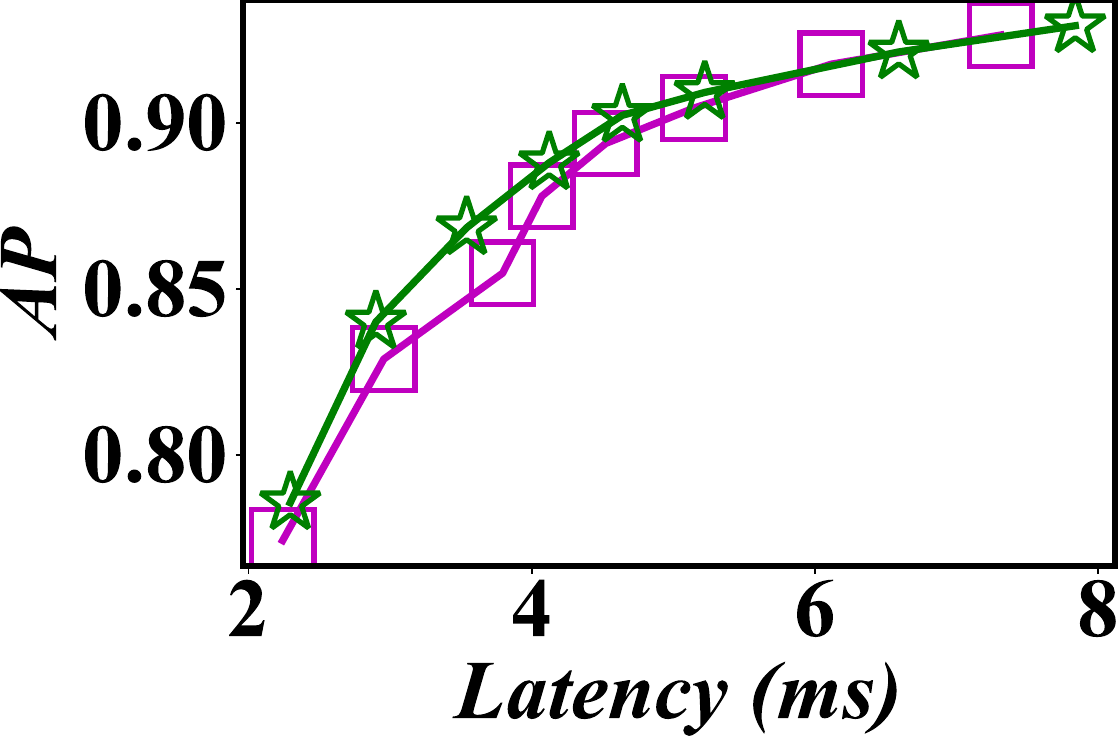}}{(a) SSNPP}
  \hspace{-0.1cm}
  \caption{RS performance (\textit{AP} vs \textit{Latency}).}
  \label{fig: latency_rs}
\end{minipage}
\begin{minipage}{0.42\textwidth}
  \setlength{\abovecaptionskip}{0cm}
  \setlength{\belowcaptionskip}{0.1cm}
  \centering
  \footnotesize
  \stackunder[0.75pt]{\includegraphics[scale=0.14]{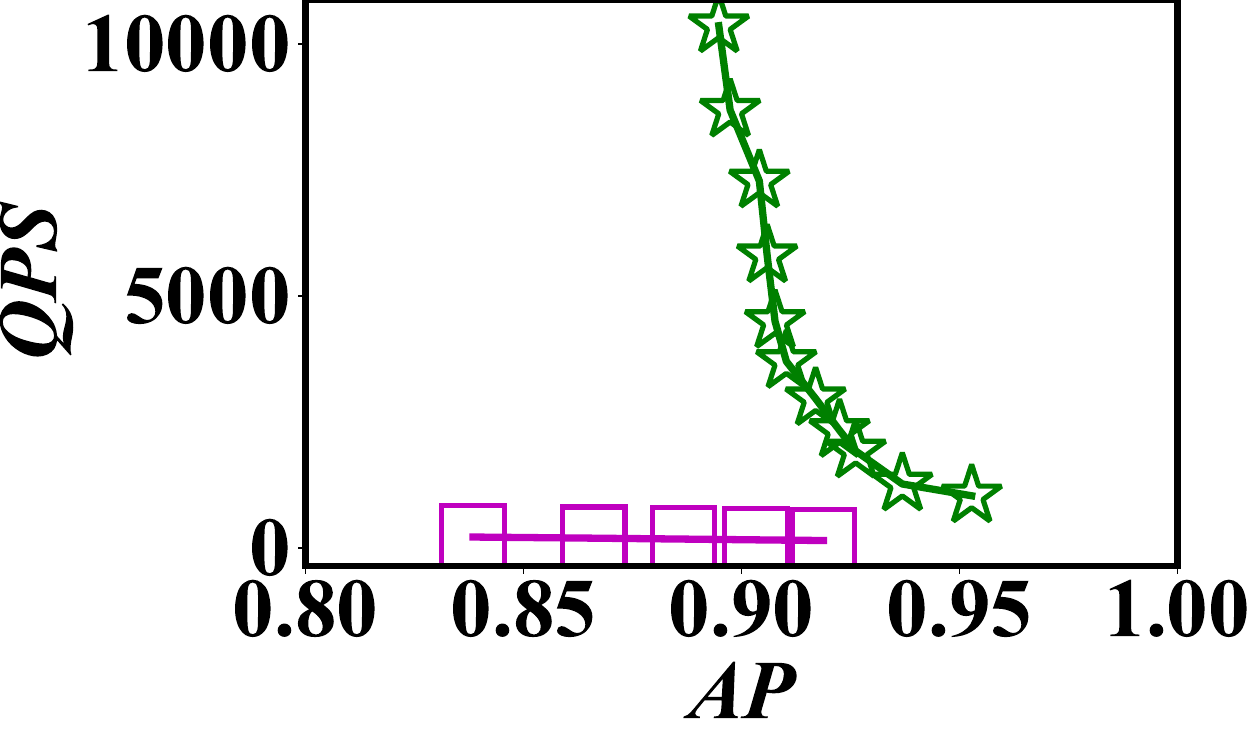}}{(a) BIGANN}
  \hspace{-0.2cm}
  \stackunder[0.75pt]{\includegraphics[scale=0.14]{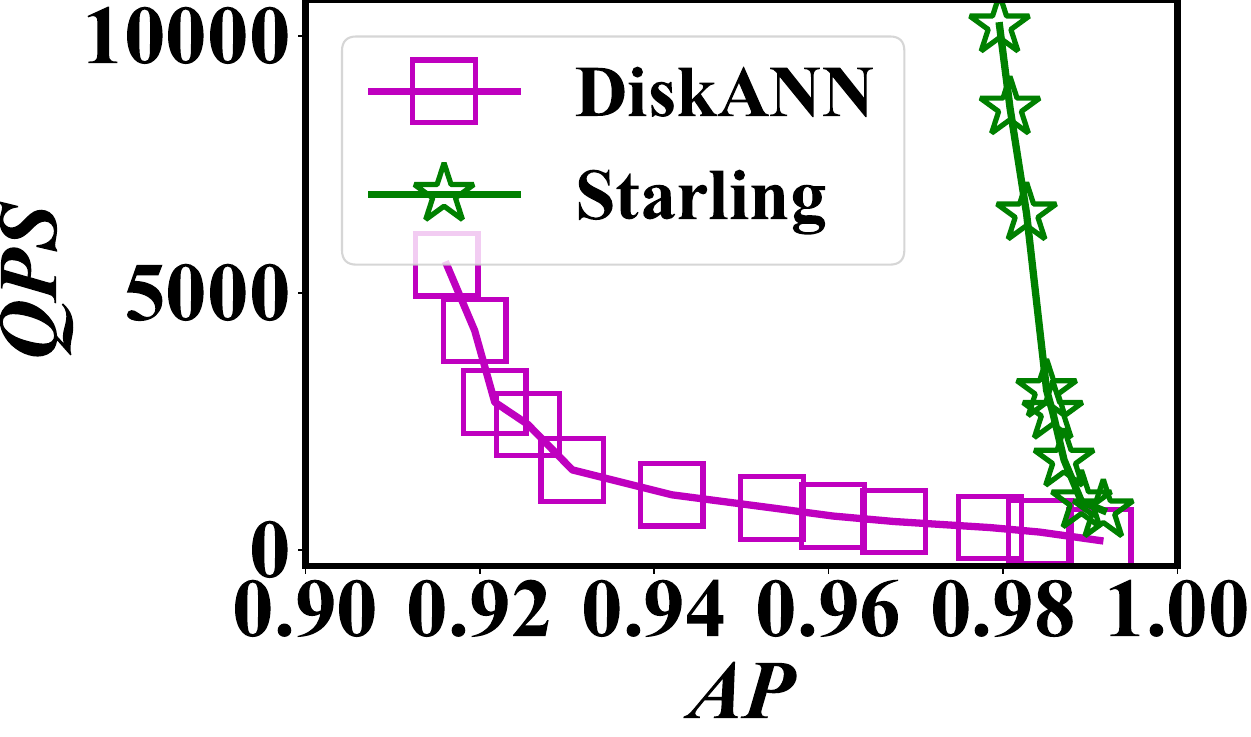}}{(b) DEEP}
  \hspace{-0.2cm}
  \caption{RS performance (\textit{QPS} vs \textit{AP}).}
  \label{fig: qps_rs}
\end{minipage}
% \vspace{-0.4cm}
\end{figure*}

\noindent\textbf{RS performance.} Current disk-based HVSS methods largely ignore RS, except for a baseline in a competition \cite{bigann}. RS support is provided by calling ANNS iteratively on DiskANN (referred to as DiskANN) \cite{DiskANN_code}. Fig. \ref{fig: latency_rs} and \ref{fig: qps_rs} compare the RS performance of DiskANN and {\name}. Under the same \textit{AP}, {\name} reduces query latency by up to 98\% compared to DiskANN. On SSNPP, {\name} exhibits a slight performance boost. We observed that most query results are near the centroid for that dataset. DiskANN's cache policy loads the data near the centroid, covering most query results. However, {\name} still achieves better RS performance. Fig. \ref{fig: qps_rs} shows that {\name} has a significantly higher throughput than DiskANN under the same \textit{AP}. For example, it is 43.9$\times$ faster than DiskANN when $AP=0.9$ on BIGANN. Further investigation reveals that DiskANN requires numerous disk I/Os for some queries with long RS results (e.g., 1,000). In that case, {\name} achieves greater gains as it avoids more redundant disk I/Os.

\setlength{\textfloatsep}{0cm}
\setlength{\floatsep}{0cm}
\begin{figure*}[!th]
\setlength{\abovecaptionskip}{0cm}
\setstretch{0.9}
\fontsize{8pt}{4mm}\selectfont
\begin{minipage}{0.575\textwidth}
  \setlength{\abovecaptionskip}{0cm}
  \setlength{\belowcaptionskip}{0.1cm}
  \centering
  \footnotesize

  \stackunder[0.8pt]{\includegraphics[scale=0.14]{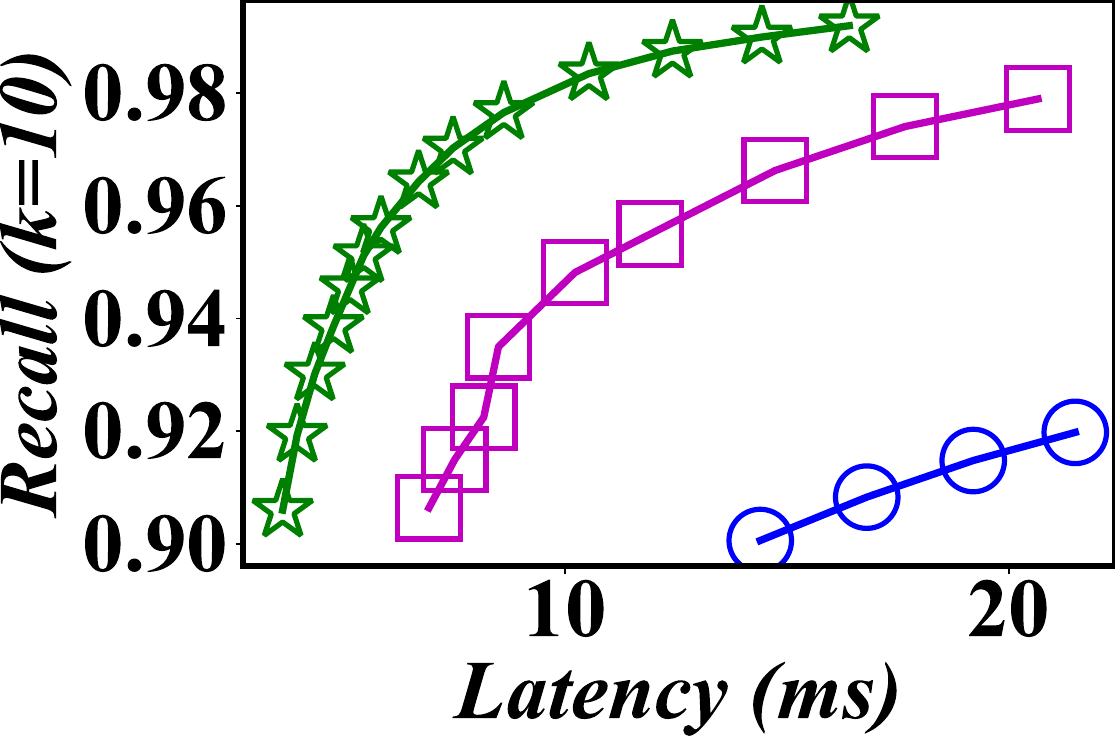}}{(a) BIGANN}
  \hspace{-0.1cm}
  \stackunder[0.75pt]{\includegraphics[scale=0.14]{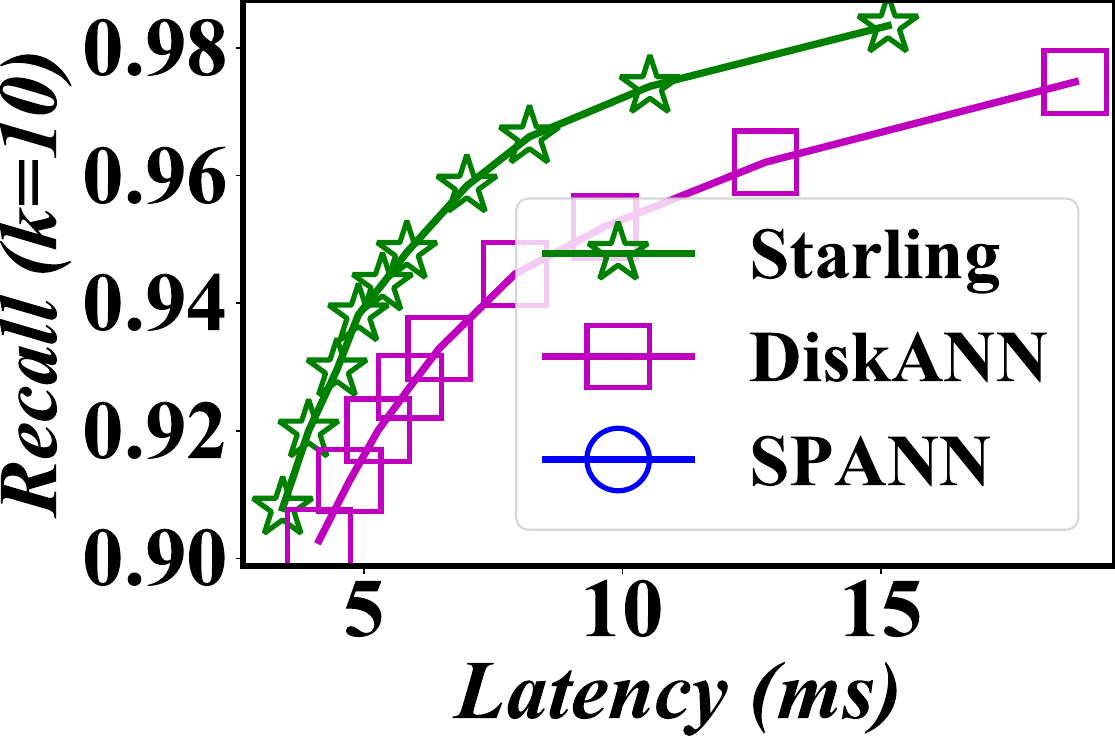}}{(b) Text2image}
  \hspace{-0.1cm}
  \stackunder[0.75pt]{\includegraphics[scale=0.14]{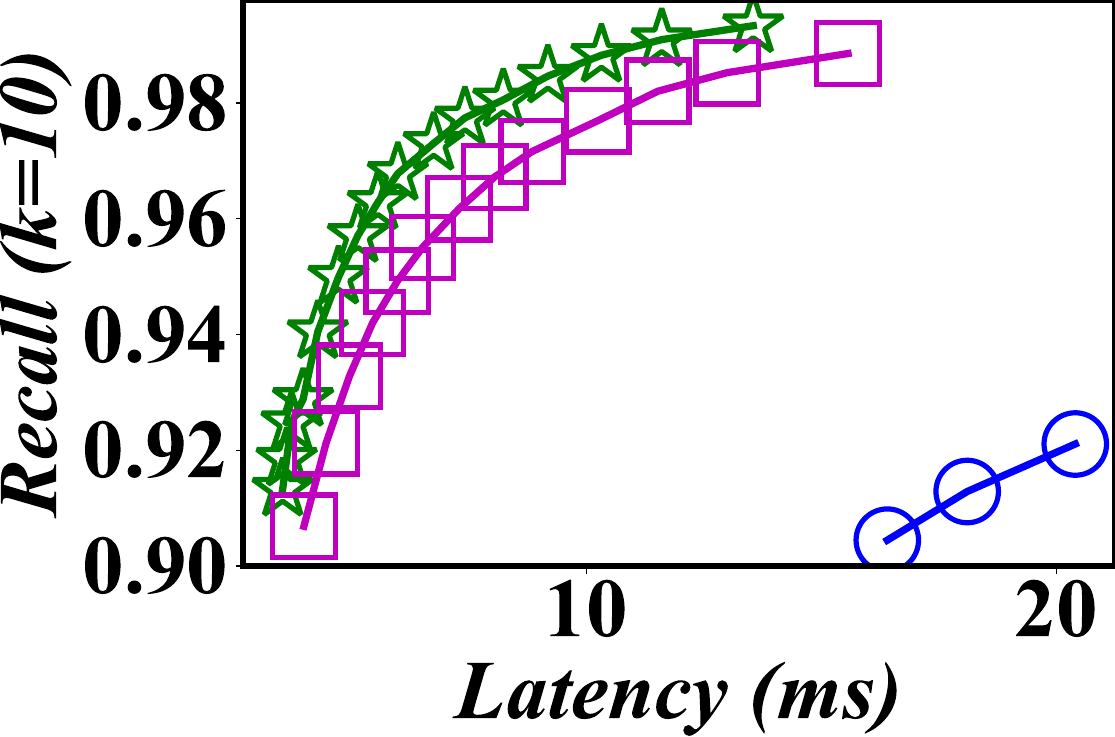}}{(a) DEEP}
  \hspace{-0.1cm}
  \caption{ANNS performance (\textit{Recall} vs \textit{Latency}).}
  \label{fig: intra-segment_knn}
\end{minipage}
\begin{minipage}{0.42\textwidth}
  \setlength{\abovecaptionskip}{0cm}
  \setlength{\belowcaptionskip}{0.1cm}
  \centering
  \footnotesize
  \stackunder[0.75pt]{\includegraphics[scale=0.14]{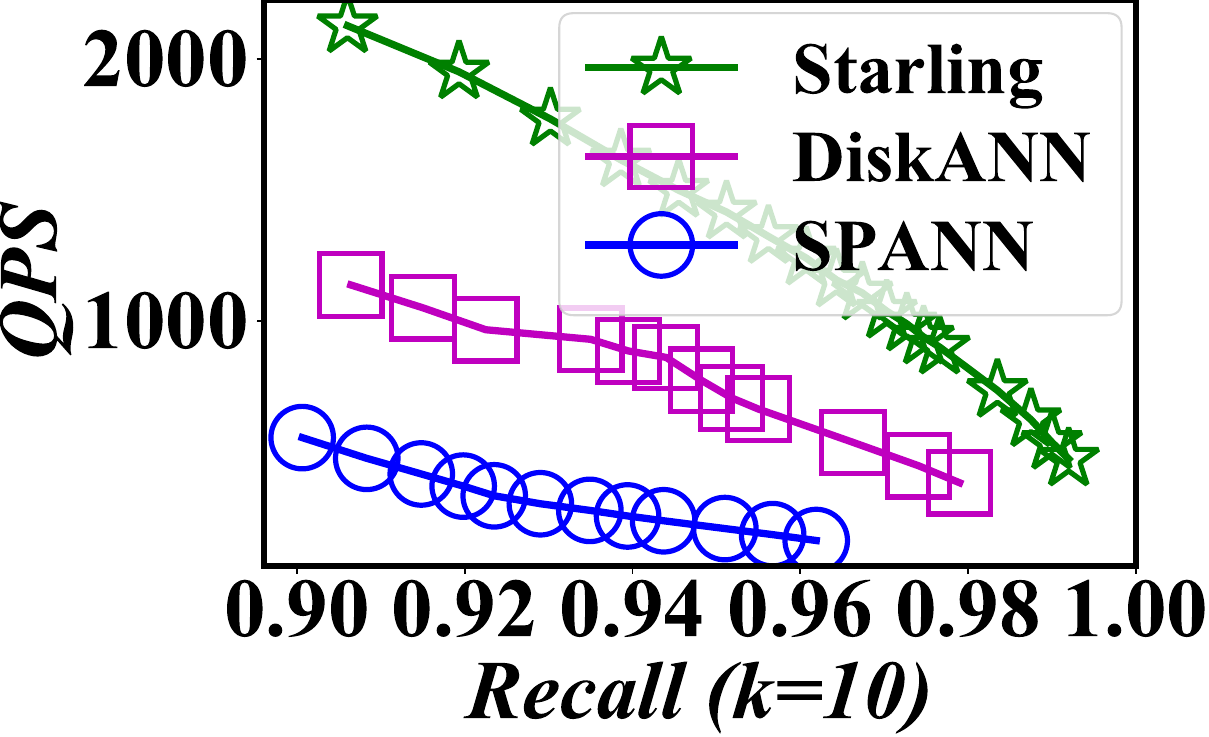}}{(a) BIGANN}
  \hspace{-0.2cm}
  \stackunder[0.75pt]{\includegraphics[scale=0.14]{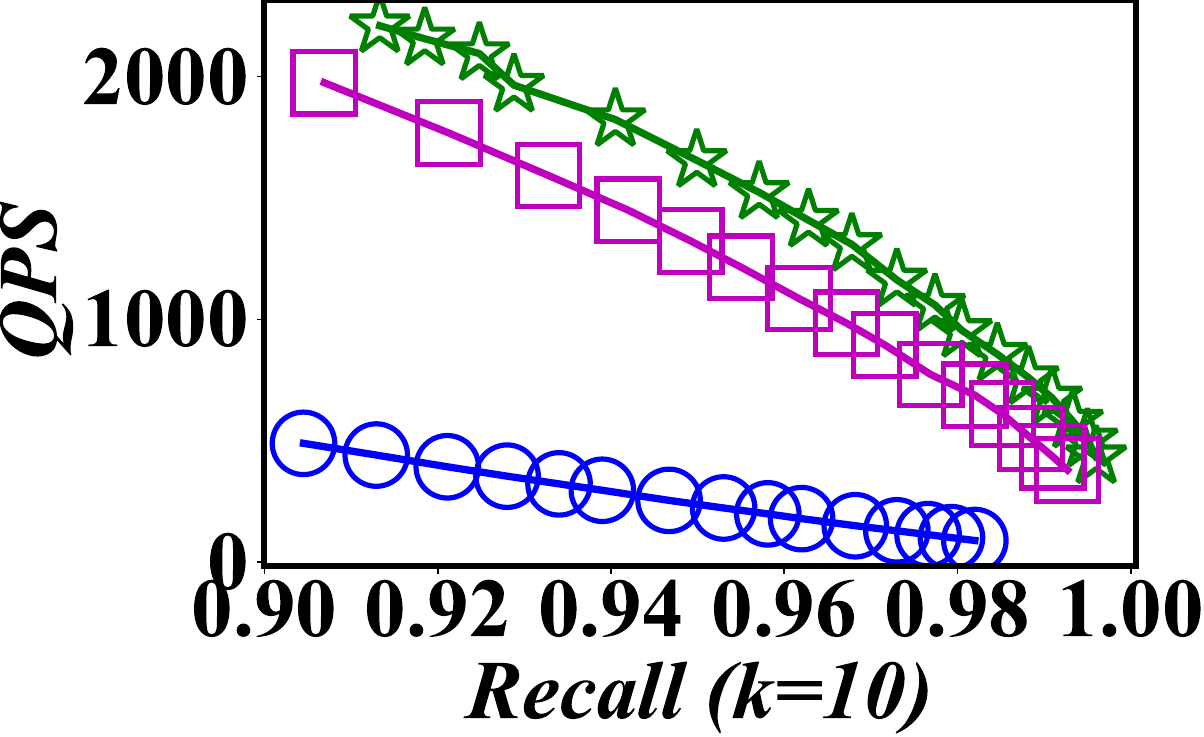}}{(b) DEEP}
  \hspace{-0.2cm}
  \caption{ANNS performance (\textit{QPS} vs \textit{Recall}).}
  \label{fig: qps_knn}
\end{minipage}
\vspace{-0.4cm}
\end{figure*}

\vspace{0.2em}
\noindent\textbf{ANNS performance.} Fig. \ref{fig: intra-segment_knn} contrasts the \textit{Recall} and \textit{Latency} of assorted methods. SPANN on Text2image is omitted, as its latency exceeds 20ms when $Recall > 0.9$. {\name} consistently delivers lower latency compared to two rivals under similar \textit{Recall} conditions. For example, at $Recall=0.95$ on BIGANN, {\name} (5ms) achieves over 2$\times$ more speed than DiskANN (10ms) and 10$\times$ more than SPANN (50ms). Fig. \ref{fig: qps_knn} delineates the \textit{QPS} and \textit{Recall} of various methods. {\name} continually exhibits superior performance over its competitors. SPANN presents lackluster results, worse than anticipated. The reason is that SPANN relies heavily on data duplication (up to eight-fold the base data size \cite{SPANN}). Given the segment configuration, the number of data replications is restrained, leading to a performance dip (for additional analysis, see \S \ref{subsec: segment setup}).

\subsection{I/O-efficiency} \label{subsec: i/o efficiency}

\setlength{\textfloatsep}{0cm}
\setlength{\floatsep}{0cm}
\begin{table}[!h]
  \centering
  \setlength{\abovecaptionskip}{0.05cm}
  \setlength{\belowcaptionskip}{0.1cm}
  \setstretch{0.8}
  \fontsize{6.5pt}{3.3mm}\selectfont
  \caption{Vertex utilization ratio ($\xi$) and search path length ($\ell$).}
  \label{tab: vertex utilization and search path}
  \setlength{\tabcolsep}{.02\linewidth}{
  \begin{tabular}{|l|c|l|l|l|l|}
    \hline
     \textbf{Framework}$\downarrow$ & \textbf{Metric}$\downarrow$ & \textbf{BIGANN} & \textbf{DEEP} & \textbf{SSNPP} & \textbf{Text2image} \\
    \hline
    DiskANN & \multirow{2}*{$\xi$} & 0.0625 & 0.1429 & 0.1111 & 0.2500 \\
    \cline{1-1}
    \cline{3-6}
    {\name} & ~ & \textbf{0.3438} & \textbf{0.4429} & \textbf{0.4111} & \textbf{0.8760} \\
    \hline
    DiskANN & \multirow{2}*{$\ell$} & 362 & 341 & 127 & 269 \\
    \cline{1-1}
    \cline{3-6}
    {\name} & ~ & \textbf{182} & \textbf{240} & \textbf{100} & \textbf{167} \\
    \hline
  \end{tabular}
  }\vspace{-0.2cm}
\end{table}

\noindent\textbf{Vertex utilization ratio ($\xi$).} We calculate $\xi$ for each block loaded in the queries and average these to obtain the overall $\xi$ for the loaded blocks. As displayed in Tab. \ref{tab: vertex utilization and search path}, {\name} consistently outperforms DiskANN across all datasets. This superior performance is credited to {\name}'s data locality enhancement through block shuffling, which permits a single block from the disk to encompass multiple relevant vertices. This approach heightens $\xi$, thereby mitigating disk bandwidth wastage. Hence, {\name} exhibits higher I/O-efficiency. {Nonetheless, $\xi$ is not the exclusive factor that influences query performance–other factors such as data distribution and optimizations like the in-memory navigation graph also impact overall performance, potentially causing a non-linear relationship between $\xi$ and query performance metrics such as latency.}

\vspace{0.2em}
\noindent\textbf{Search path length ($\ell$).} We undertake an evaluation of $\ell$ for all queries, subsequently reporting the average $\ell$ under specified search accuracy. As denoted in Tab. \ref{tab: vertex utilization and search path}, {\name} exhibits a shorter $\ell$ compared to DiskANN. Within {\name}, the in-memory navigation graph yields query-aware dynamic entry points that are in closer proximity to the query, thus reducing $\ell$. Consequently, the IO-efficiency of {\name} sees further improvement.

\subsection{Index Cost}
\label{subsec: index cost}

\begin{figure}[!th]
  % \vspace{0.1cm}
  \setlength{\abovecaptionskip}{0.1cm}
  \setlength{\belowcaptionskip}{0.2cm}
  \centering
  \footnotesize
  \stackunder[0.7pt]{\includegraphics[scale=0.65]{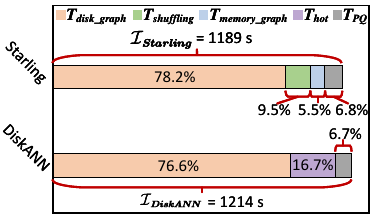}}{(a) Index processing time}
  \hspace{0cm}
  \stackunder[0.7pt]{\includegraphics[scale=0.65]{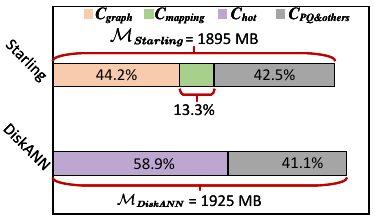}}{(b) Memory cost}
  \caption{Index cost of two different frameworks on BIGANN.}
  \label{fig: index cost}
  \vspace{-0.5cm}
\end{figure}

\noindent\textbf{Index processing time.} The offline index process for {\name} encompasses four components: disk-based graph construction time $T_{disk\_graph}$, block shuffling time $T_{shuffling}$, in-memory graph construction time $T_{memory\_graph}$, and preprocessing time for PQ short codes $T_{PQ}$ (cf. \S\ref{subsec: basic_page_search}) Consequently, the aggregate offline index procession time for {\name} can be calculated as:
\begin{equation}
  \label{equ: offline index cost}
  \mathcal{I}_{Starling} = T_{disk\_graph} + T_{shuffling} + T_{memory\_graph} + T_{PQ} \quad.
\end{equation}
The offline index process for DiskANN comprises three parts: disk-based graph construction time $T_{disk\_graph}$, hot vertices acquisition time $T_{hot}$, and preprocessing time for PQ short codes $T_{PQ}$. As per this structure, the total offline index processing for DiskANN is:
\begin{equation}
  \label{equ: offline index cost}
  \mathcal{I}_{DiskANN} = T_{disk\_graph} + T_{hot} + T_{PQ} \quad.
\end{equation}
Fig. \ref{fig: index cost}(a) exposes the breakdown of index processing time on BIGANN. Despite the additional shuffling and construction time, {\name} is substantially more efficient compared to DiskANN since the summation of $T_{shuffling}$ and $T_{memory\_graph}$ in {\name} is less than $T_{hot}$ in DiskANN. DiskANN necessitates the generation of hot vertices which involves the sampling of a large pool of queries and executing a slow disk-based graph search to tally vertex visit frequency. This procedure can be extremely time-consuming without hot vertices in memory. In contrast, {\name}'s approach of building an in-memory navigation graph proves to be faster and more effective. It is also worth noting that both {\name} and DiskANN have an equal $T_{disk\_graph}$ and $T_{PQ}$.

\vspace{0.3em}
\noindent\textbf{Memory cost.} {In the main memory, {\name} manages an in-memory navigation graph ($\mathcal{C}_{graph}$), the mapping of vertex IDs to block IDs ($\mathcal{C}_{mapping}$), and PQ short codes along with other data structures ($\mathcal{C}_{PQ\&others}$). Hence, the total memory cost of {\name} is
\begin{equation}
  \label{equ: starling memory cost}
  \mathcal{M}_{Starling} = \mathcal{C}_{graph} + \mathcal{C}_{mapping} + \mathcal{C}_{PQ\&others} \quad.
\end{equation}
DiskANN, on the other hand, houses a set of hot vertices ($\mathcal{C}_{hot}$) and the PQ short codes along with other data structures ($\mathcal{C}_{PQ\&others}$). The memory cost is hence:
\begin{equation}
  \label{equ: memory cost}
  \mathcal{M}_{DiskANN} = \mathcal{C}_{hot} + \mathcal{C}_{PQ\&others} \quad.
\end{equation}
DiskANN does not maintain $C_{mapping}$ as a group of ID-contiguous vertices are allocated into a block (it locates a block by vertex ID).}

Fig. \ref{fig: index cost}(b) offers a comparison of the memory costs of both frameworks when handling the same set of queries with the same accuracy on BIGANN. The results reveal {\name} has lower memory overhead compared to DiskANN. The reason for this is as follows. {\name} requires an extra mapping of vertex IDs to block IDs post block shuffling as vertex IDs within a block are non-consecutive. In DiskANN, each hot vertex comprises vector data and the neighbor IDs of the disk-based index. {\name}'s in-memory graph also includes vector data and neighbor IDs but it contains 20\% fewer neighbor IDs than the disk-based graph since it possesses 10\% less vector data. Resultingly, $\mathcal{C}_{graph}+\mathcal{C}_{mapping}$ in {\name} is less than $\mathcal{C}_{hot}$ in DiskANN.

\vspace{0.2em}
\noindent\textbf{Disk cost.} {\name} and DiskANN undergo the same disk cost as they utilize the same disk-based graph, albeit in different layouts.

\subsection{Ablation Study} \label{subsec: ablation_study}

\setlength{\textfloatsep}{0cm}
\setlength{\floatsep}{0cm}
\begin{figure*}[!th]
\setlength{\abovecaptionskip}{0cm}
\setlength{\belowcaptionskip}{-0.3cm}
\setstretch{0.9}
\fontsize{8pt}{4mm}\selectfont
\begin{minipage}{0.525\textwidth}
  \setlength{\abovecaptionskip}{0.1cm}
  \setlength{\belowcaptionskip}{0cm}
  \centering
  \footnotesize
  \stackunder[0.9pt]{\includegraphics[scale=0.161]{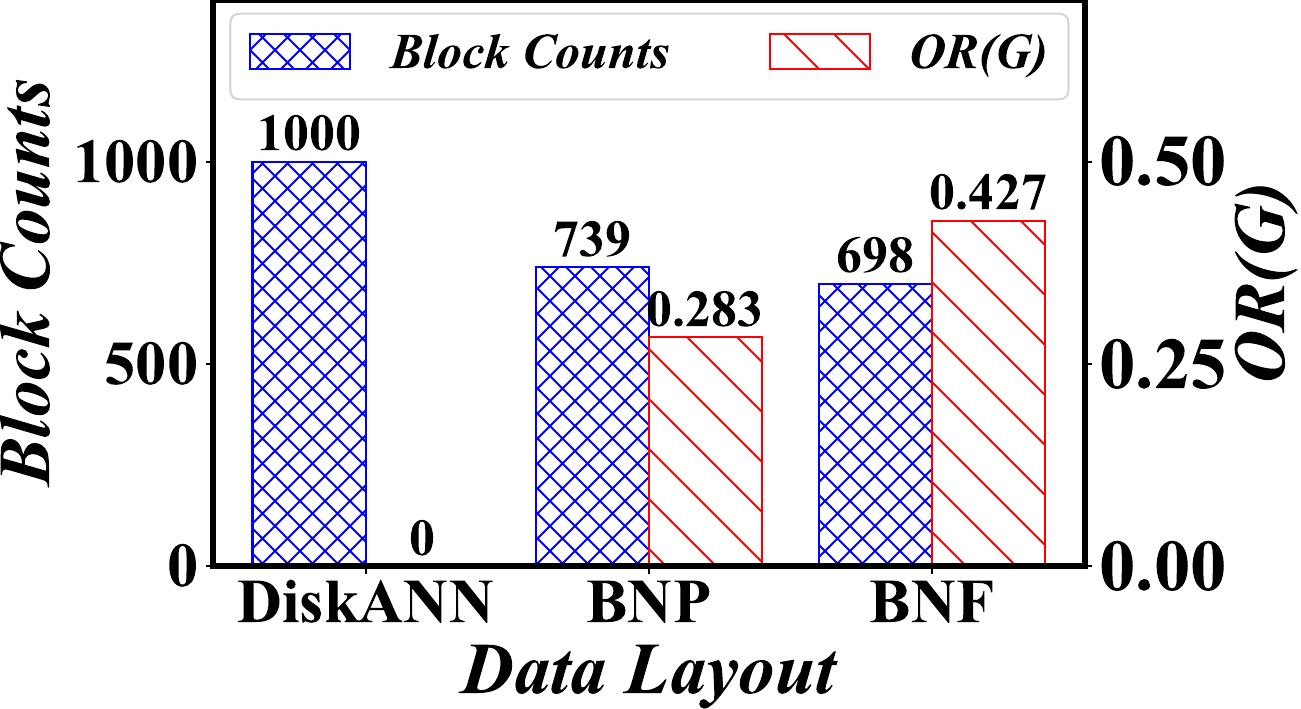}}{(a) Block shuffling}
  \hspace{0.1cm}
  \stackunder[0.6pt]{\includegraphics[scale=0.161]{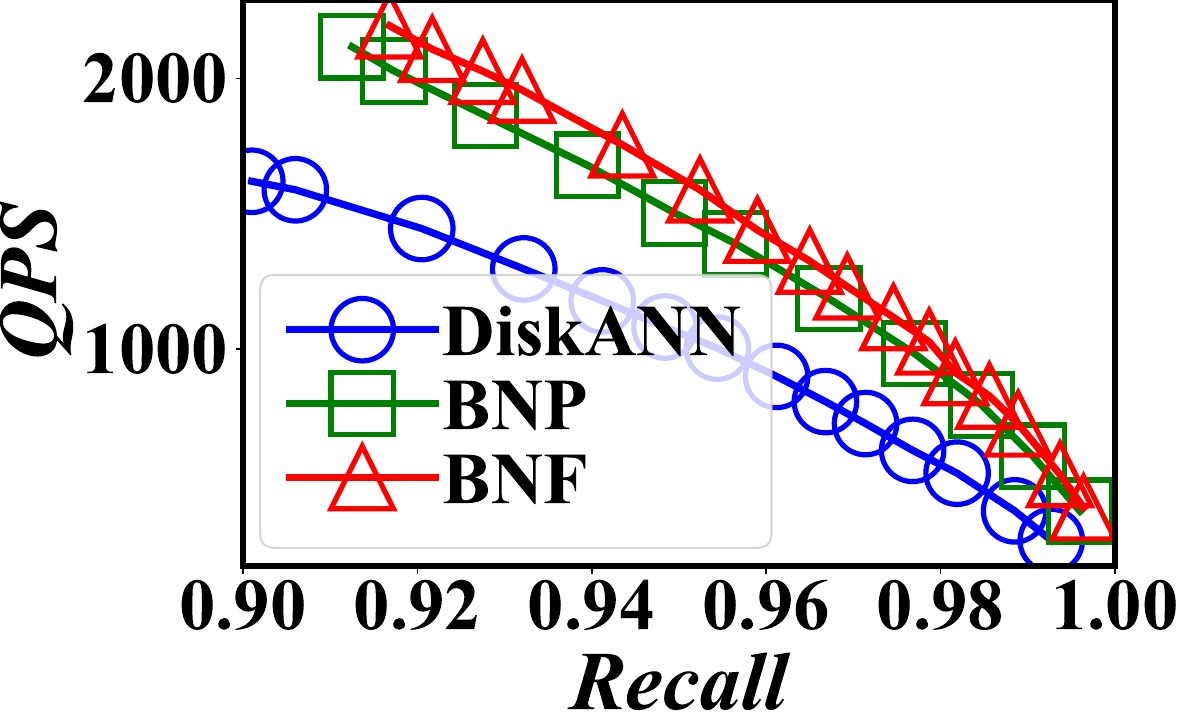}}{(b) Search performance}
  \caption{Effect of block shuffling.}
  \label{fig: data_locality}
\end{minipage}
\begin{minipage}{0.47\textwidth}
  \setlength{\abovecaptionskip}{0.1cm}
  \setlength{\belowcaptionskip}{0cm}
  \centering
  \footnotesize
  \stackunder[0.75pt]{\includegraphics[scale=0.161]{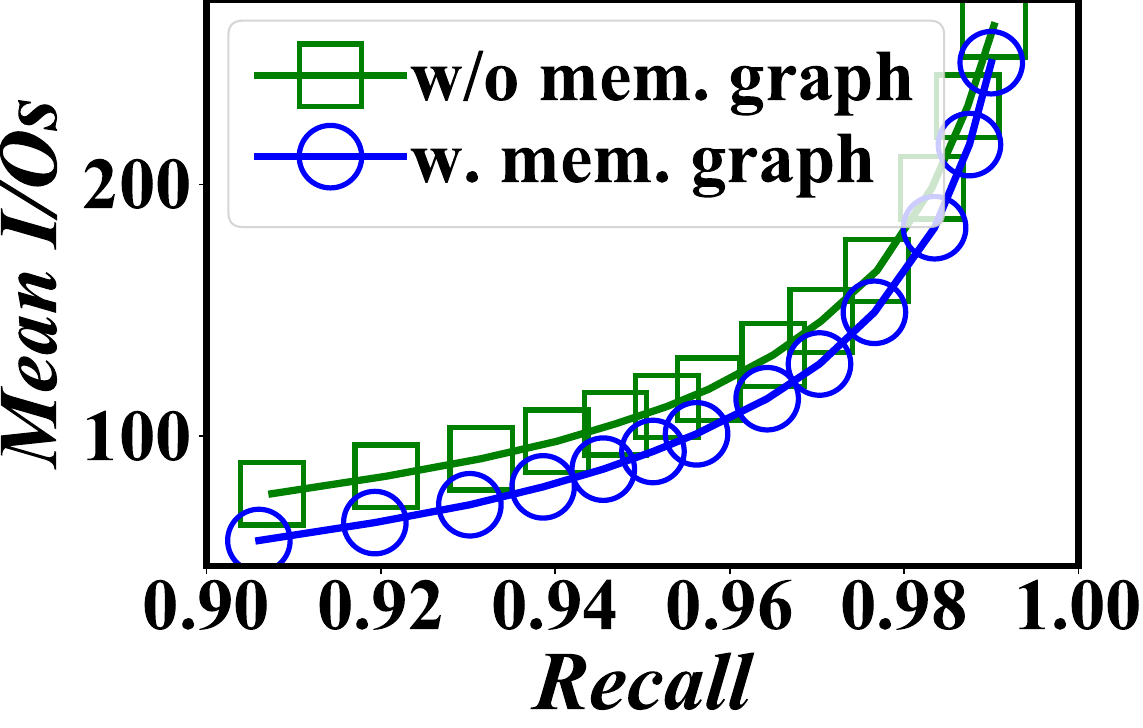}}{(a) Disk I/Os}
  \hspace{-0.15cm}
  \stackunder[0.75pt]{\includegraphics[scale=0.161]{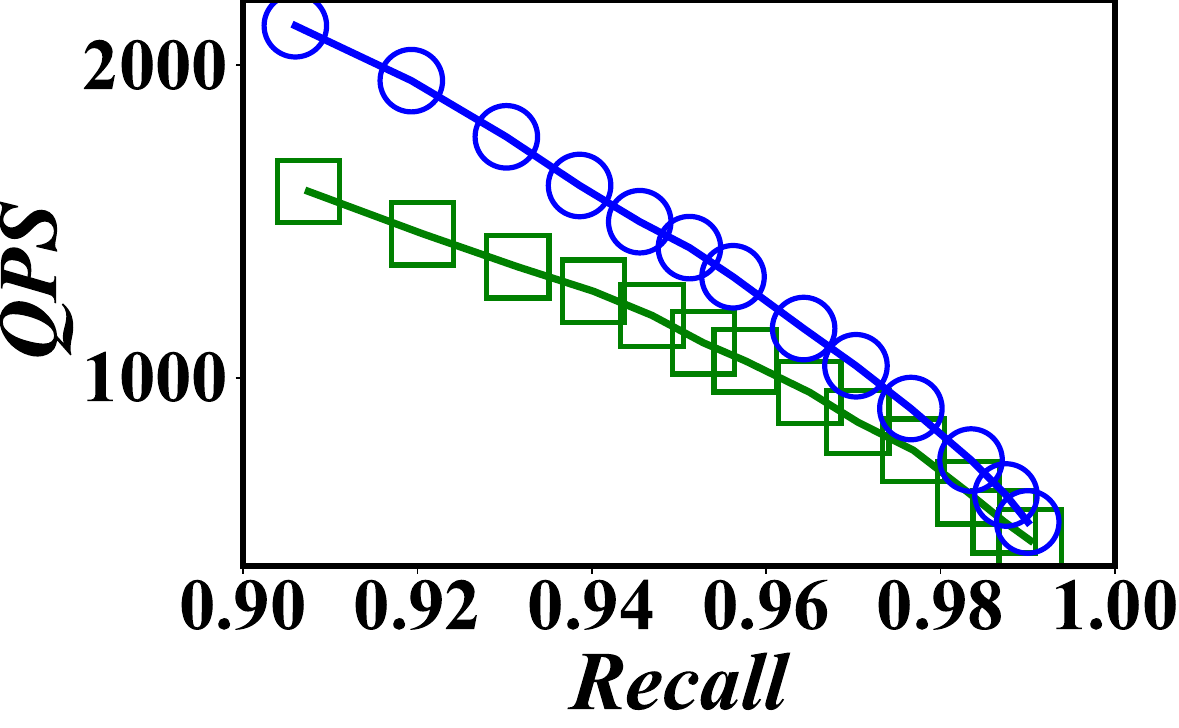}}{(b) Throughput}
  \caption{Effect of in-memory graph.}
  \label{fig: navigation_graph}
\end{minipage}
\vspace{-0.3cm}
\end{figure*}

\noindent\textbf{Block shuffling.} Fig. \ref{fig: data_locality}(a) reports the average number of blocks containing the top-1,000 nearest neighbors for each query within a given set (blue bar), as well as the overlap ratio $OR(G)$ (red bar) on DEEP. We limit our discussion to BNP and BNF, as BNS has slow processing times on a full-size segment. For BNF, we set the iteration parameter $\beta$ to eight for an optimal balance between efficiency and efficacy. DiskANN's $OR(G)$ is found to be nearly zero, as the majority of vertices within a block are irrelevant. Both BNP and BNF consistently register higher $OR(G)$ than DiskANN, indicating the positive impact of our shuffling algorithms on data locality. BNF further amplifies $OR(G)$ on BNP, incurring a minor additional time cost, which is negligible about the total index processing time ($\sim$9.5\%). Due to DiskANN's poor locality, almost all of the top-1,000 nearest neighbors are stored across disparate blocks, necessitating the checking of a minimum of 1,000 blocks for the retrieval of top-1,000 results. In comparison, BNP and BNF cut down the number of blocks by over 30\%, enabling {\name} to search through fewer blocks. Fig. \ref{fig: data_locality}(b) evidences BNP and BNF's exceptional performance in search tasks over DiskANN, with BNF outdoing BNP owing to its superior locality. By default, BNF is utilized for block shuffling in {\name}.

\vspace{0.2em}
\noindent\textbf{In-memory navigation graph.}
To ascertain the impact of the in-memory navigation graph, we turn it on/off within {\name}, ensuring all other settings remain identical. As depicted in Fig. \ref{fig: navigation_graph}, this analysis highlights the changes in disk I/Os and throughput on BIGANN. The activation of the in-memory graph leads to a reduction in disk I/Os by approximately 20\% for the same \textit{Recall}. This outcome stems from initiating our search on the disk-based graph from entry points positioned closer to the query, which effectively shortens the search path length (refer to \textbf{\S \ref{subsec: problem_analysis}} for further details). Fig. \ref{fig: navigation_graph}(b) reveals that the in-memory graph notably elevates throughput. Importantly, the in-memory graph implementation does not alter the vertex utilization ratio.

\begin{figure}[!th]
  % \vspace{0.1cm}
  \setlength{\abovecaptionskip}{0.1cm}
  \setlength{\belowcaptionskip}{0.2cm}
  \centering
  \footnotesize
  \stackunder[0.75pt]{\includegraphics[scale=0.161]{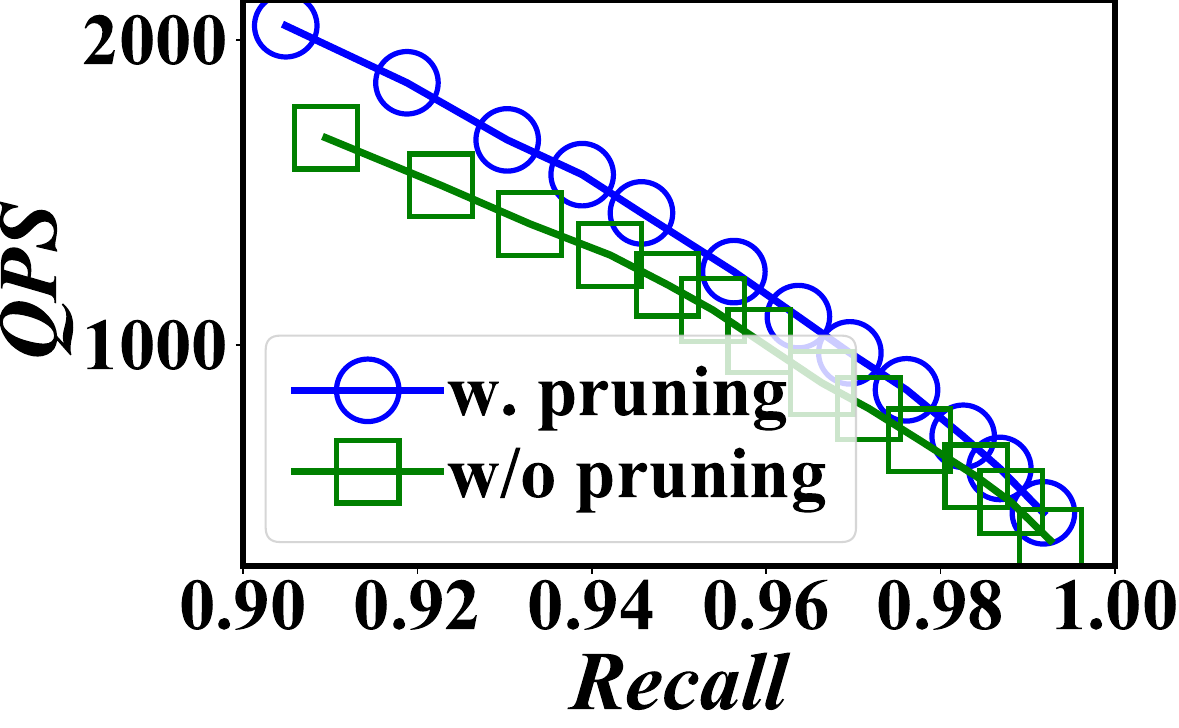}}{(a) Block pruning}
  \hspace{-0.15cm}
  \stackunder[0.75pt]{\includegraphics[scale=0.161]{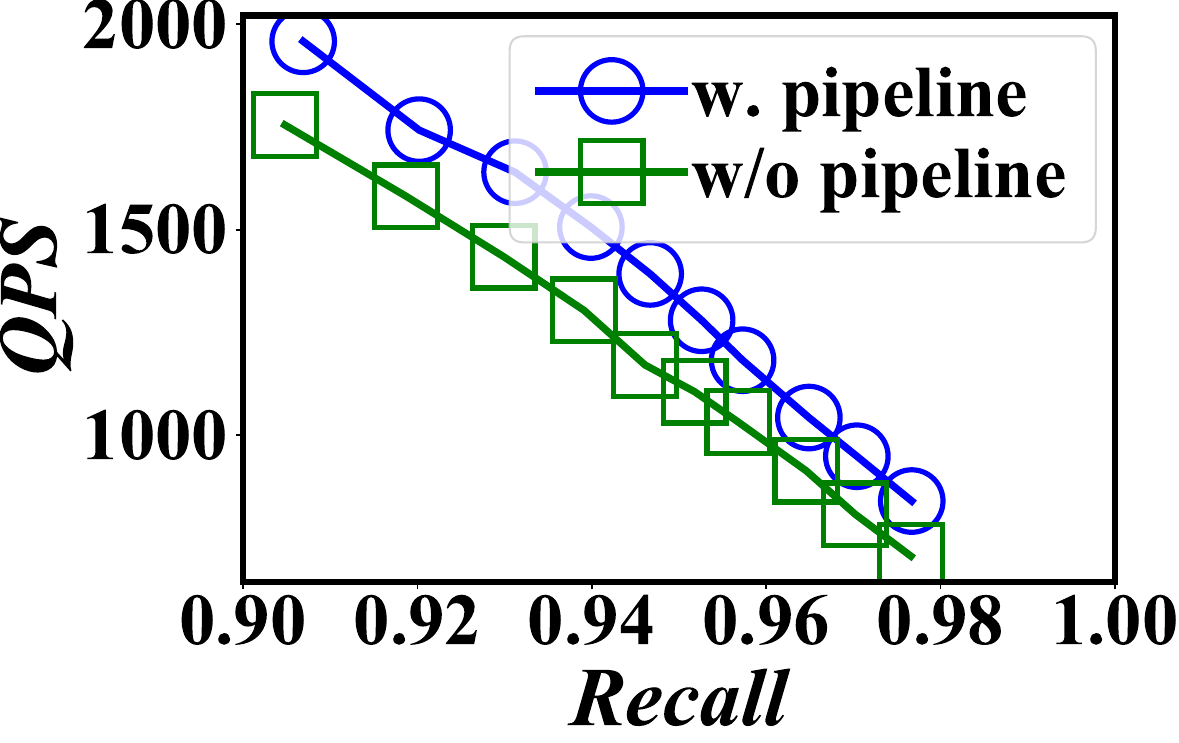}}{(b) I/O-computation pipeline}
  \hspace{-0.15cm}
  \stackunder[0.75pt]{\includegraphics[scale=0.161]{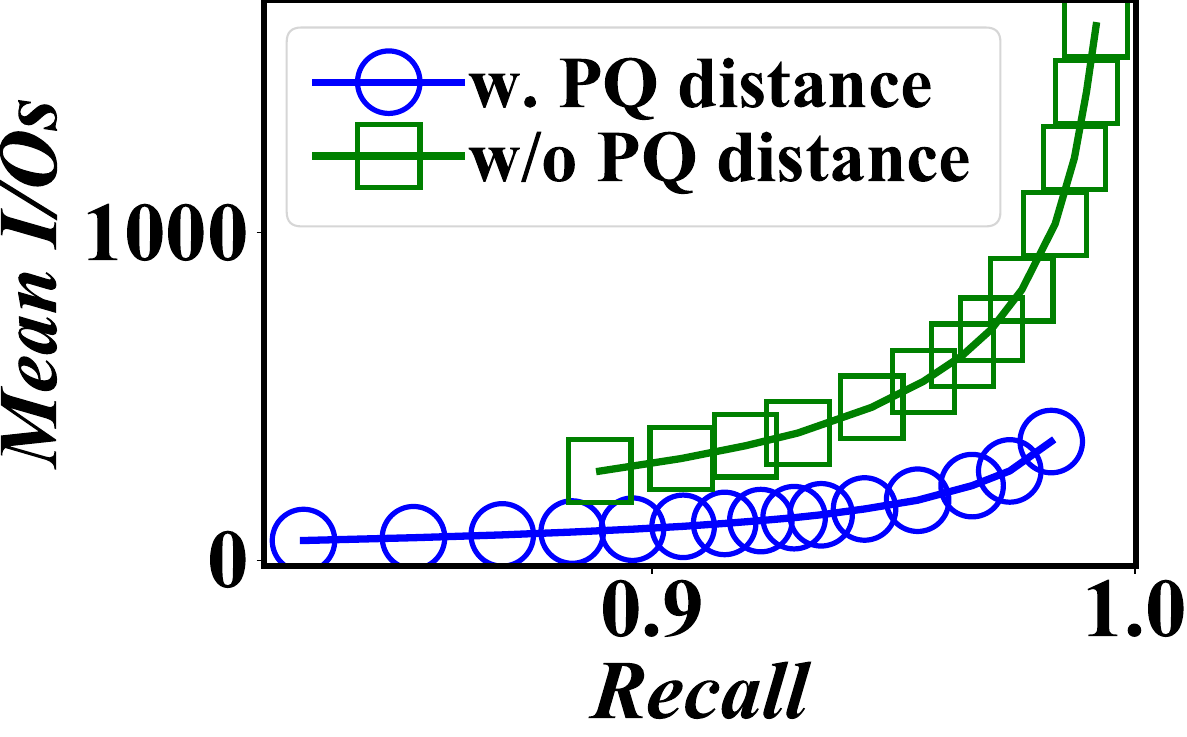}}{(c) PQ distance}
  \hspace{-0.15cm}
  \stackunder[0.75pt]{\includegraphics[scale=0.6]{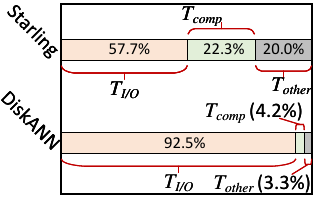}}{(d) Search profile}
  \caption{Effect of the block search optimizations.}
  \label{fig: block search optimization}
  \vspace{-0.5cm}
\end{figure}

\vspace{0.2em}
\noindent\textbf{Block pruning.}
We assess the influence of block pruning on BIGANN by comparing {\name} both with and without the implementation of pruning. The trade-off between \textit{QPS} and \textit{Recall} is presented for both scenarios. Fig. \ref{fig: block search optimization}(a) demonstrates that {\name}, when equipped with block pruning, exhibits superior performance. This can be attributed to the negation of unnecessary computations by pruning vertices that are distant from the query. Note that achieving an ideal graph layout with a perfect overlap ratio $OR(G)=1$ is a challenging task (refer to \textbf{\S \ref{subsec: graph_reorder}} for details). Some vertices located within a loaded block are not neighbors of the target vertex. Block pruning aids in filtering out these non-neighbor vertices, effectively circumventing unnecessary visits.

\vspace{0.2em}
\noindent\textbf{I/O and computation pipeline.}
Fig. \ref{fig: block search optimization}(d) delineates the breakdown of search time for two frameworks with an identical \textit{Recall} on BIGANN. For DiskANN, disk I/O is accountable for as much as 92.5\% of the total search time, thus establishing itself as the bottleneck for HVSS on the disk-based graph. On the contrary, {\name} enhances the vertex utilization ratio by leveraging superior data locality, thereby diminishing the I/O time ($T_{I/O}$) ratio to 57.7\%. However, {\name} also escalates the distance computation time $T_{comp}$ since more vertices are examined in each loaded block. As a result, both I/O and computation are dominant and need to be executed in parallel for optimal performance. Fig. \ref{fig: block search optimization}(b) outlines the impact of the I/O and computation pipeline, illustrating that the pipeline boosts \textit{QPS} for the same \textit{Recall}. This implementation may result in additional computations, as some vertex visitations within the current loaded block are deferred (see \textbf{\S \ref{subsec: basic_page_search}} for more information). Despite this, the trade-off proves to be beneficial, enabling more efficient utilization of both the disk bandwidth and CPU, thereby improving the overall search performance.

\vspace{0.2em}
\noindent\textbf{PQ-based approximate distance.}
Fig. \ref{fig: block search optimization}(c) depicts the average disk I/Os before and after the implementation of the PQ-based approximate distance optimization on BIGANN for a given batch of queries. It is clear that the number of disk I/Os is significantly reduced by applying this form of optimization, while maintaining the same level of accuracy. This optimization ranks the candidate set by approximate distance, which is calculated using PQ short codes in the main memory, thereby circumventing the need for expensive disk I/O operations. Although employing the exact distance might theoretically result in enhanced upper bound accuracy, the increased disk I/O cost can be prohibitively high. With the employment of this optimization, it is possible to maintain high levels of accuracy by adjusting other parameters (such as the size of the candidate set) while limiting the amount of disk I/Os, thus making it more suitable for practical applications.

\subsection{Parameter Sensitivity} \label{subsec: param_sensi}

\setlength{\textfloatsep}{0cm}
\setlength{\floatsep}{0cm}
\begin{figure*}[!th]
\setlength{\abovecaptionskip}{0cm}
\setstretch{0.9}
\fontsize{8pt}{4mm}\selectfont
\begin{minipage}{0.487\textwidth}
  \setlength{\abovecaptionskip}{0.1cm}
  \setlength{\belowcaptionskip}{0cm}
  \centering
  \footnotesize
  \stackunder[0.75pt]{\includegraphics[scale=0.161]{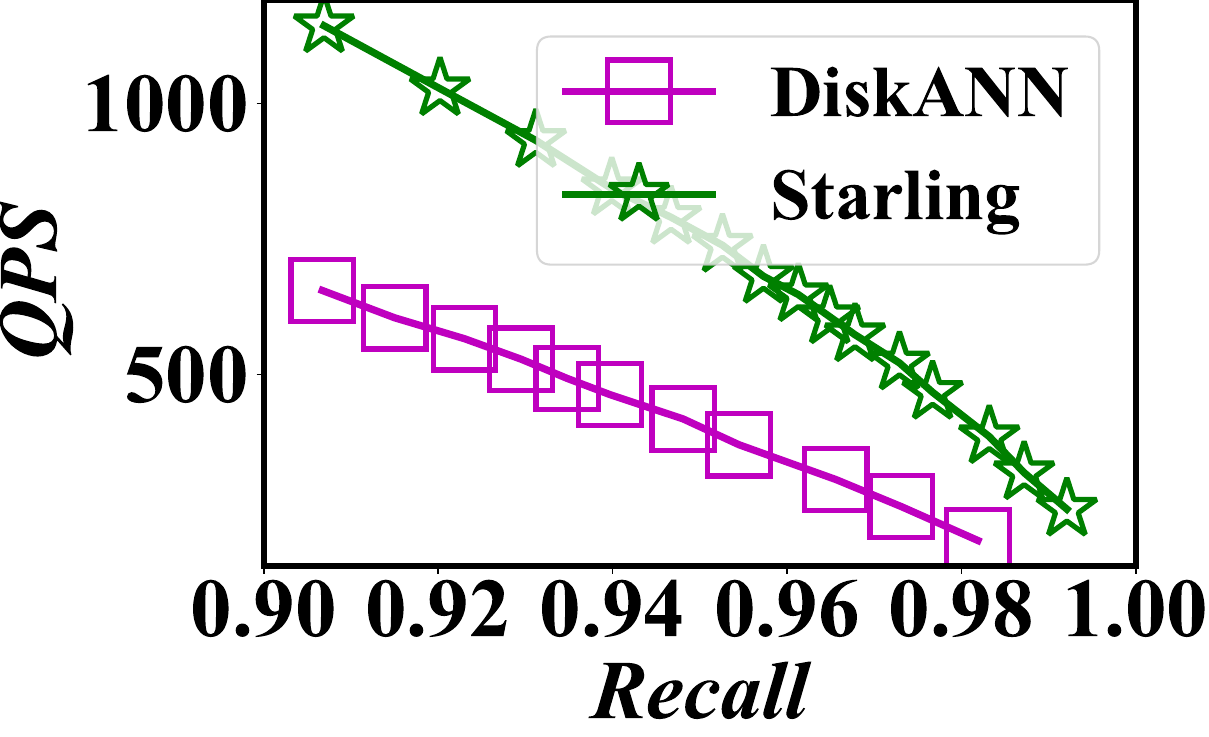}}{(a) Threads = 4}
  \hspace{-0.16cm}
  \stackunder[0.75pt]{\includegraphics[scale=0.161]{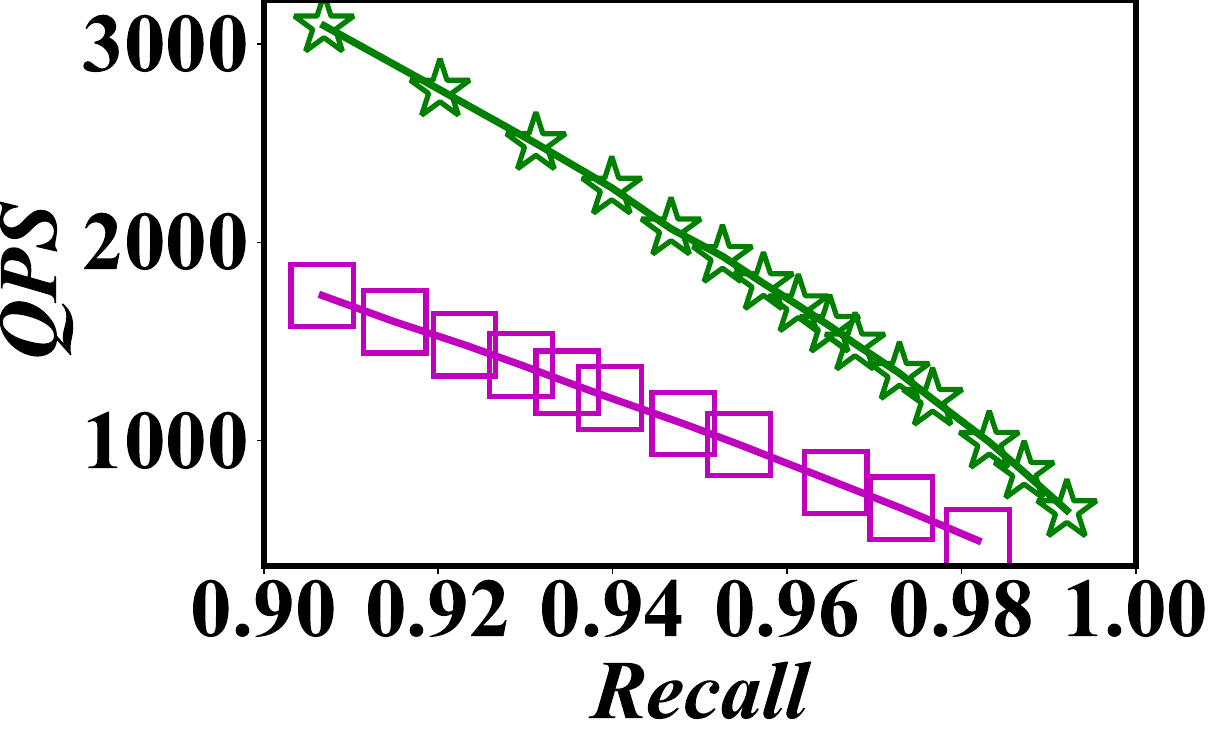}}{(b) Threads = 16}
  \caption{Effect of different threads.}
  \label{fig: threads}
\end{minipage}
\begin{minipage}{0.487\textwidth}
  \setlength{\abovecaptionskip}{0.1cm}
  \setlength{\belowcaptionskip}{0cm}
  \centering
  \footnotesize
  \stackunder[0.75pt]{\includegraphics[scale=0.161]{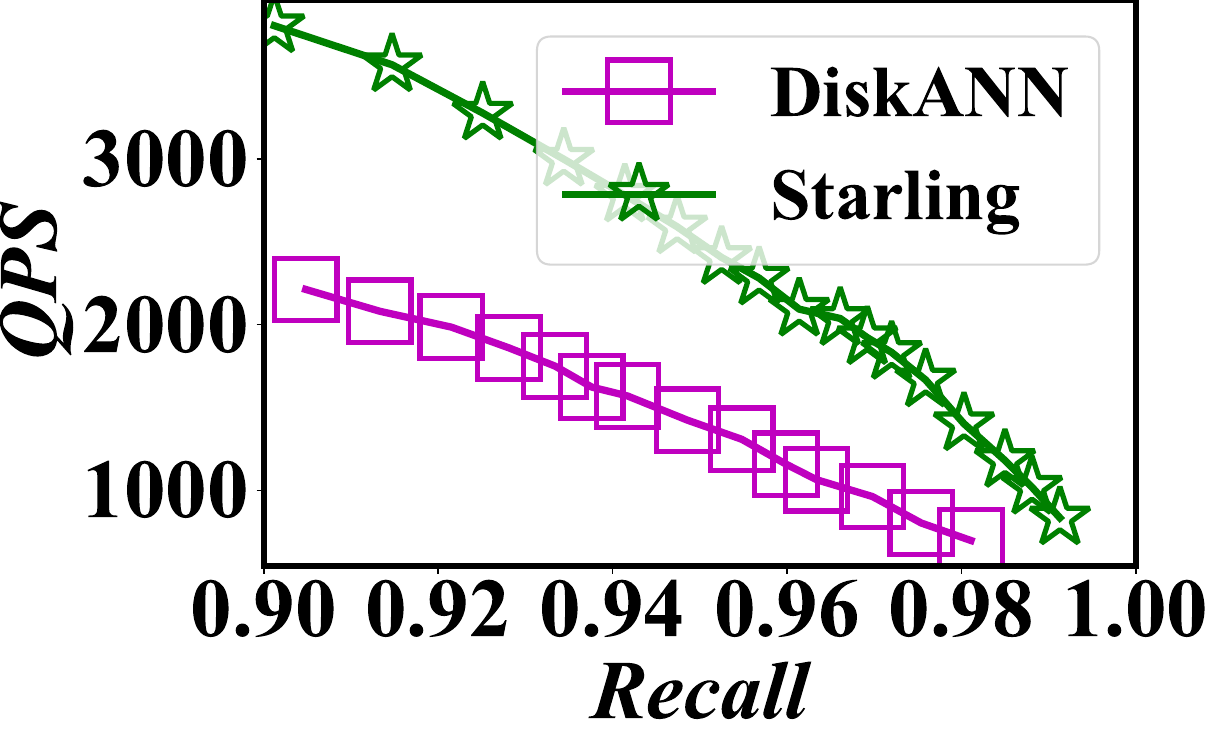}}{(a) $k$ = 1}
  \hspace{-0.16cm}
  \stackunder[0.75pt]{\includegraphics[scale=0.161]{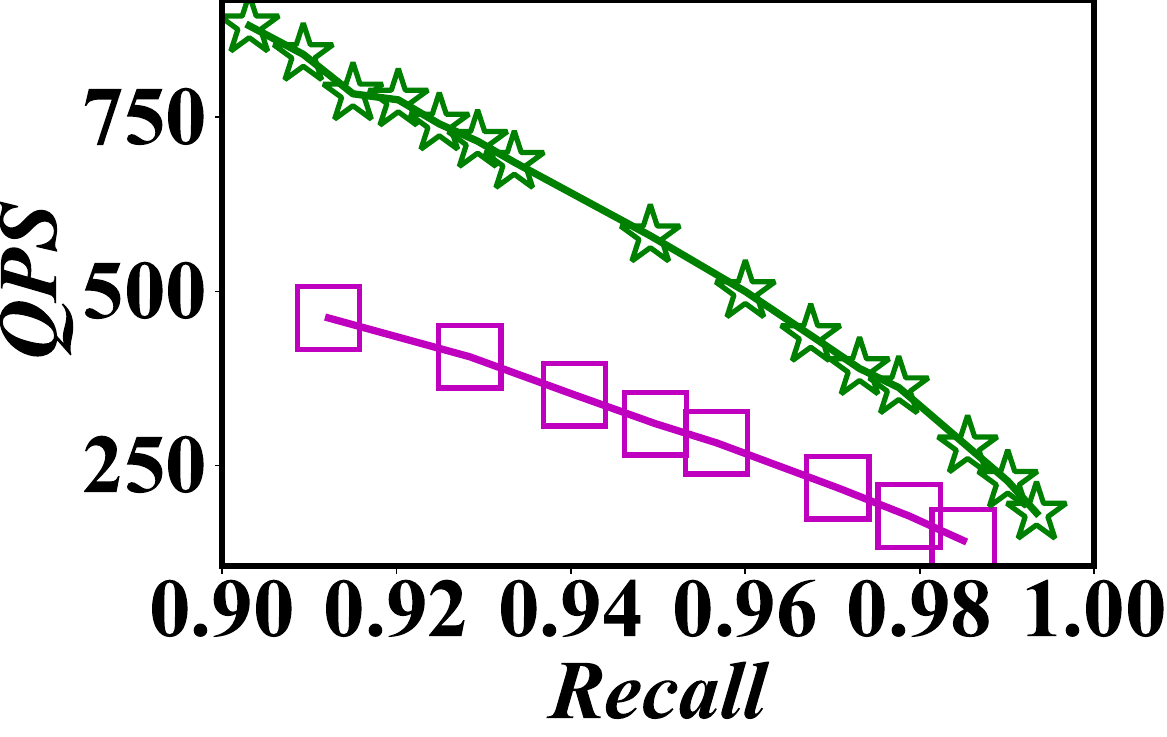}}{(b) $k$ = 50}
  \caption{Effect of different $k$ in ANNS.}
  \label{fig: recall_k}
\end{minipage}
\vspace{-0.4cm}
\end{figure*}

\noindent\textbf{Number of threads.}
We carry out an evaluation of the parallelism of both {\name} and DiskANN, employing various thread settings ranging from 4 to 16 on BIGANN. Fig.\ref{fig: threads} highlights the \textit{QPS} vs \textit{Recall} relationship. In both frameworks, the \textit{Recall} remains constant with thread variations as the number of threads does not influence the search path. Notably, irrespective of the number of threads used, the \textit{QPS} of {\name} consistently proves to be twice as fast as DiskANN.

\vspace{0.2em}
\noindent\textbf{Number of search results.}
We conduct an examination on the effect of varying $k$ from 1 to 50 on BIGANN. In Fig.\ref{fig: recall_k}, it is clear that {\name} exhibits a higher \textit{QPS} than DiskANN for different values of $k$. This observation showcases the robustness of {\name} in managing ANNS problems.

\begin{figure}[!th]
  % \vspace{0.1cm}
  \setlength{\abovecaptionskip}{0.1cm}
  \setlength{\belowcaptionskip}{0.2cm}
  \centering
  \footnotesize
  \stackunder[0.75pt]{\includegraphics[scale=0.161]{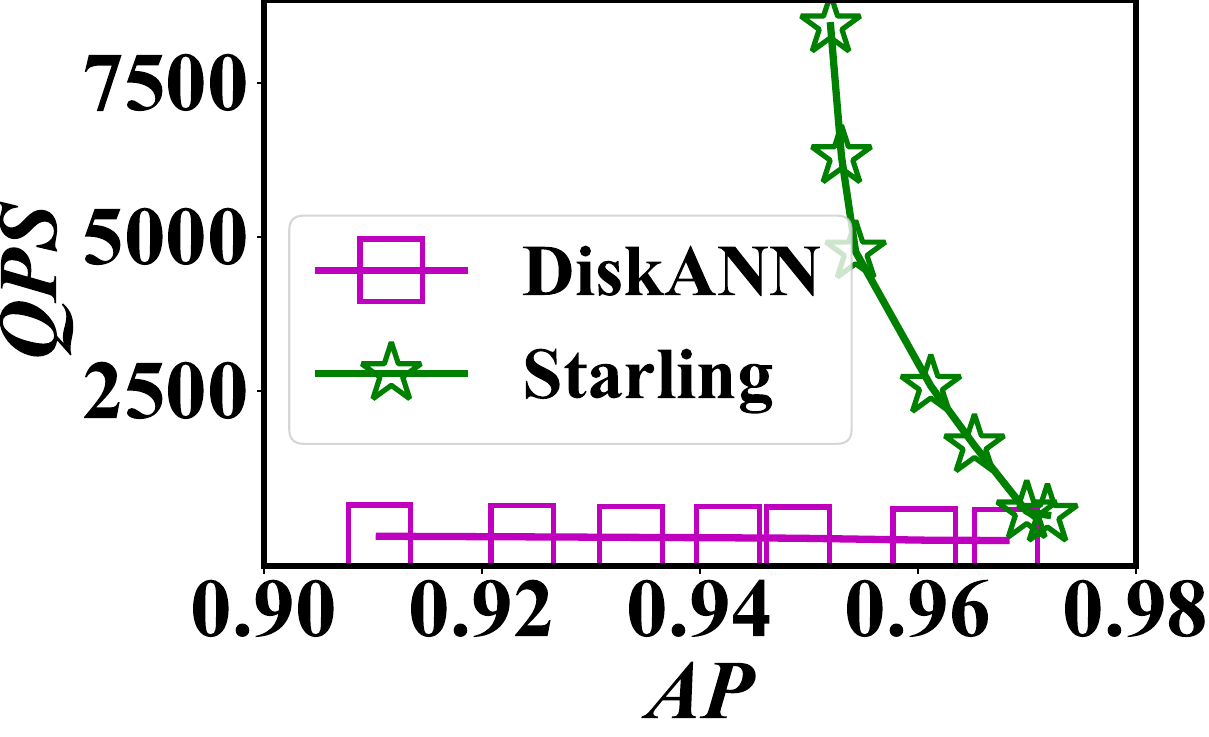}}{(a) $r$ = 3,400}
  \hspace{-0.16cm}
  \stackunder[0.75pt]{\includegraphics[scale=0.161]{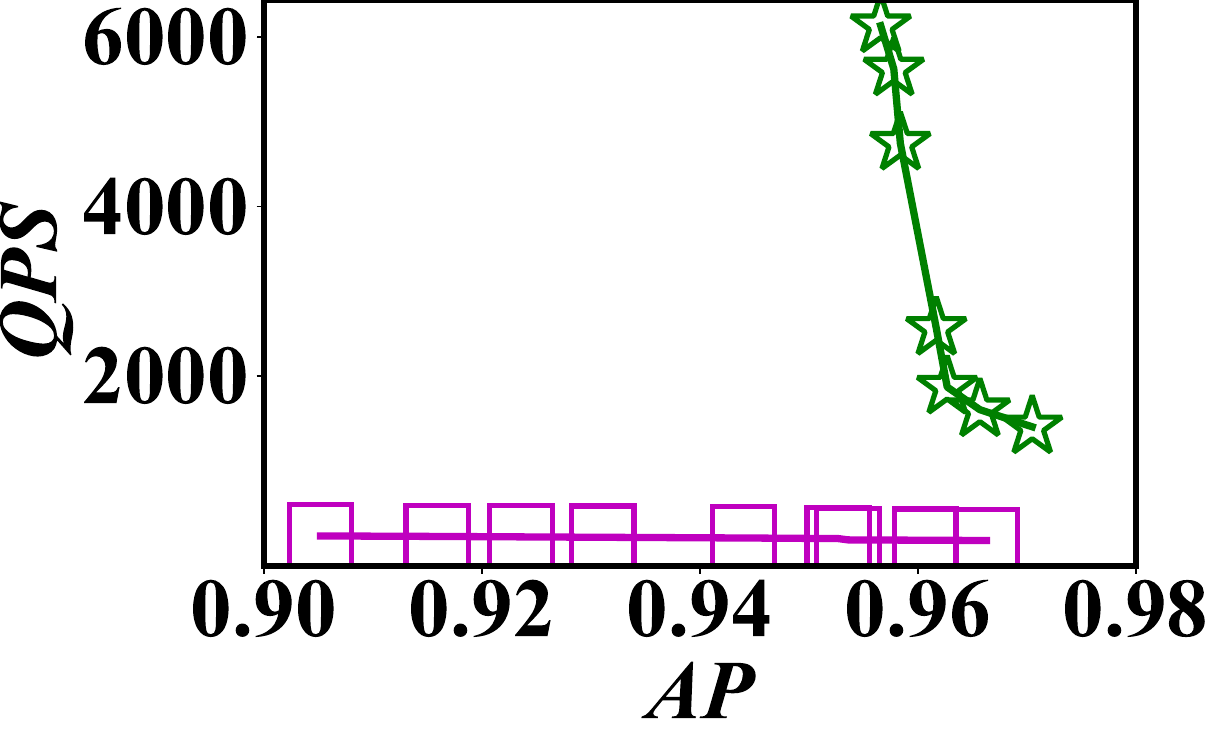}}{(b) $r$ = 3,920}
  \caption{Effect of different $r$ in RS.}
  \label{fig: ap_r}
  \vspace{-0.2cm}
\end{figure}

\vspace{0.2em}
\noindent\textbf{Search radius.}
Fig. \ref{fig: ap_r} presents a comparison of the RS performance with different radii $r$ on BIGANN. The results unequivocally demonstrate that {\name} surpasses DiskANN across different $r$. This stark contrast highlights the superior performance of {\name} in handling various $r$.

\subsection{Scalability} \label{subsec: scalability}

\setlength{\textfloatsep}{0cm}
\setlength{\floatsep}{0cm}
\begin{table}[!th]
  \centering
  \setlength{\abovecaptionskip}{0.05cm}
  \setlength{\belowcaptionskip}{-0.3cm}
  \setstretch{0.8}
  \fontsize{6.5pt}{3.3mm}\selectfont
  \tabcaption{\textit{QPS} of different number of segments on BIGANN.}
  \label{tab: multi-segment}
  \setlength{\tabcolsep}{.011\linewidth}{
  \begin{tabular}{|l|l|l|l|l|}
    \hline
    \textbf{Query $\rightarrow$} & \multicolumn{2}{c|}{\textbf{RS ($AP=0.90$)}} & \multicolumn{2}{c|}{\textbf{ANNS ($Recall=0.99$)}}\\
    \hline
    \textbf{\# Seg. $\downarrow$} & \textbf{DiskANN} & \boldsymbol{\name} & \textbf{DiskANN} & \boldsymbol{\name} \\
    \hline
    1 & 181 & \textbf{\textbf{8,690} (48.0$\times$)} & 239 & \textbf{538 (2.3$\times$)} \\
    \hline
    2 & 42 & \textbf{1,040 (24.8$\times$)} & 161 & \textbf{321 (2.0$\times$)} \\
    \hline
    3 & 31 & \textbf{687 (22.2$\times$)} & 131 & \textbf{257 (2.0$\times$)} \\
    \hline
    4 & 23 & \textbf{472 (20.5$\times$)} & 98 & \textbf{193 (2.0$\times$)} \\
    \hline
    5 & 19 & \textbf{204 (10.7 $\times$)} & 79 & \textbf{163 (2.1$\times$)} \\
    \hline
  \end{tabular}
  }\vspace{-0.3cm}
\end{table}

\noindent\textbf{Number of segments.}
We set a fixed segment size and manipulate the number of segments needed to process a batch of queries. Tab. \ref{tab: multi-segment} displays the scalability of {\name} on a single machine. We report the \textit{QPS} of both frameworks maintaining equivalent accuracy for ANNS and RS. The findings illustrate that {\name} persistently surpasses DiskANN in performance for both ANNS and RS, testifying to its superior scalability in relation to the number of segments.

\setlength{\textfloatsep}{0cm}
\setlength{\floatsep}{0cm}
\begin{figure*}[!th]
\setlength{\abovecaptionskip}{0cm}
\setstretch{0.9}
\fontsize{8pt}{4mm}\selectfont
\begin{minipage}{0.447\textwidth}
  \setlength{\abovecaptionskip}{0.1cm}
  \setlength{\belowcaptionskip}{0cm}
  \centering
  \footnotesize
  \stackunder[0.75pt]{\includegraphics[scale=0.161]{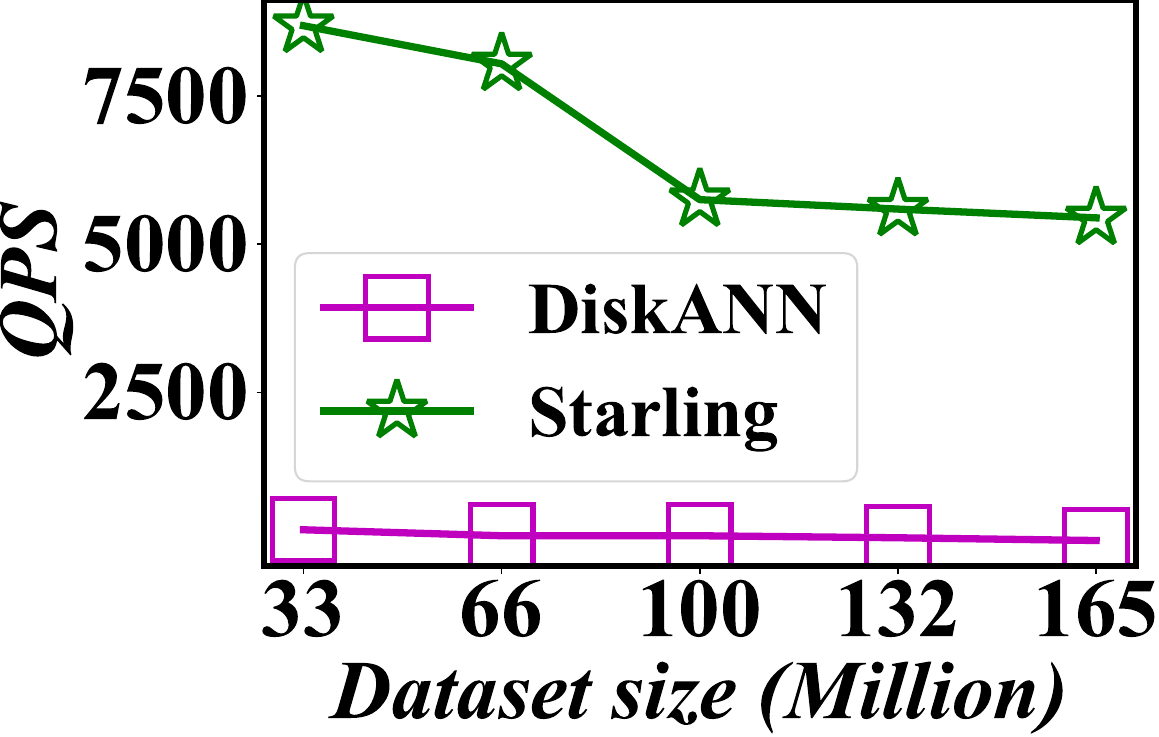}}{(a) RS ($AP$ = 0.90)}
  \hspace{-0.15cm}
  \stackunder[0.75pt]{\includegraphics[scale=0.161]{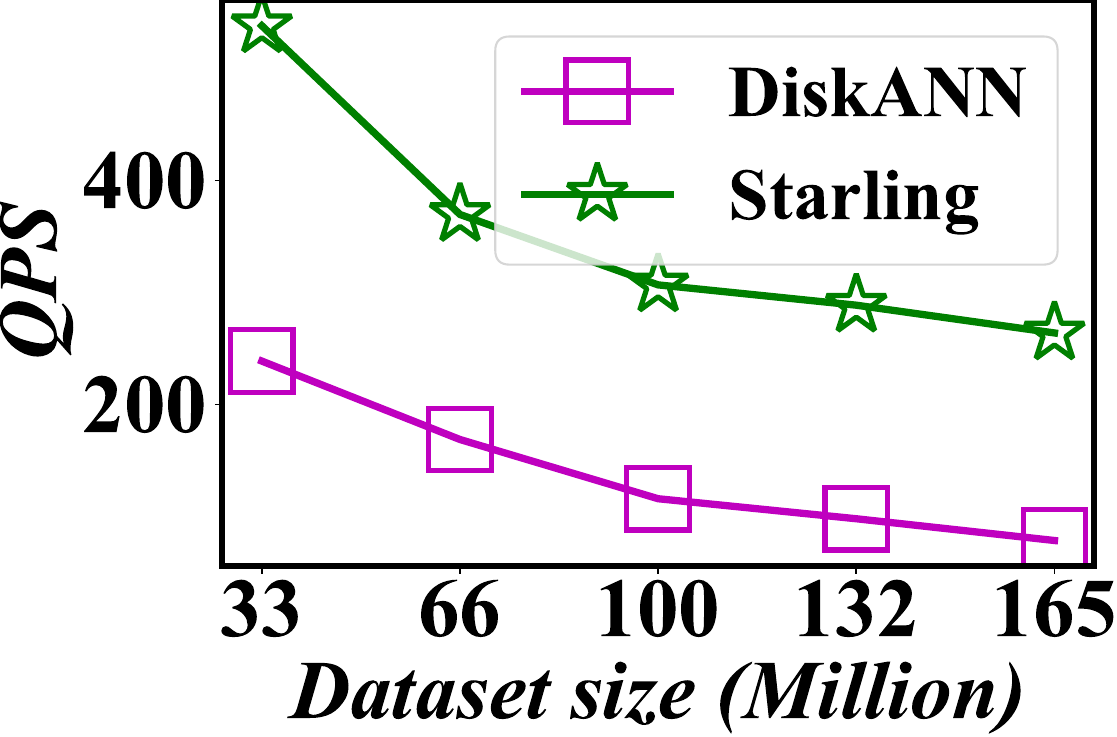}}{(b) ANNS ($Recall$ = 0.99)}
  \caption{Different data segment sizes.}
  \label{fig: larger_segment}
\end{minipage}
\begin{minipage}{0.482\textwidth}
  \setlength{\abovecaptionskip}{0.1cm}
  \setlength{\belowcaptionskip}{0cm}
  \centering
  \footnotesize
  \stackunder[0.75pt]{\includegraphics[scale=0.161]{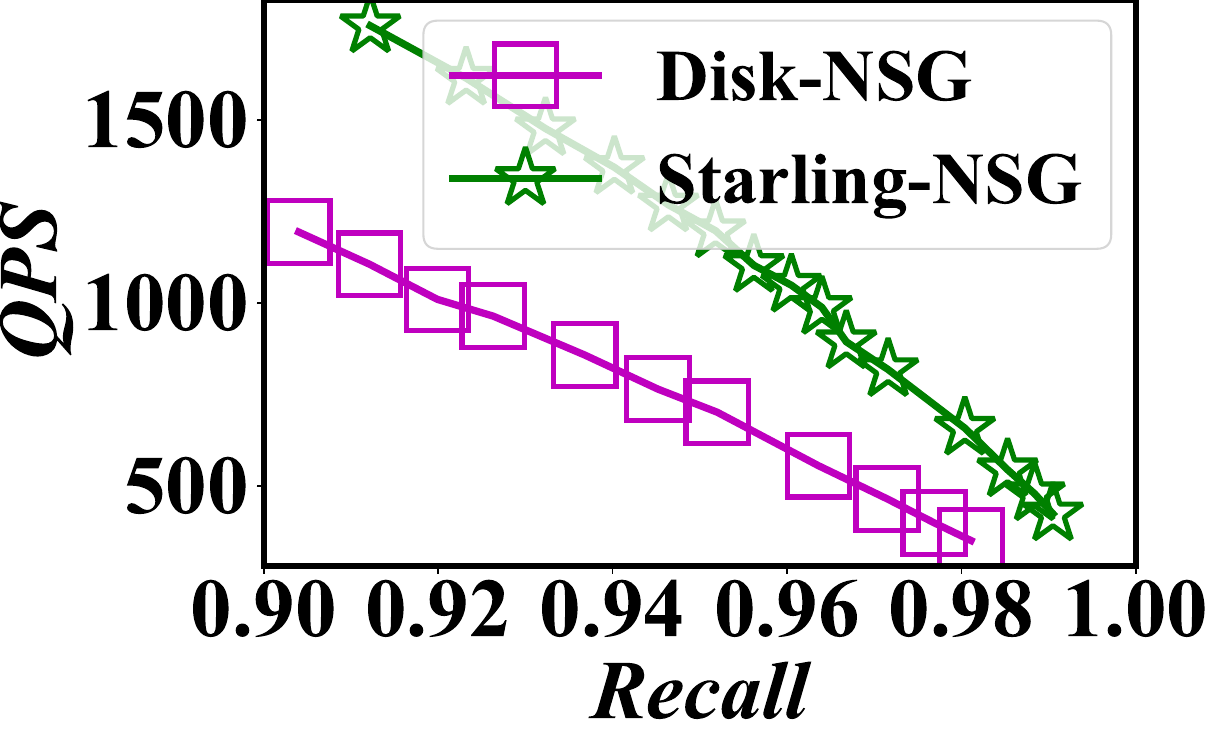}}{(a) NSG}
  \hspace{-0.15cm}
  \stackunder[0.75pt]{\includegraphics[scale=0.161]{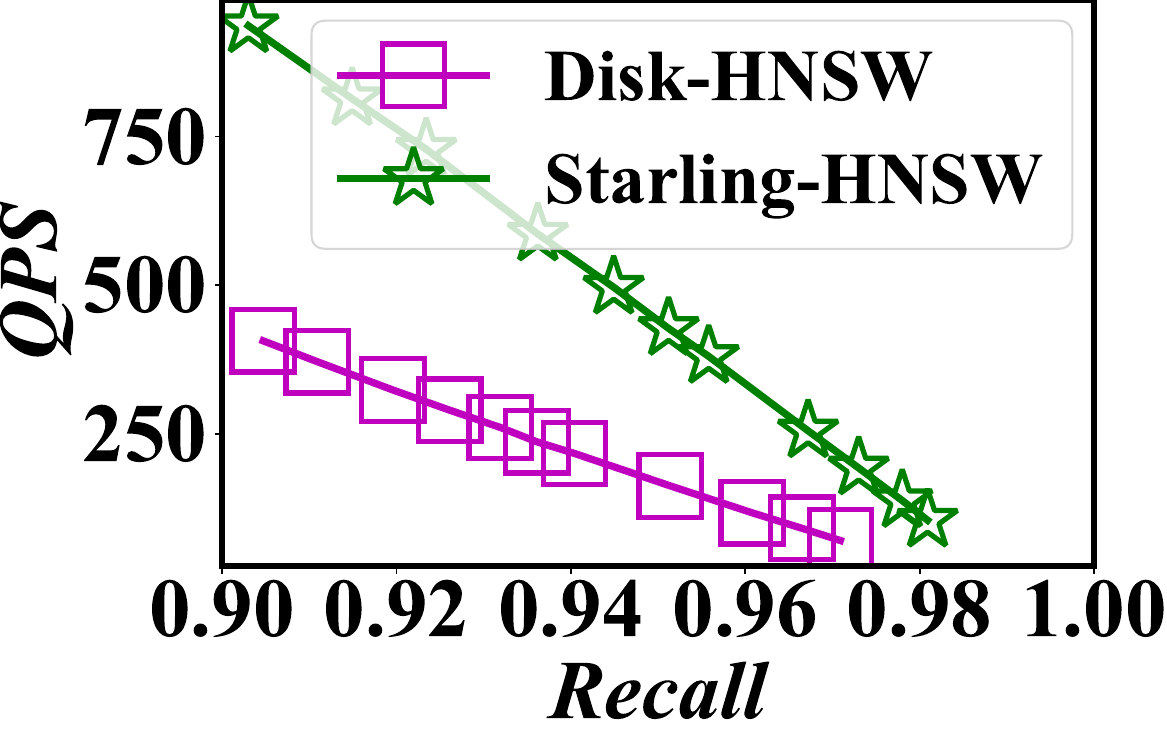}}{(b) HNSW}
  \caption{Different graph algorithms.}
  \label{fig: framework_graph}
\end{minipage}
\vspace{-0.4cm}
\end{figure*}

\vspace{0.2em}
\noindent\textbf{Segment size.}
For other experiments, we establish a fixed dataset size of 4GB within a segment. We test {\name}'s performance on varying segment sizes via distinct dataset volumes. Remember, each segment's memory and disk space increase proportionately. Fig. \ref{fig: larger_segment} demonstrates the impact of data segment sizes. {\name} sustains a higher \textit{QPS} than DiskANN for both RS and ANNS across varying data volumes. This emphasizes that {\name} scales efficiently with segment size.

\vspace{0.2em}
\noindent\textbf{Graph algorithms.} By default, {\name} employs the Vamana algorithm \cite{DiskANN} to construct the disk-based graph index. Here, we assess search performance on disk with two other graph indexes, NSG \cite{NSG} and HNSW \cite{HNSW}. We follow DiskANN to implement NSG and HNSW, but substitute NSG and HNSW ({layer-0}) for Vamana, resulting in \texttt{Disk}-NSG and \texttt{Disk}-HNSW. We also adapt {\name} to use NSG and HNSW, obtaining {\name}-NSG and {\name}-HNSW. {Notably, {\name}-HNSW has a multi-layered in-memory navigation graph with HNSW's upper-layered graphs.} Fig. \ref{fig: framework_graph} shows the \textit{QPS} vs \textit{Recall} comparison. We find that the two graph indexes on {\name} perform better than the baseline framework. For example, {\name}-HNSW's QPS is over 2$\times$ higher than \texttt{Disk}-HNSW. This shows {\name}'s generality to support other graph algorithms.

\setlength{\textfloatsep}{0cm}
\setlength{\floatsep}{0cm}
\begin{figure*}[!th]
\setlength{\abovecaptionskip}{0cm}
\setstretch{0.9}
\fontsize{8pt}{4mm}\selectfont
\begin{minipage}{0.497\textwidth}
  \setlength{\abovecaptionskip}{0.1cm}
  \setlength{\belowcaptionskip}{0cm}
  \centering
  \footnotesize
  \stackunder[0.75pt]{\includegraphics[scale=0.176]{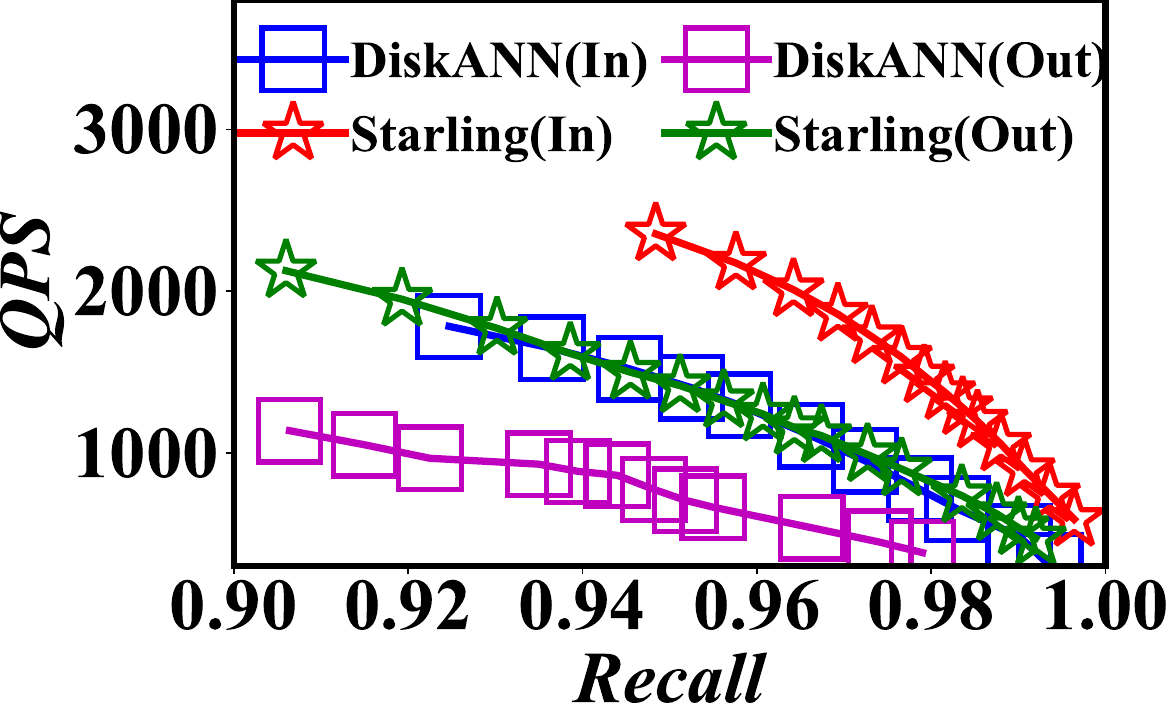}}{(a) Query distribution}
  \hspace{-0.25cm}
  \stackunder[0.75pt]{\includegraphics[scale=0.176]{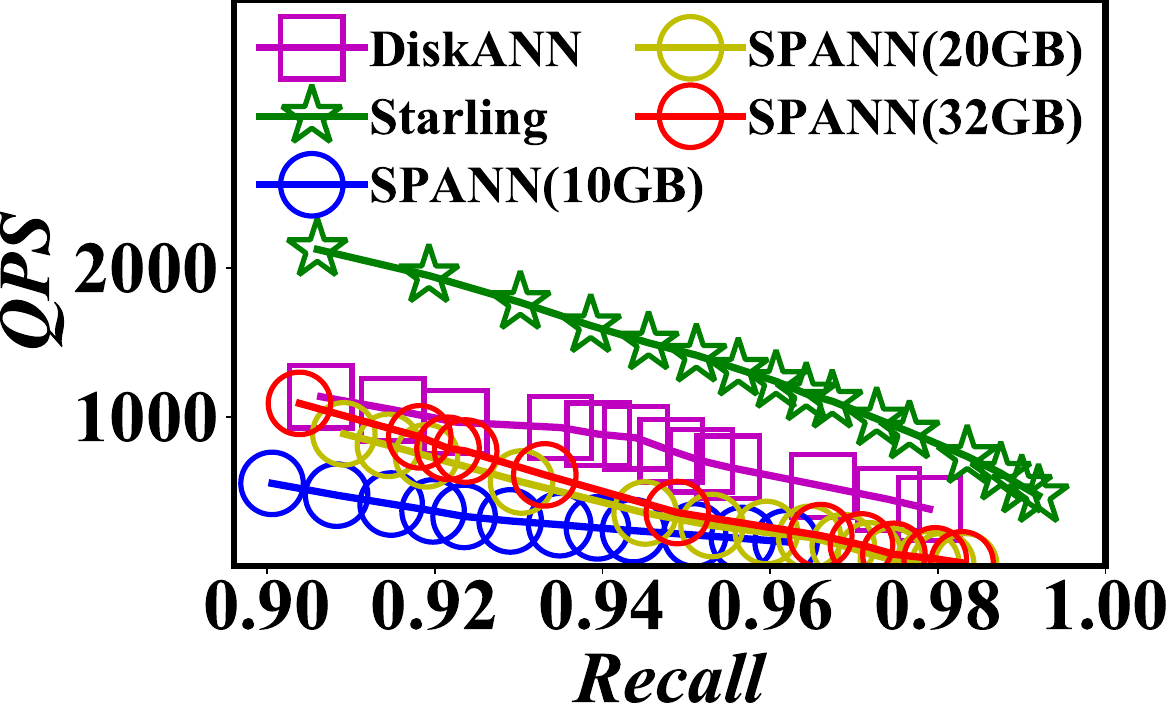}}{(b) Segment setup}
  \caption{Different queries and segments.}
  \label{fig: query distribution and segment setup}
\end{minipage}
\begin{minipage}{0.497\textwidth}
  \setlength{\abovecaptionskip}{0.1cm}
  \setlength{\belowcaptionskip}{0cm}
  \centering
  \footnotesize
  \stackunder[0.75pt]{\includegraphics[scale=0.18]{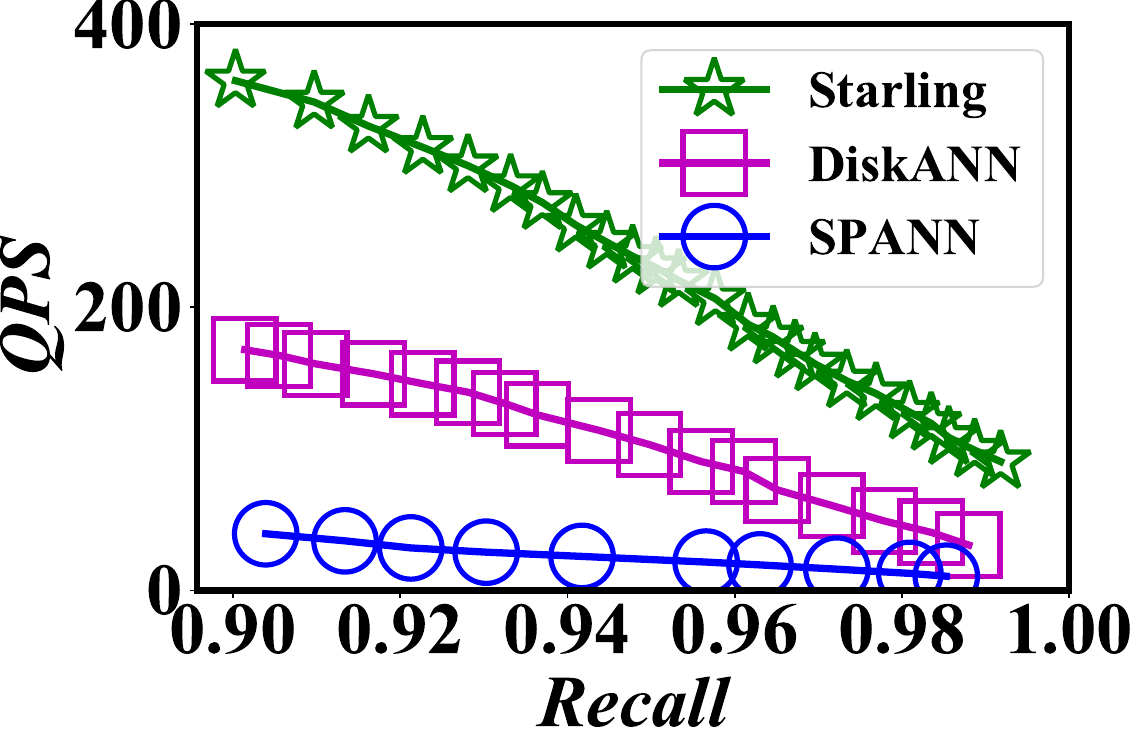}}{(a) 8GB}
  \hspace{-0.16cm}
  \stackunder[0.75pt]{\includegraphics[scale=0.18]{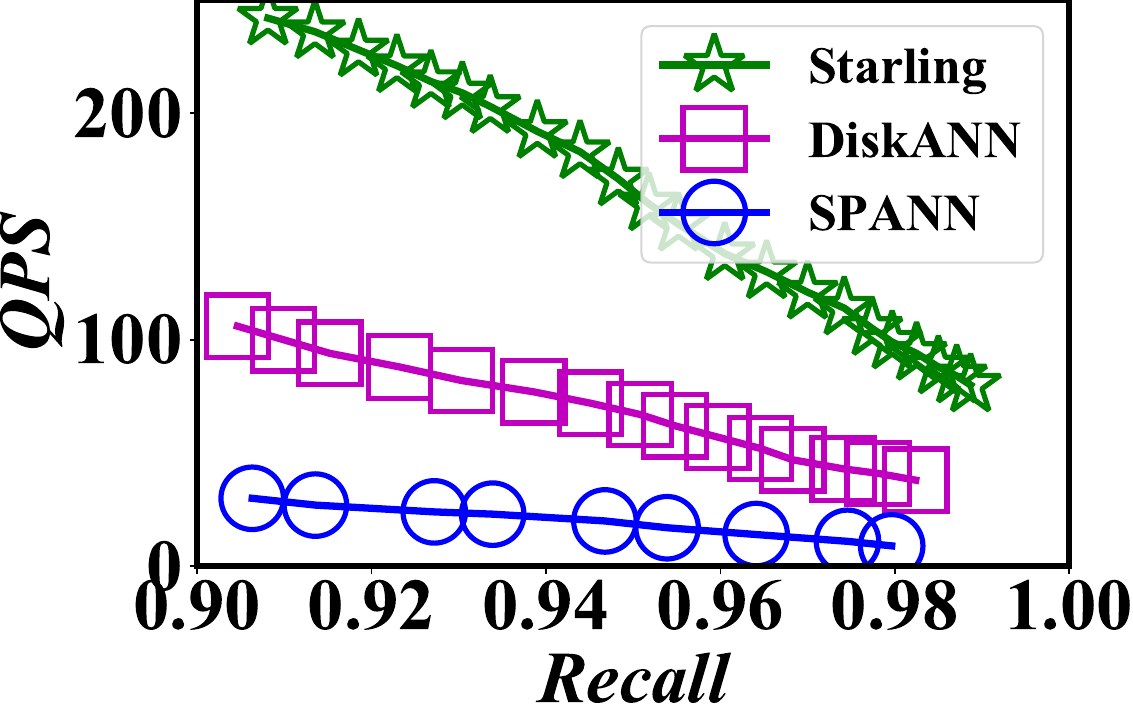}}{(b) 16GB}
  \caption{Different data sizes.}
  \label{fig: different data sizes}
\end{minipage}
% \vspace{-0.2cm}
\end{figure*}

\subsection{In- and Not-in-Database Queries}
\label{subsec: differnt query type}
We compare different methods on in-database and not-in-database queries on BIGANN. In-database queries are vectors that exist in the graph index, while not-in-database queries are not. We randomly sample 10,000 vectors from the base data as in-database queries and execute the search on two frameworks. Fig. \ref{fig: query distribution and segment setup}(a) shows that {\name} consistently outperforms DiskANN for both types of queries. Moreover, in-database queries achieve higher throughput than not-in-database queries for both frameworks. This is because in-database queries can leverage some better query-aware entry points, which may be the query points themselves. However, the in-memory navigation graph can also find some entry points close to the not-in-database queries to shorten their search path. Notably, both types of query workloads benefit from data locality through offline block shuffling.

\subsection{Evaluation on Other Segment Setups}
\label{subsec: segment setup}
\noindent\textbf{Effect of disk capacity.} {We fix the dataset size at 4GB and memory size at 2GB, and test different segment setups with the disk capacity from 10GB to 32GB. Fig. \ref{fig: query distribution and segment setup}(b) shows the search performance on BIGANN. DiskANN and {\name} exhibit the best \textit{QPS}-\textit{Recall} trade-off with an index size of less than 10GB, so we only plot one curve for each of them. SPANN improves its performance as the disk space increases, because it can duplicate more boundary data points, thereby reducing disk I/Os. However, {\name} still performs much better than SPANN, with a smaller disk setup.}

\noindent\textbf{Effect of dataset size.} We employ a fixed segment space with a memory of 2GB and a disk capacity of 32GB to test diverse dataset sizes of 4GB, 8GB, and 16GB. Fig. \ref{fig: different data sizes} shows the search performance of different methods on BIGANN (see Fig. \ref{fig: query distribution and segment setup}(b) for the 4GB dataset). We tune the parameters of all methods to get the best \textit{QPS}-\textit{Recall} trade-off under the segment space constraint. The results conclusively depict the superior performance of {\name} over rival methods across all dataset sizes. Importantly, the performance gap increases as the dataset size magnifies. SPANN suffers from a larger dataset size because it cannot replicate enough data to minimize disk I/Os.

\subsection{Large-scale Search Results}
\label{subsec: big top-k result}

{We set the number of search results ($k$) to less than 50 for other experiments. In some cases, we may need much more results, such as thousands. For example, in recommendation systems, a large number of candidates are first recalled and then filtered to get the final recommendations \cite{ChenLZWLMHJXDZ22,NHQ}. Fig. \ref{fig: big results and data}(a) shows the search performance with $k=5,000$ on BIGANN. We can see that {\name} has a much lower \textit{Mean I/Os} than DiskANN. For instance, with a \textit{Recall} of 0.99, {\name} saves more than 20,000 disk I/Os per query than DiskANN. This shows that {\name} is more efficient and effective in scenarios with large-scale search results.}

\begin{figure}[!th]
  % \vspace{0.1cm}
  \setlength{\abovecaptionskip}{0.1cm}
  \setlength{\belowcaptionskip}{0.2cm}
  \centering
  \footnotesize
  \stackunder[0.75pt]{\includegraphics[scale=0.2]{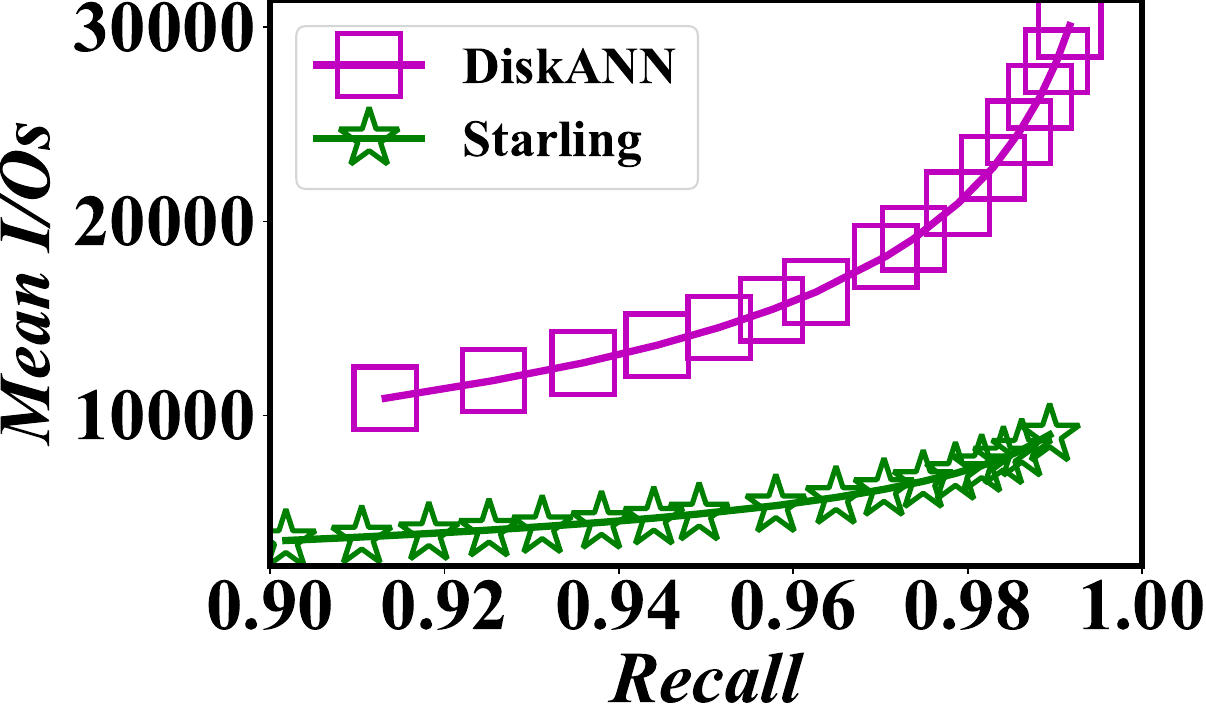}}{(a) $k=5000$}
  \hspace{-0.16cm}
  \stackunder[0.75pt]{\includegraphics[scale=0.2]{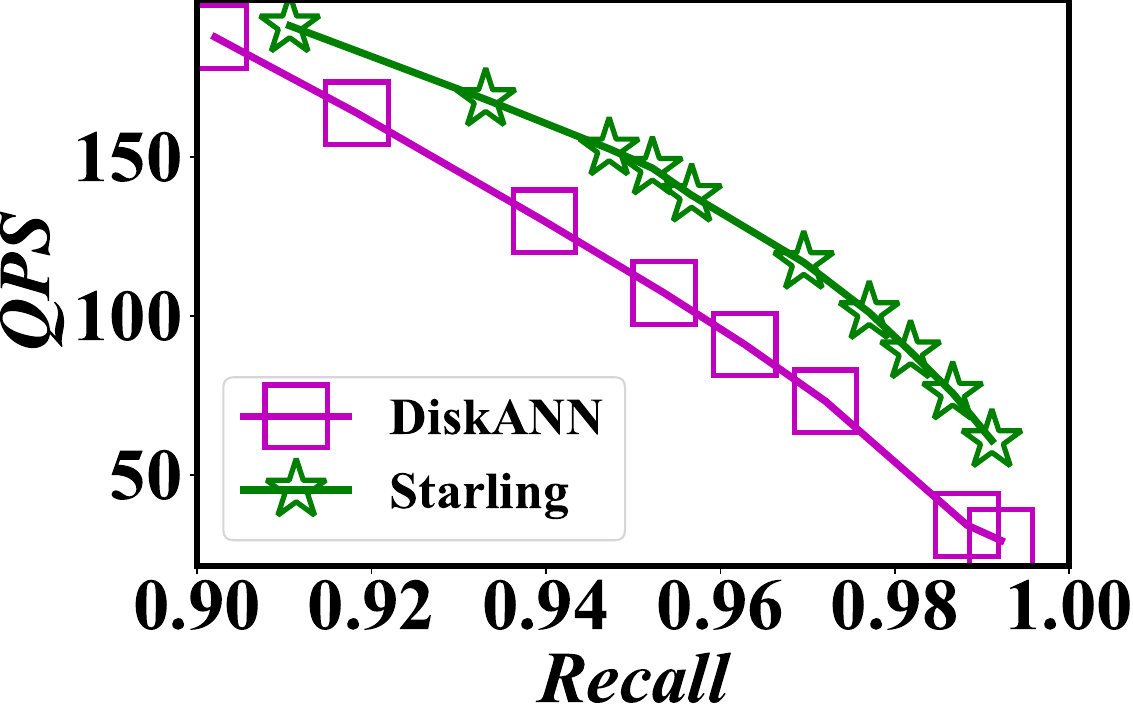}}{(b) Billion-scale data}
  \caption{Large-scale search results and Billion-scale data.}
  \label{fig: big results and data}
  \vspace{-0.2cm}
\end{figure}

\subsection{Evaluation on Billion-scale Data}
\label{subsec: billion-scale evaluation}
{We evaluate our method on the one billion BIGANN dataset. We split the dataset into 31 segments, each with 2GB memory and 10GB disk. One query node has only 32GB of memory in our experiments, so we assigned 31 segments to two query nodes. We merge candidates from each segment to get the final results. We use the same setting for {\name} and DiskANN to ensure fairness. Fig. \ref{fig: big results and data}(b) shows that {\name} is more than 2$\times$ faster than DiskANN in the high recall regime (e.g., \textit{Recall} $>0.96$ ). Note that current vector databases use some data segmentation strategies and a query coordinator, which can avoid scanning all segments for a query \cite{Manu_zilliz}. Our method can work well with these strategies. For more details on data segmentation strategies, please refer to \cite{lanns,zhang2022leqat}.}

%% file: sections/6_discussion.tex
\section{Discussion}
\label{sec: discussion}
\noindent\textbf{Memory-based related work.}
{Memory-based HVSS often involves preprocessing data to strike a balance between efficiency and accuracy. Existing algorithms fall into four categories: tree- \cite{DasguptaF08, LuWWK20, MujaL14}; quantization- \cite{PQ, ScaNN, AndreKS15}; hashing- \cite{HuangFZFN15, GongWOX20, LiZSWT020}; and graph-based \cite{HNSW,NSG,NSSG}. Tree-based and hashing-based methods are not prevalent in HVSS due to the ``curse of dimensionality'' and low accuracy \cite{DPG}. Quantization-based methods (e.g., IVFPQ) prove efficient and memory-saving but tend to suffer from a poor recall rate\cite{HM_ANN,DiskANN}. Graph-based methods (e.g., HNSW) exhibit leading efficiency-accuracy trade-offs. However, they necessitate both raw vector data and the graph index to be in the main memory, elevating memory consumption and impeding scalability for large-scale vectors \cite{SPANN}. {\name} is a universal and I/O-efficient framework capable of integrating different graph algorithms (e.g., HNSW) into the disk-index component. For example, {\name} accommodates HNSW seamlessly without sacrificing functionality (cf. Fig. \ref{fig: framework_graph}(b)).}

\vspace{0.2em}
\noindent\textbf{Comparison analysis with SPANN.}
While both {\name} and SPANN leverage in-memory graphs and locality-centric disk indexes, they address varying challenges and utilize unique strategies.
SPANN, a clustering-based disk index, partitions data utilizing k-means, thereby achieving an inherent locality allowing clustered data to be stored and accessed synchronously. Furthermore, it builds an in-memory graph for fast cluster retrieval. In contrast, {\name} addresses disk-based graph search by optimizing data layout and search strategy. When residing in memory, graph-based methods excel in terms of accuracy and efficiency. However, once placed on disk, they incur many I/Os due to long search path and poor data locality.
{\name} mitigates these issues by building an in-memory graph to shorten the search path and using block shuffling to enhance locality. Note that the graph index neighborhood exhibits both similarity and navigation traits \cite{graph_survey_vldb2021}. This implies that a vertex may have neighbors from different clusters \cite{HNSW}, posing a major challenge for the locality of the graph index. We compare block shuffling with a naive strategy that assigns vertices to blocks by k-means on SSNPP. The results show that block shuffling achieves a 12$\times$ higher overlap ratio.

\vspace{0.2em}
\noindent\textbf{In-memory graph.}
Our in-memory graph serves as an index directing query routes, instead of acting as a cache for frequently accessed data. We initially used memory mapping (\textsf{mmap}) instead of direct I/O (\textsf{o\_direct}) but found more disk I/Os linked to a memory-mapped graph. This lowers efficiency compared to the hot points strategy in the baseline framework. Our evaluation showed that the in-memory graph outperforms the hot points strategy in search performance and memory overhead. In {\name}, we have the flexibility to utilize any graph algorithm like HNSW \cite{HNSW} or NSG \cite{NSG} for the in-memory graph. We typically use the same algorithm for both in-memory and disk-based graphs. In HNSW, the upper-layer graphs are a subset of the layer-0 graph. Thus, we can keep the higher layers in memory as a multi-layered in-memory graph and store the layer-0 graph on disk. This makes HNSW easy to implement in {\name} (see Fig. \ref{fig: framework_graph}(b)).

\vspace{0.2em}
\noindent\textbf{Central assumption.} With modern SSD (high IOPS), recent methods like DiskANN can fetch multiple blocks simultaneously in each disk round-trip. This is because fetching a small number of random blocks from a disk takes almost the same time as one block \cite{DiskANN}. We keep this feature in {\name} while increasing the vertex utilization ratio in each loaded block by enhancing locality. This reduces the total I/Os for a query by minimizing round-trips.

\vspace{0.2em}
\noindent\textbf{Data update.} {\name}, primarily emphasizing query optimization for static indexes, can handle incremental updates at the database level \cite{Manu_zilliz}. Some vector databases (e.g., ADBV \cite{ADBV}) segregate dynamic and static indexes to facilitate updates. The dynamic index, residing in memory, is built incrementally and uses a bitset to monitor deleted data. As incoming data continues to grow, the dynamic index expands correspondingly. Consequently, an asynchronous merging process transfers the burgeoning dynamic index to the disk-based static index, necessitating a comprehensive index reconstruction. Then, the block shuffling and in-memory graph techniques come into play.

\vspace{0.2em}
\noindent\textbf{Applications of range search (RS).} {RS is widely used in vector analytics, such as face recognition \cite{schroff2015facenet}, near-duplicate detection \cite{facebook_ssnpp}, and various applications on general embeddings \cite{qin2018gph,WangMW22,song2021promips}. In some retrieval systems \cite{park2015reversed,barthel2019real}, ANNS is followed by RS with seed vertices and similarity thresholds \cite{hezel2023fast} to obtain a similar cluster and offer more comprehensive results \cite{li2020novel}. For example, users can choose a seed vertex to get all related results. Several vector databases (e.g., PASE \cite{PASE}, VBase \cite{zhang2023vbase}, Milvus \cite{Milvus_sigmod2021}) and memory-based works \cite{HVS,aoyama2011fast} support both ANNS and RS on the same dataset.}

%% file: sections/7_conclusion.tex
\section{Conclusion}\label{sec: conclusion}
We conduct a study on HVSS for the data segment, which is an essential component in vector databases. HVSS for the data segment needs to meet strict requirements in terms of accuracy, efficiency, memory usage, and disk capacity. However, existing methods only address some of these aspects. Our framework, called {\name}, adopts a disk-based graph index approach and considers all these requirements simultaneously. It optimizes the data layout and search strategy to minimize costly disk I/O operations. Experimental results demonstrate that {\name} achieves significant performance improvements compared to state-of-the-art methods while maintaining a small space overhead. {In the future, we plan to apply our methods to cache and GPU optimizations.} We will also integrate {\name} into Milvus for distributed optimization.

%% file: sections/appendix.tex
\clearpage
\appendix
\section*{Appendix}
\section{Proof of Theorem 4.1}
\label{appendix: proof theorem 4.1}
We recall the triple shuffling problem. Given $t=3\cdot \rho$ integers $\alpha_0$, $\alpha_1$, $\cdots$, $\alpha_{t-1}$ and a threshold $\Omega$ such that $\Omega/4 < \alpha_i < \Omega/2$ and 
\begin{equation}
  \label{equ: block reconstruction proof}
  \sum_{i=0}^{t-1} \alpha_i = \rho \cdot \Omega \quad,
\end{equation}
the task is to partition the numbers into $\rho$ triples and decide if these triples can be shuffled by swapping numbers between triples so that each triple sums up to $\Omega$. This problem is strongly NP-complete \cite{book1980michael,andreev2004balanced}.

We construct a graph $G$ with cliques of size $\alpha_i$ for each integer $\alpha_i$. This graph has polynomial size $\rho \cdot \Omega$ since the integers are bounded by $\Omega/4 < \alpha_i < \Omega/2$.

\vspace{0.2cm}
\noindent\textit{\textbf{1. The block shuffling problem is NP-hard.}}

In the following, we prove that the block shuffling problem on $G$ has a solution iff the triple shuffling problem can be solved, which implies that the block shuffling problem is NP-hard. Note that $\Omega$ is the number of vertices in a block and $\rho$ is the number of blocks in the graph layout of $G$.

\noindent($\Rightarrow$): If there is a block shuffling of $G$ with the largest $OR(G)$ (let $c$ be the largest $OR(G)$), then each block contains three cliques, and the sum of their sizes is $\Omega$. We can map each clique size in a block to an integer and get a solution to the triple shuffling problem.

\noindent($\Leftarrow$): If there is a solution to the triple shuffling problem, then the integers can be grouped into triples that sum up to $\Omega$. We can map each triple to a block and each integer to a clique size in that block, resulting in a block shuffling of $G$ with $OR(G)=c$.

\vspace{0.2cm}
\noindent\textit{\textbf{2. The block shuffling problem does not have a polynomial time approximation algorithm with a finite approximation factor unless P=NP.}}

Suppose that we have a polynomial time approximation algorithm with a finite approximation factor for the block shuffling problem. We can use it to solve the triple shuffling problem. \underline{Case 1:} If there is a block shuffling of $G$ with the largest $OR(G)$, then the triple shuffling problem has a solution. \underline{Case 2:} If there is no block shuffling of $G$ with the largest $OR(G)$, then $OR(G)$ is less than $c$ and the triple shuffling problem has no solution. Any approximation algorithm with a finite approximation factor for the block shuffling problem must be able to tell these two cases apart to solve the triple shuffling problem. Since we know that checking whether it returns the largest $OR(G)$ or not is NP-hard, which contradicts the assumption.

\section{Disk-based Graph Index Method}
\label{appendix: disk graph method}
Graph-based vector search methods are more efficient and accurate than other solutions \cite{DiskANN,NSG}. Fig. \ref{fig: disk_graph} shows how they build a graph index $G=(V,E)$ for the vector dataset. Each data point (i.e., vector) is a vertex in $V$ and each edge in $E$ connects close vertices \cite{HVS,graph_survey_vldb2021}. A block is the smallest unit for disk I/O operations. The disk stores the adjacency list of each vertex with its vector data, and packs data of multiple vertices into one block. This way, one disk access can read both vector data and neighbor IDs for a vertex. In the example, each block stores data of four vertices. The entire block is loaded into the main memory when accessing data of any vertex in it.

Most graph index methods~\cite{HNSW,HVS,NSSG} use a similar search strategy, called the greedy search strategy (we call it as the vertex search strategy in this paper) \cite{NSG}. Fig.~\ref{fig: disk_graph} shows an example of finding the nearest vertex to a query vector. The search starts from a random/fixed vertex and traverses the graph. In the example, vertex \ID{12} is the entry point. Its neighbors are visited to measure their distance to the query vector. Then vertex \ID{8} is chosen as the closest visited vertex to the query vector. This process repeats iteratively. In each iteration, it finds a neighbor of the current vertex that is closer to the query vector. In the example, the search visits vertices \ID{8}, \ID{5}, \ID{2}, \ID{15} in order, until it reaches vertex \ID{3} who has no closer neighbors to the query vector than itself. In practice, a priority queue maintains candidate vertices. In each iteration, we pop one vertex from the queue, visit its neighbors, and add the neighbors that may update the current result to the queue. In the example, we set the queue length as 1 for simplicity.

The vertex search strategy has a computational cost of $O(o \cdot \ell)$, where $\ell$ is the search path length (i.e., the number of hops) and $o$ is the average out-degree of visited vertices. This is because we need to calculate the distance to the query vector for all the vertices in the search path and their neighbors. Let $\xi$ be the vertex utilization ratio and $\varepsilon$ be the number of vertices in a disk block. Then $\xi\cdot \varepsilon$ represents the number of vertices accessed in each disk I/O. The I/O complexity is then $O((o\cdot \ell)/(\xi \cdot \varepsilon))$. The current disk-based graph index method only checks the target vertex and discards the rest for a loaded block, so $\xi\cdot \varepsilon$ is 1. The average disk I/O complexity is then $O(o \cdot \ell)$, since each vertex data access requires one disk I/O operation. A recent study~\cite{DiskANN} reduced the I/O complexity by keeping the compressed vector data of all vertices in memory. It uses compressed data to calculate the approximate distance to the query before accessing a vertex. It skips vertices that are unlikely to be closer to the query without reading data from disk. However, this strategy trades accuracy for efficiency. Many disk I/Os are still needed for high accuracy. Modern storage media have relatively high performance on latency and peak bandwidth, but they are much slower than processors. Therefore, disk I/O operation is the main cost of the disk-based vector search procedure.

\begin{figure}[!tb]
  \centering
  \setlength{\abovecaptionskip}{0cm}
  \includegraphics[width=.6\linewidth]{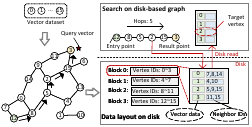}
  \caption{{Illustration of a disk-based graph index: logical topology (left), physical data layout (right).}}
%   \Description{Native hybrid query framework.}
  \label{fig: disk_graph}
%   \vspace{-0.1cm}
\end{figure}

\section{Parameters of BNF}
\label{appendix: parameters of bnf}
BNF has two hyper-parameters: maximum iterations $\beta$ and $OR(G)$ gain threshold $\tau$. They affect the $OR(G)$ of the graph layout. We evaluate their effect below. We set the maximum number of neighbors $\Lambda$ to 48 and block size $\eta$ to 4 KB. The experimental dataset is BIGANN, where each vector is 128-dimensional and one byte per value. So, each vertex takes $\gamma=(128+4+48\times4)/1,024$ KB (ID is unsigned integer type), and a block can hold at most $\varepsilon=12$ vertices. Tab. \ref{tab: blocks and edges of bnf} shows the number of blocks ($\rho$) and the number of edges of $G$ ($|E|$) on different datasets. We use 64 threads for block shuffling.

\noindent\textbf{Results.} Tab. \ref{tab: overlap ratio of bnf} and \ref{tab: time of bnf} show the $OR(G)$ and execution time results on BIGANN with different data volumes. We can see that BNF reaches a stable $OR(G)$ after a few iterations and $\tau$ is always 0.01. A smaller $\tau$ gives a higher $OR(G)$ but takes more time. A larger $\tau$ gives a lower $OR(G)$. So, we use $\tau=0.01$ by default in other experiments. Also, we observe that the $OR(G)$ increases slowly and the time increases sharply as $\beta$ increases. $\beta=8$ is enough to get a high $OR(G)$. So, we use $\beta=8$ by default. We note that the $OR(G)$ is lower and the time is higher for a larger dataset. In vector databases, each segment has about tens of millions of vectors and BNF is efficient on a data segment. We evaluate the $OR(G)$ of graph layout on a data segment based on BNF in {Appendix \ref{appendix: bnf evaluation}}.

\begin{table}[!tb]
  \centering
  \setlength{\abovecaptionskip}{0.1cm}
  \setlength{\belowcaptionskip}{0.3cm}
  \setstretch{0.8}
  \fontsize{6.5pt}{3.3mm}\selectfont
  \caption{$\rho$ and $|E|$ on BIGANN with different volumes.}
  \label{tab: blocks and edges of bnf}
  \setlength{\tabcolsep}{.015\linewidth}{
  \begin{tabular}{|l|l|l|l|l|}
    \hline
     & \textbf{BIGANN10K} & \textbf{BIGANN1M} & \textbf{BIGANN10M} & \textbf{BIGANN100M} \\
    \hline
    $\rho$ & 834 & 83,334 & 833,334 & 8,333,334 \\
    \hline
    $|E|$ & 48$\times10^4$ & 48$\times10^6$ & 48$\times10^7$ & 48$\times10^8$ \\
    \hline
  \end{tabular}
  }
\end{table}

\begin{table}[!tb]
  \centering
  \setlength{\abovecaptionskip}{0.1cm}
  \setlength{\belowcaptionskip}{0.3cm}
  \setstretch{0.8}
  \fontsize{6.5pt}{3.3mm}\selectfont
  \caption{$OR(G)$ of graph layout generated by BNF.}
  \label{tab: overlap ratio of bnf}
  \setlength{\tabcolsep}{.015\linewidth}{
  \begin{tabular}{|l|l|l|l|l|l|}
    \hline
    \textbf{$\beta\rightarrow$} & \textbf{4} & \textbf{8} & \textbf{16} & \textbf{32} & \textbf{64} \\
    \hline
    BIGANN10K & 0.4763 & 0.4979 & 0.5015 & 0.5048 & 0.5056 \\
    \hline
    BIGANN1M & 0.3312 & 0.3462 & 0.3542 & 0.3587 & 0.3612 \\
    \hline
    BIGANN10M & 0.2951 & 0.3103 & 0.3183 & 0.3230 & 0.3256 \\
    \hline
    BIGANN100M & 0.2640 & 0.2792 & 0.2876 & - & - \\
    \hline
  \end{tabular}
  }
\end{table}

\begin{table}[!tb]
  \centering
  \setlength{\abovecaptionskip}{0.1cm}
  \setlength{\belowcaptionskip}{0.3cm}
  \setstretch{0.8}
  \fontsize{6.5pt}{3.3mm}\selectfont
  \caption{Execution time (seconds) of BNF under different $\beta$.}
  \label{tab: time of bnf}
  \setlength{\tabcolsep}{.015\linewidth}{
  \begin{tabular}{|l|l|l|l|l|l|}
    \hline
    \textbf{$\beta\rightarrow$} & \textbf{4} & \textbf{8} & \textbf{16} & \textbf{32} & \textbf{64} \\
    \hline
    BIGANN10K & 0.0158 & 0.0303 & 0.0637 & 0.1393 & 0.2819 \\
    \hline
    BIGANN1M & 2.508 & 4.850 & 9.759 & 20.15 & 39.31 \\
    \hline
    BIGANN10M & 35.68 & 70.64 & 140.7 & 280.2 & 554.8 \\
    \hline
    BIGANN100M & 451.0 & 887.2 & 1760.5 & - & - \\
    \hline
  \end{tabular}
  }
\end{table}

\begin{figure}
  \vspace{0.2cm}
  \setlength{\abovecaptionskip}{0cm}
  \setlength{\belowcaptionskip}{0cm}
  \centering
  \footnotesize
  \stackunder[0.5pt]{\includegraphics[scale=0.18]{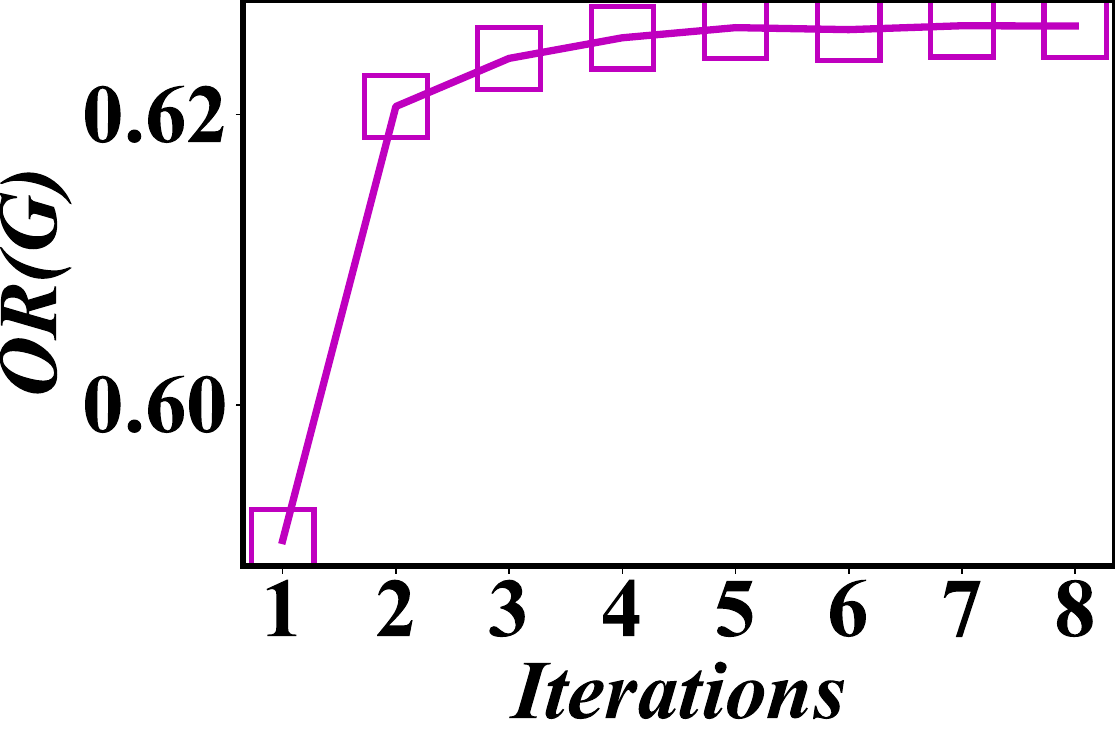}}{(a) Text2image}
  \stackunder[0.5pt]{\includegraphics[scale=0.18]{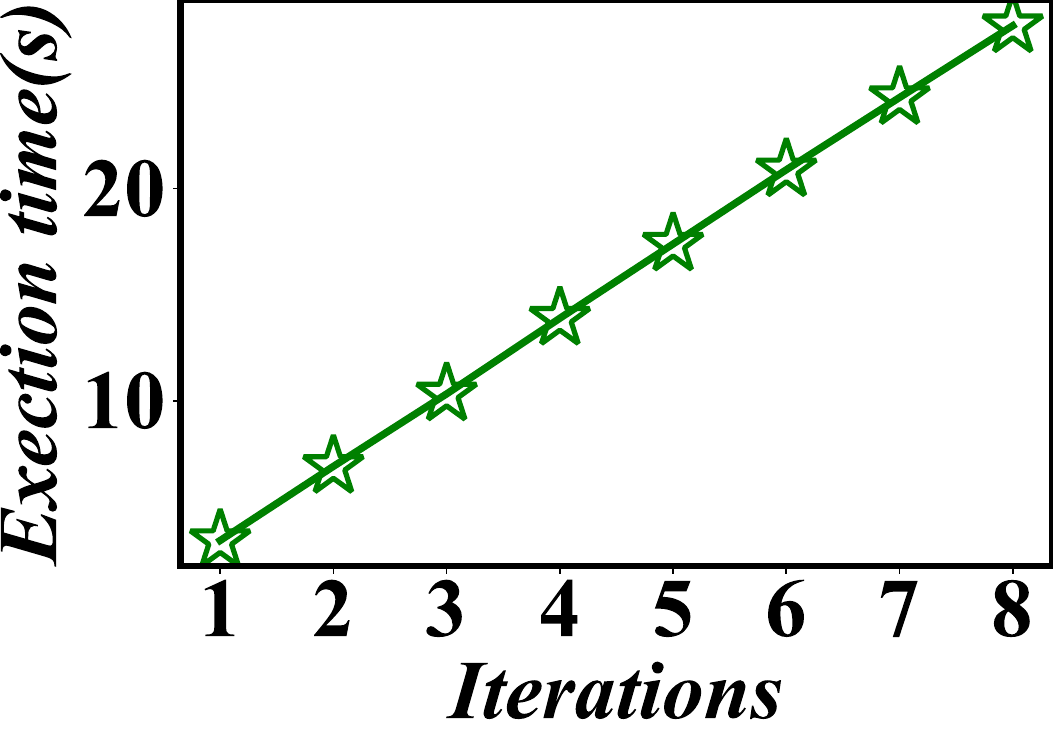}}{(b) Text2image}
  \stackunder[0.5pt]{\includegraphics[scale=0.18]{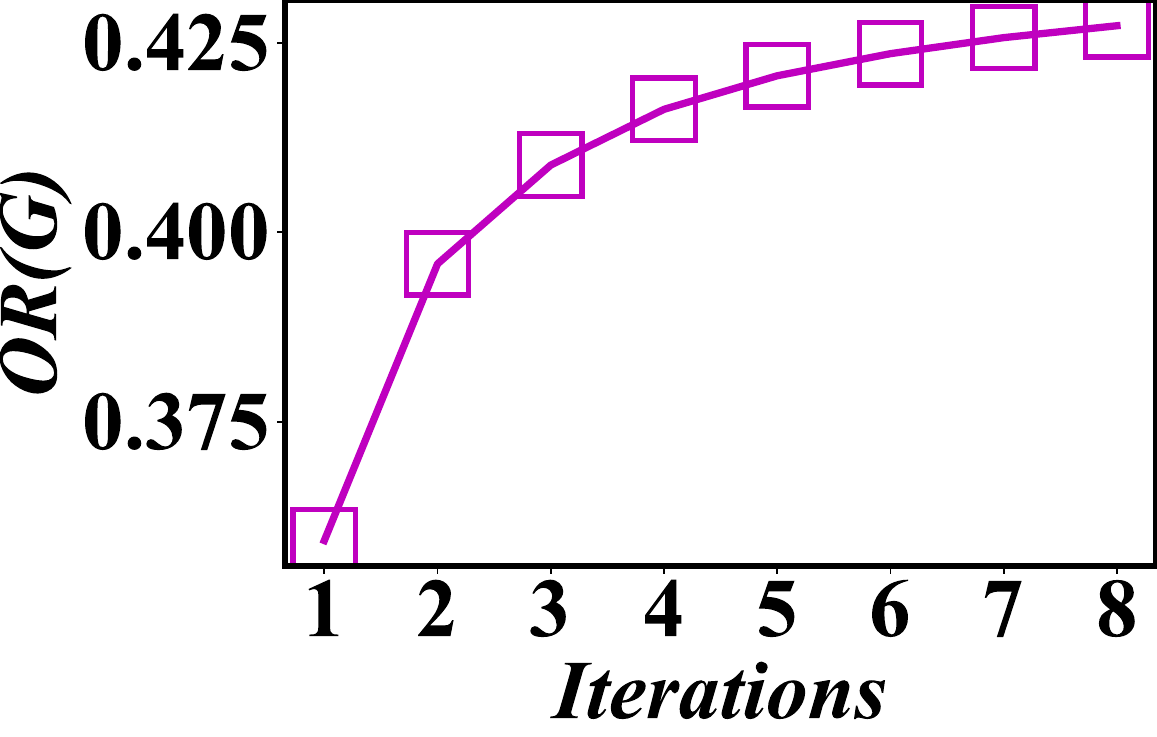}}{(c) DEEP}
  \stackunder[0.5pt]{\includegraphics[scale=0.18]{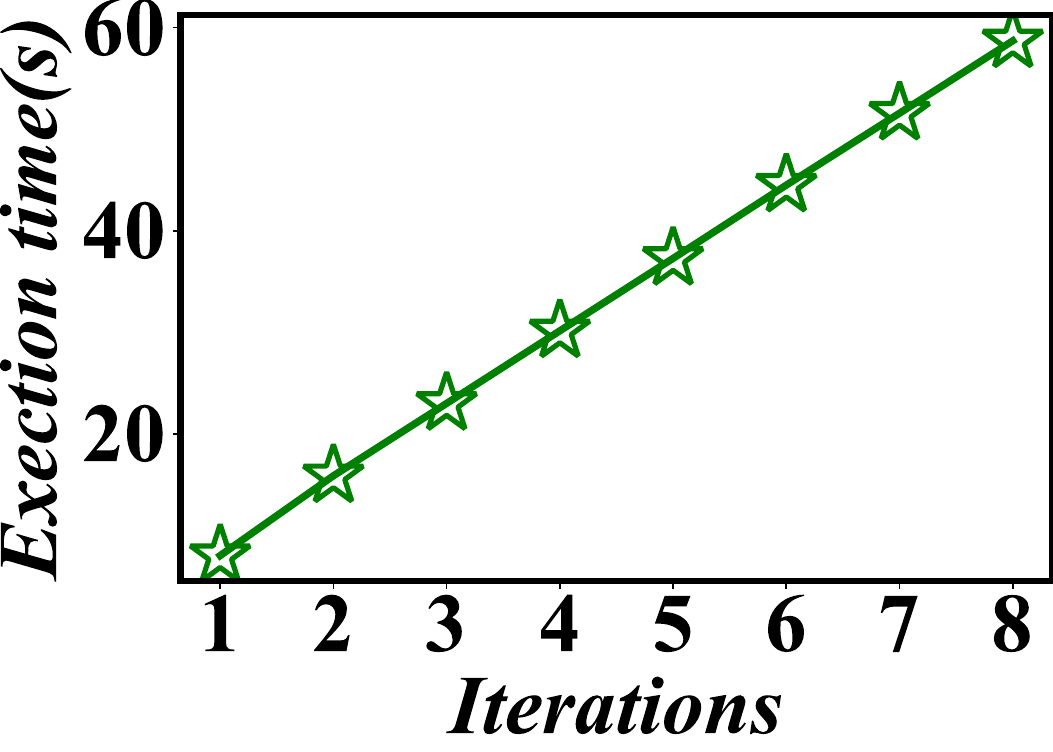}}{(d) DEEP}
  \newline
  \stackunder[0.5pt]{\includegraphics[scale=0.18]{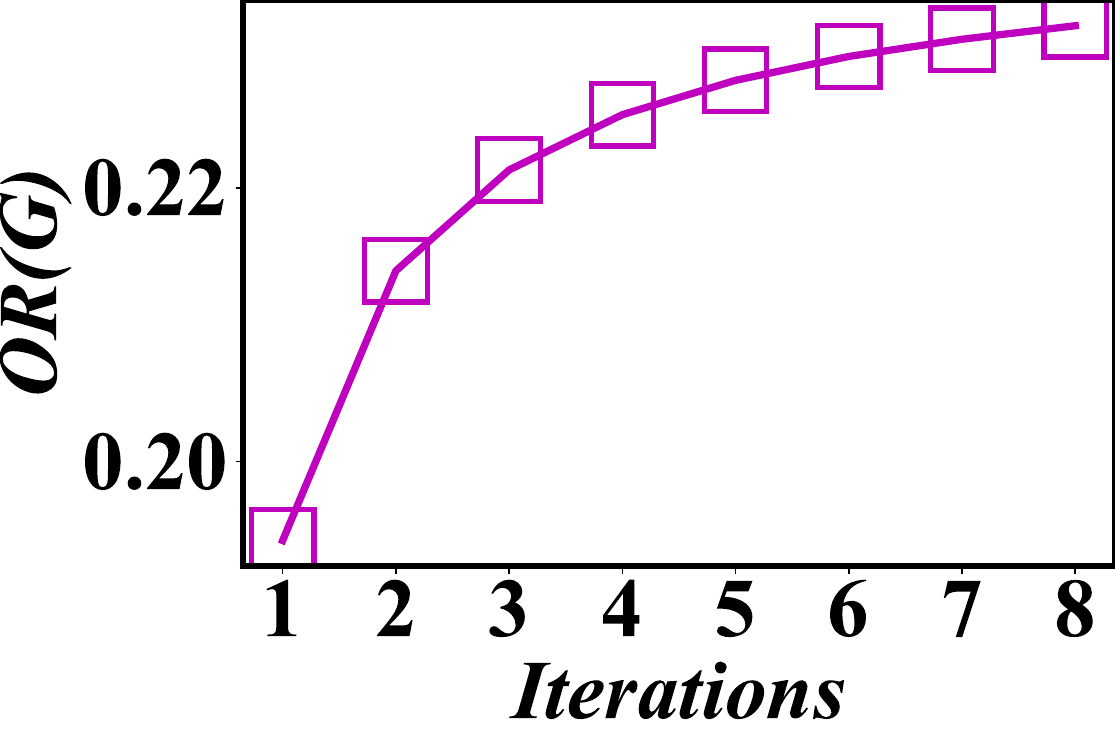}}{(e) SSNPP}
  \stackunder[0.5pt]{\includegraphics[scale=0.18]{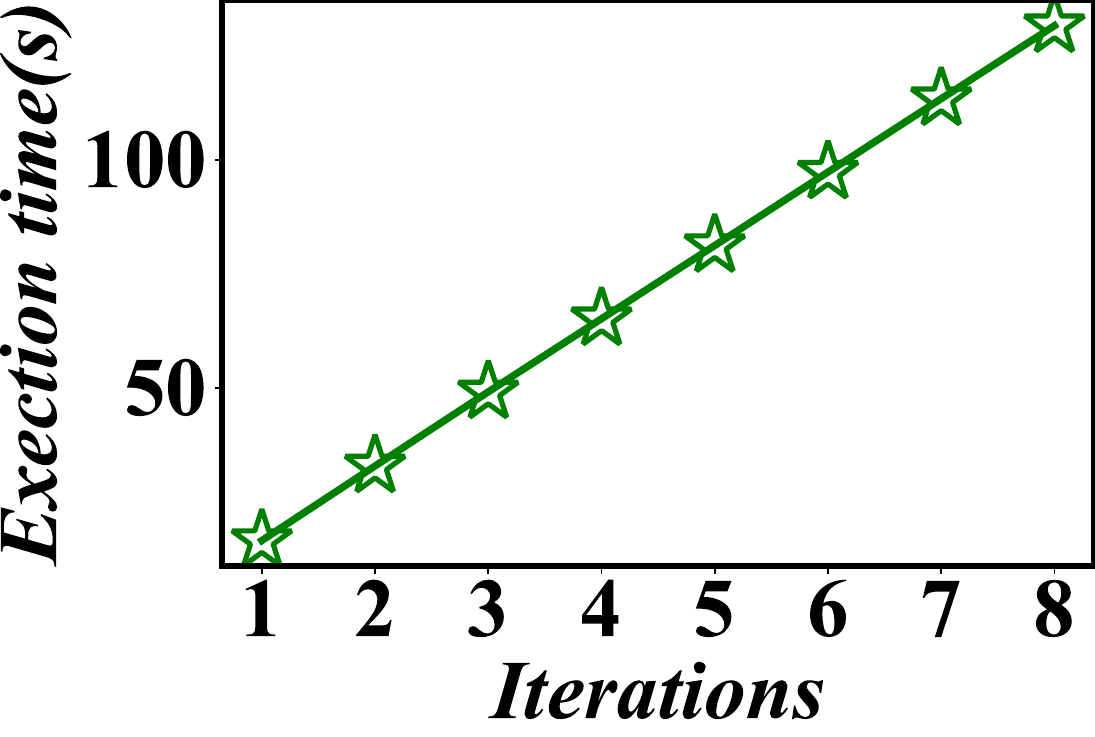}}{(f) SSNPP}
  \stackunder[0.5pt]{\includegraphics[scale=0.18]{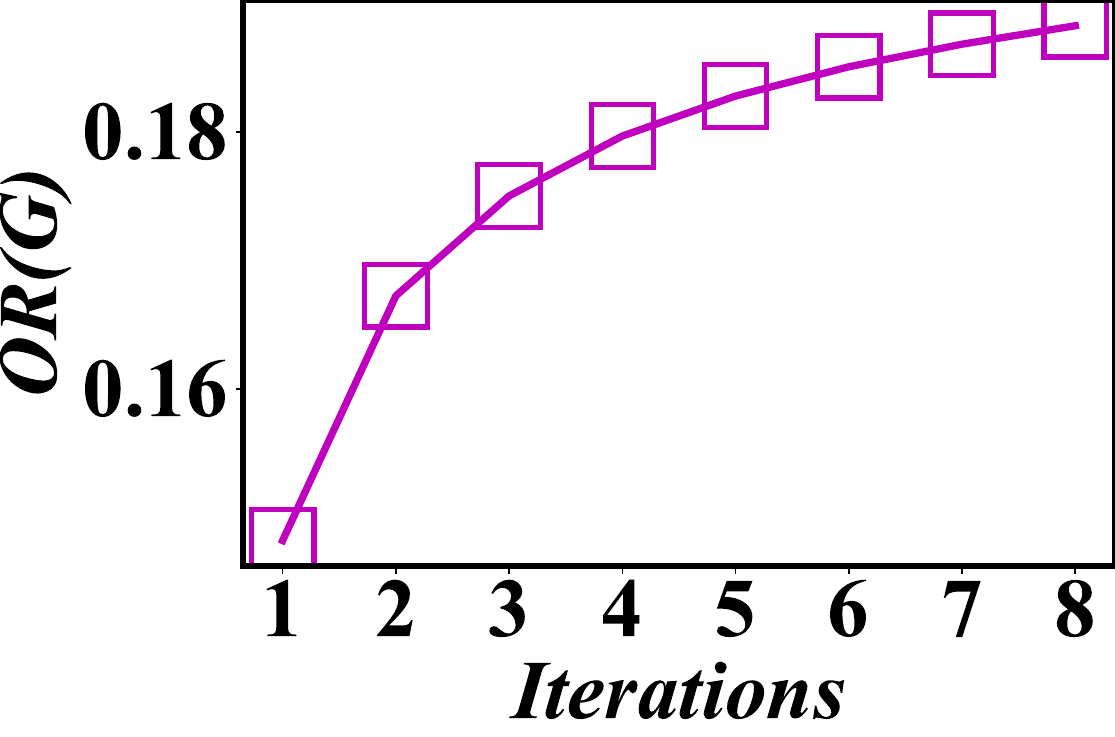}}{(g) BIGANN}
  \stackunder[0.5pt]{\includegraphics[scale=0.18]{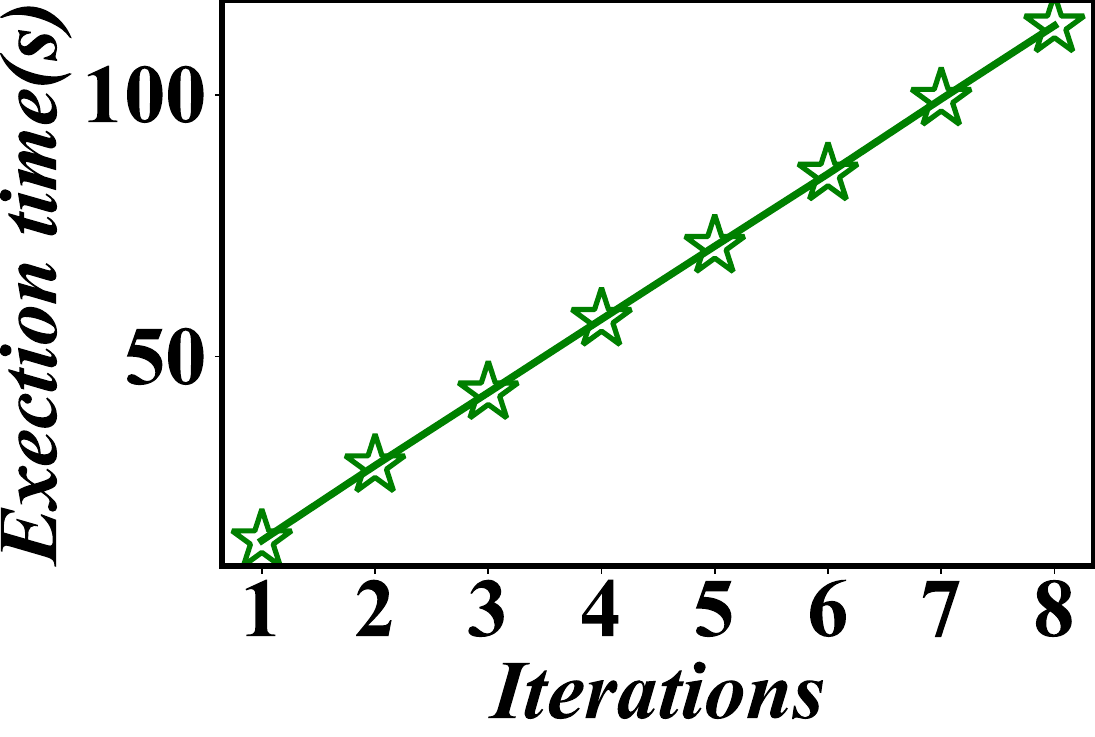}}{(h) BIGANN}
  \newline
  \caption{The effect of different iterations ($\beta$) on $OR(G)$ and execution time for a data segment on different datasets.}
  \label{fig: bnf on data segment}
\end{figure}

\section{BNF Evaluation on Data Segment}
\label{appendix: bnf evaluation}
In our experiment study, we use BNF for block shuffling by default. We show the $OR(G)$ and execution time of BNF with different iterations $\beta$ on a data segment in Fig. \ref{fig: bnf on data segment}. We use $\beta=8$ by default. See {Appendix \ref{appendix: parameters of bnf}} for parameters evaluation.

\section{BNS Algorithm}
\label{appendix: bns algorithm}
\hyperref[alg: npd]{Algorithm 2} shows the detailed procedure of BNS. First, we obtain the input graph layout based on BNP or BNF. Then, for each $u$ in $V$, we optimize the blocks where $u$'s neighbors are located by swapping the vertices with the lowest $OR$ in these blocks. This refinement can also be repeated iteratively like BNF (\hyperref[alg: npf]{Algorithm 1}), to get a better graph layout. Let $o$ be the average out-degree. For each vertex $u$ in $V$, the number of swaps is $o^2$. In each swap, the time complexity of computing one block's $OR$ is $O(o\cdot \varepsilon)$, where $\varepsilon$ is the number of vertices in a block. Therefore, the time complexity of BNS is $O(\beta \cdot o^3\cdot \varepsilon \cdot |V|)$ for $\beta$ iterations, or $O(\beta \cdot o^2\cdot \varepsilon \cdot |E|)$ ($|E|=o\cdot|V|$).

\begin{algorithm}[t]
\label{alg: npd}
  \caption{\textsc{Block Shuffling by BNS}}
  \LinesNumbered
  \KwIn{block-level graph layout of $G=(V,E)$ returned by BNP or BNF, number of iterations $\beta$, $OR(G)$ gain threshold $\tau$}
  \KwOut{new block-level graph layout of $G$}
  
  \While{iterations $\leq \beta$}{ 
    \ForAll{$u \in V$}{
      \ForAll{$a,e \in N(u)$ and $a\neq e$ }{
        \textcolor{blue}{\tcc*[f]{\textsf{vertex that has minimal $OR$ in $B(a)$}}}
      
        $x \gets \arg \min_{x\in B(a)}OR(x)$;

        \textcolor{blue}{\tcc*[f]{\textsf{vertex that has minimal $OR$ in $B(e)$}}}
        
        $y \gets \arg \min_{y\in B(e)}OR(y)$;

        \textcolor{blue}{\tcc*[f]{\textsf{$OR$ before swapping}}}
        
        $OR_{old} \gets OR(B(a))+OR(B(e))$;

        \textcolor{blue}{\tcc*[f]{\textsf{$OR$ after swapping}}}
        
        $OR_{new}\gets OR(B(a)) \setminus \{x\} \cup \{y\})+OR(B(e)) \setminus \{y\} \cup \{x\})$;
        
        \If{$OR_{old} < OR_{new}$ }{
          \textcolor{blue}{\tcc*[f]{\textsf{update $B(a)$ and $B(e)$}}}
        
          $B(a) \gets B(a) \setminus \{x\} \cup \{y\}$;
          
          $B(e) \gets B(e) \setminus \{y\} \cup \{x\}$;
        }
        
      }
    }
    % $OR(G)_e$ $\gets$ overlap ratio of $G$ after an iteration;
    
    \If{$OR(G)$ gain $< \tau$}{
      break;
    }
  }
  \textbf{return} new layout of $G$
\end{algorithm}

\section{BNF vs. BNS}
\label{appendix: bnf vs bns}

\begin{table}[!tb]
  \centering
  \setlength{\abovecaptionskip}{0.1cm}
  \setlength{\belowcaptionskip}{0.3cm}
  \setstretch{0.8}
  \fontsize{6.5pt}{3.3mm}\selectfont
  \caption{BNF vs. BNS on BIGANN1M.}
  \label{tab: bnf vs bns}
  \setlength{\tabcolsep}{.015\linewidth}{
  \begin{tabular}{|l|l|l|l|}
    \hline
    \textbf{Algorithm$\downarrow$} & $\beta$ & Execution time (s) for each iteration & $OR(G)$ \\
    \hline
    BNF & 64 & 0.627 & 0.361 \\
    \hline
    BNS & 6 & 1,877 & 0.401 \\
    \hline
  \end{tabular}
  }\vspace{0.2cm}
\end{table}

Tab. \ref{tab: bnf vs bns} shows the block shuffling performance of BNF and BNS on BIGANN1M (which contains one million vectors in 128 dimensions). We use 64 threads for block shuffling. The results show that BNS achieves better $OR(G)$ with fewer iterations, but it is much slower than BNF in each iteration. This makes BNS impractical on larger datasets, unless we can speed up BNS. We leave this as future work. In our experiments, we use BNF for block shuffling by default.

\section{Graph Partitioning Evaluation}
\label{appendix: graph partitioning evaluation}
We evaluate how three existing graph partitioning methods perform on our block shuffling task for disk-based graph index. We summarize the three methods below.
\begin{itemize}[leftmargin=*]
  \item \textbf{GP1 \cite{SPANN}.} This method partitions $V$ by hierarchical balanced clustering. It iteratively splits the vertices in a large cluster into smaller clusters until each cluster has a limit number of vertices.
  \item \textbf{GP2 \cite{predari2016k}.} This method is based on K-way greedy graph growing algorithm (KGGGP). The main idea is a standard greedy approach for bipartitioning. It starts with two random ``seed'' vertices in the two parts. Then, it adds vertices alternately to the parts, choosing the vertex that minimizes a selected criterion.
  \item \textbf{GP3 \cite{awadelkarim2020prioritized}.} This method is from prioritized restreaming algorithms for balanced graph partitioning. It can partition graphs into balanced subsets, which is important for many distributed computing applications. It also prioritizes the gain order of nodes.
\end{itemize}

\begin{table}[!tb]
  \centering
  \setlength{\abovecaptionskip}{0.1cm}
  \setlength{\belowcaptionskip}{0.3cm}
  \setstretch{0.8}
  \fontsize{6.5pt}{3.3mm}\selectfont
  \caption{GP1 vs. BNF on BIGANN10K.}
  \label{tab: gp1 vs bns bigann10k}
  \setlength{\tabcolsep}{.015\linewidth}{
  \begin{tabular}{|l|l|l|l|}
    \hline
    \textbf{Algorithm$\downarrow$} & $\beta$ & Total execution time (s) & $OR(G)$ \\
    \hline
    GP1 & - & 6.8 & 0.3130 \\
    \hline
    BNF & 64 & 0.2819 & 0.5056 \\
    \hline
  \end{tabular}
  }
\end{table}

\begin{table}[!tb]
  \centering
  \setlength{\abovecaptionskip}{0.1cm}
  \setlength{\belowcaptionskip}{0.3cm}
  \setstretch{0.8}
  \fontsize{6.5pt}{3.3mm}\selectfont
  \caption{GP1 vs. BNF on BIGANN1M.}
  \label{tab: gp1 vs bns bigann1m}
  \setlength{\tabcolsep}{.015\linewidth}{
  \begin{tabular}{|l|l|l|l|}
    \hline
    \textbf{Algorithm$\downarrow$} & $\beta$ & Total execution time (s) & $OR(G)$ \\
    \hline
    GP1 & - & 192.24 & 0.2060 \\
    \hline
    BNF & 64 & 39.31 & 0.3612 \\
    \hline
  \end{tabular}
  }
\end{table}

\begin{table}[!tb]
  \centering
  \setlength{\abovecaptionskip}{0.1cm}
  \setlength{\belowcaptionskip}{0.3cm}
  \setstretch{0.8}
  \fontsize{6.5pt}{3.3mm}\selectfont
  \caption{GP1 vs. BNF on SSNPP10M.}
  \label{tab: gp1 vs bns ssnpp10m}
  \setlength{\tabcolsep}{.015\linewidth}{
  \begin{tabular}{|l|l|l|l|}
    \hline
    \textbf{Algorithm$\downarrow$} & $\beta$ & Total execution time (s) & $OR(G)$ \\
    \hline
    GP1 & - & 234.0 & 0.0200 \\
    \hline
    BNF & 16 & 200.3 & 0.2374 \\
    \hline
  \end{tabular}
  }
\end{table}

\begin{table}[!tb]
  \centering
  \setlength{\abovecaptionskip}{0.1cm}
  \setlength{\belowcaptionskip}{0.3cm}
  \setstretch{0.8}
  \fontsize{6.5pt}{3.3mm}\selectfont
  \caption{GP2 vs. BNF on DEEP10M.}
  \label{tab: gp2 vs bns deep10m}
  \setlength{\tabcolsep}{.015\linewidth}{
  \begin{tabular}{|l|l|l|l|}
    \hline
    \textbf{Algorithm$\downarrow$} & $\beta$ & Total execution time (s) & $OR(G)$ \\
    \hline
    GP2 & - & 300.0 & 0.3100 \\
    \hline
    BNF & 64 & 306.8 & 0.4290 \\
    \hline
  \end{tabular}
  }
\end{table}

\begin{table}[!tb]
  \centering
  \setlength{\abovecaptionskip}{0.1cm}
  \setlength{\belowcaptionskip}{0.3cm}
  \setstretch{0.8}
  \fontsize{6.5pt}{3.3mm}\selectfont
  \caption{GP3 vs. BNF on BIGANN10M.}
  \label{tab: gp3 vs bns bigann10m}
  \setlength{\tabcolsep}{.015\linewidth}{
  \begin{tabular}{|l|l|l|}
    \hline
    \textbf{Algorithm$\downarrow$} & $\beta$ & $OR(G)$ \\
    \hline
    GP3 & 4 & 0.2666 \\
    \hline
    GP3 & 8 & 0.2801 \\
    \hline
    BNF & 4 & 0.2951 \\
    \hline
    BNF & 8 & 0.3103 \\
    \hline
  \end{tabular}
  }
\end{table}

\noindent\textbf{Results.} We compare the three graph partitioning methods with BNF using 64 threads to perform all methods. Tab. \ref{tab: gp1 vs bns bigann10k}–\ref{tab: gp3 vs bns bigann10m} show the efficiency and effectiveness results. GP1 and GP2 are not iterative algorithms, so they do not have $\beta$ data. We implement GP3 by adding the gain order \cite{awadelkarim2020prioritized} to BNF, so the time of GP3 is slower than BNF (it is at least equal to BNF). The results show that BNF outperforms the three graph partitioning methods in time and $OR(G)$ over different datasets. Therefore, existing graph partitioning methods are not suitable for our block shuffling task. We analyze the reasons below. Graph partitioning usually targets real-world social networks \cite{pacaci2019experimental,WeiYLL16}, which have natural clustering and power-law degree distribution \cite{abbas2018streaming}. But block shuffling handles proximity graph index built on high-dimensional vectors, where the edge needs both navigation and similarity \cite{DPG,NSSG} and follows a uniform degree distribution \cite{graph_survey_vldb2021}.

\setlength{\textfloatsep}{0cm}
\setlength{\floatsep}{0cm}
\begin{table}[!tb]
  \centering
  \setlength{\abovecaptionskip}{0.05cm}
  \setlength{\belowcaptionskip}{0.3cm}
  \setstretch{0.8}
  \fontsize{6.5pt}{3.3mm}\selectfont
  \caption{$T_{disk\_graph}$ and $T_{shuffling}$ on four datasets.}
  \label{tab: index construction}
  \setlength{\tabcolsep}{.02\linewidth}{
  \begin{tabular}{|l|l|l|}
    \hline
    Dataset$\downarrow$ & $T_{disk\_graph}$ (s) & $T_{shuffling}$ (s) \\
    \hline
    \textbf{BIGANN} & 930 & 113 \\
    \hline
    \textbf{DEEP} & 668 & 59 \\
    \hline
    \textbf{SSNPP} & 1,055 & 130 \\
    \hline
    \textbf{Text2image} & 519 & 28 \\
    \hline
  \end{tabular}
  }\vspace{0.1cm}
\end{table}

\section{Block Shuffling Cost}
\label{appendix: block shuffling cost}
We report the disk-based graph index construction time ($T_{disk\_graph}$ in \textbf{\S \ref{subsec: index cost}}) and block shuffling time ($T_{shuffling}$ in \textbf{\S \ref{subsec: index cost}}) in Tab. \ref{tab: index construction}. From the results, we can see that $T_{shuffling}$ is only 3\%$\sim$12\% of $T_{disk\_graph}$.

\section{Parameters of In-Memory Navigation Graph}
\label{appendix: parameters of memory graph}
In in-memory navigation graph, we have parameter of sample ratio $\mu$ and parameters of graph index. For graph index parameters, we follow the current proximity graph algorithms \cite{graph_survey_vldb2021} and consider the memory limit of a segment. For example, given a sample ratio $\mu$, we adjust the maximum number of neighbors in in-memory navigation graph to keep the memory overhead less than 2 GB on a segment. We analyze the effect of $\mu$ on search performance with the same graph index parameters below.

\begin{table}[!tb]
  \centering
  \setlength{\abovecaptionskip}{0.1cm}
  \setlength{\belowcaptionskip}{0.3cm}
  \setstretch{0.8}
  \fontsize{6.5pt}{3.3mm}\selectfont
  \caption{Search performance and memory overhead under different $\mu$ values.}
  \label{tab: sample ratio memory index}
  \setlength{\tabcolsep}{.015\linewidth}{
  \begin{tabular}{|l|l|l|l|}
    \hline
    \textbf{Sample ratio} & \textbf{\textit{Recall}($k=10$)} & \textbf{\textit{QPS}} & \textbf{Memory overhead (MB)} \\
    \hline
    $\mu=0.001$ & 0.9916 & 3,914 & 521 \\
    \hline
    $\mu=0.01$ & 0.9917 & 3,991 & 558 \\
    \hline
    $\mu=0.1$ & 0.9919 & 4,303 & 918 \\
    \hline
  \end{tabular}
  }\vspace{0.2cm}
\end{table}

\noindent\textbf{Results.} Tab. \ref{tab: sample ratio memory index} shows the search performance and memory overhead with different $\mu$ values. We can see that the \textit{Recall} and \textit{QPS} increase with $\mu$. This is because more data in memory may provide better entry points that are closer to query vector. These entry points can reduce disk I/Os and increase the chance of finding true nearest vectors. But a larger sample ratio also increases the memory overhead. We need to balance search performance and memory overhead on a segment due to the limited memory capacity.

\section{In-Memory Navigation Graph vs. Caching Hot Vertices}
\label{appendix: memory graph vs cache}

\begin{figure}
%   \vspace{0.2cm}
  \setlength{\abovecaptionskip}{0cm}
  \setlength{\belowcaptionskip}{0cm}
  \centering
  \footnotesize
  \stackunder[0.5pt]{\includegraphics[scale=0.21]{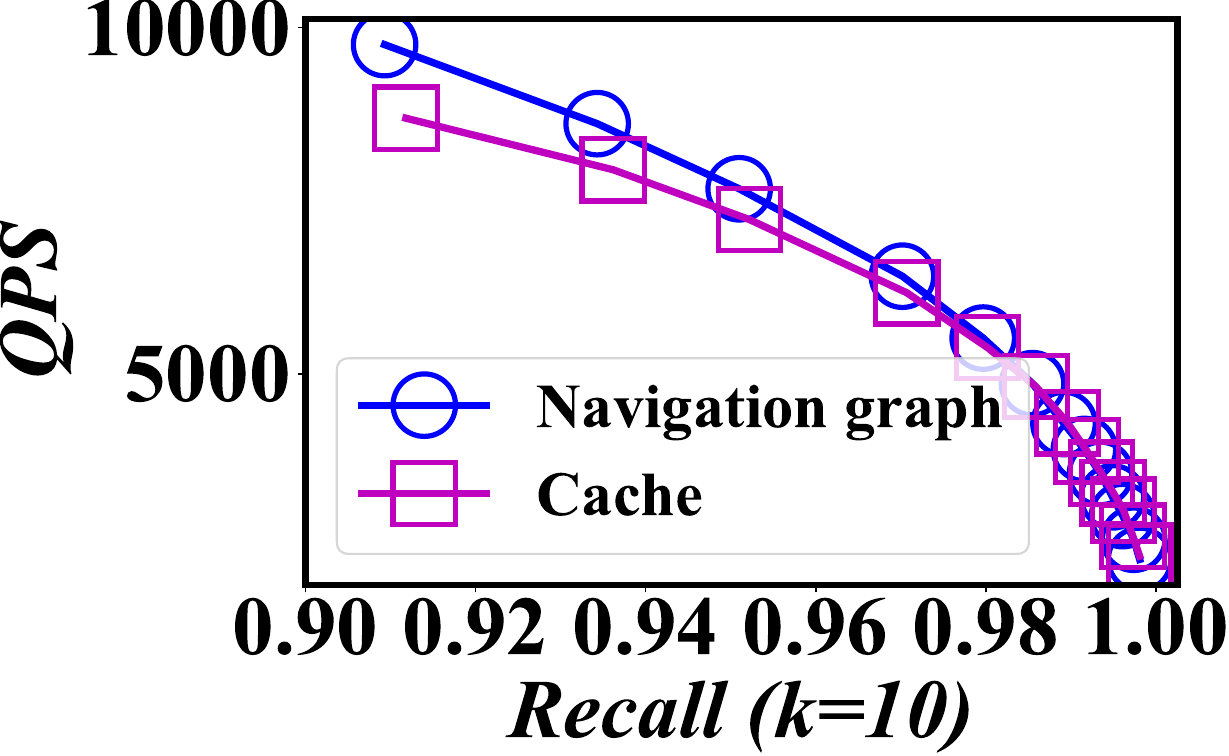}}{(a) $\mu=0.001$}
  \stackunder[0.5pt]{\includegraphics[scale=0.21]{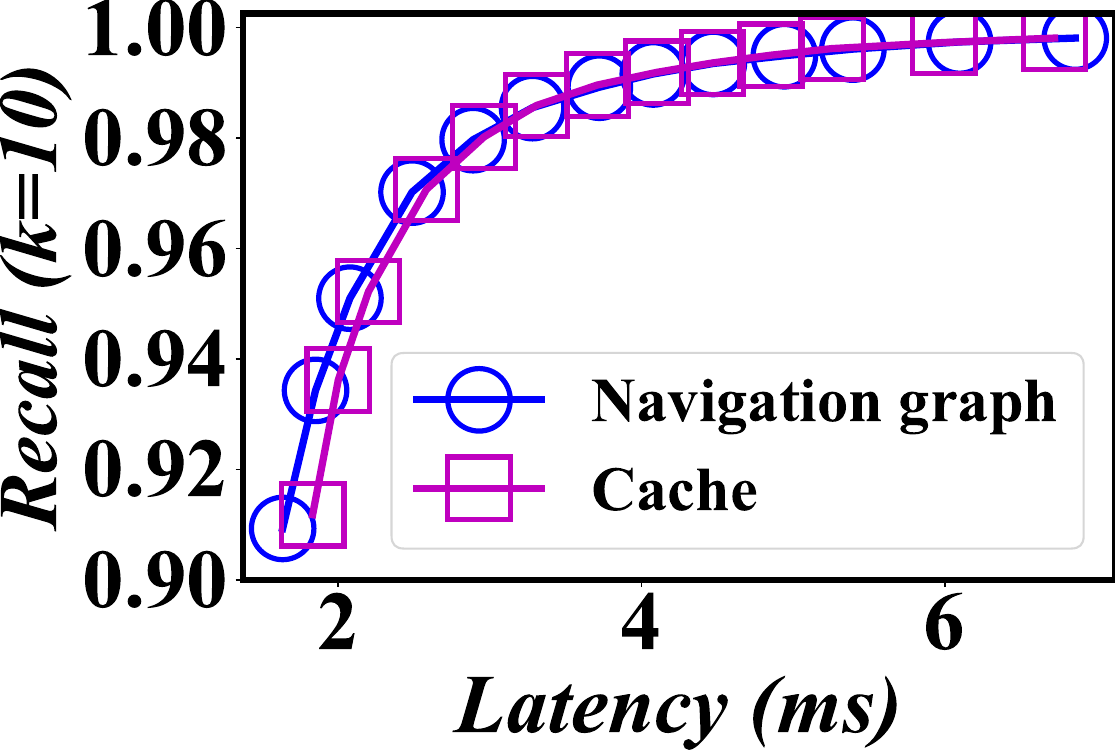}}{(b) $\mu=0.001$}
  \newline
  \stackunder[0.5pt]{\includegraphics[scale=0.21]{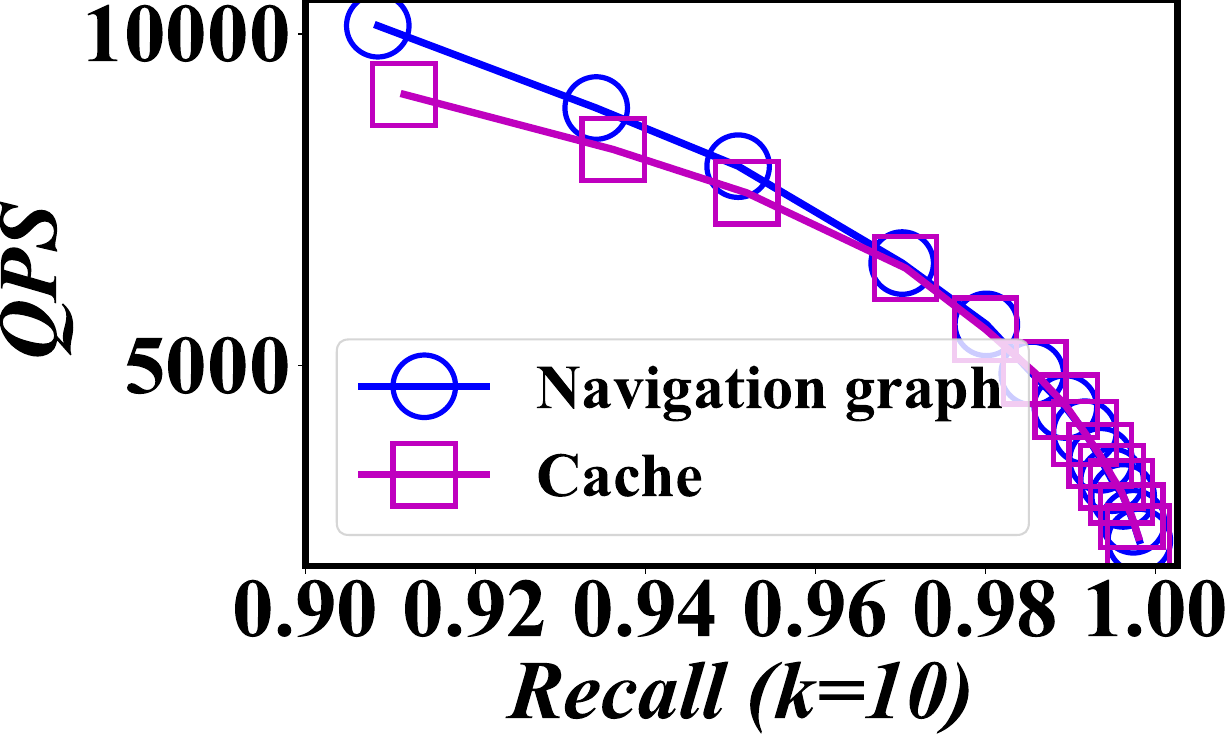}}{(a) $\mu=0.01$}
  \stackunder[0.5pt]{\includegraphics[scale=0.21]{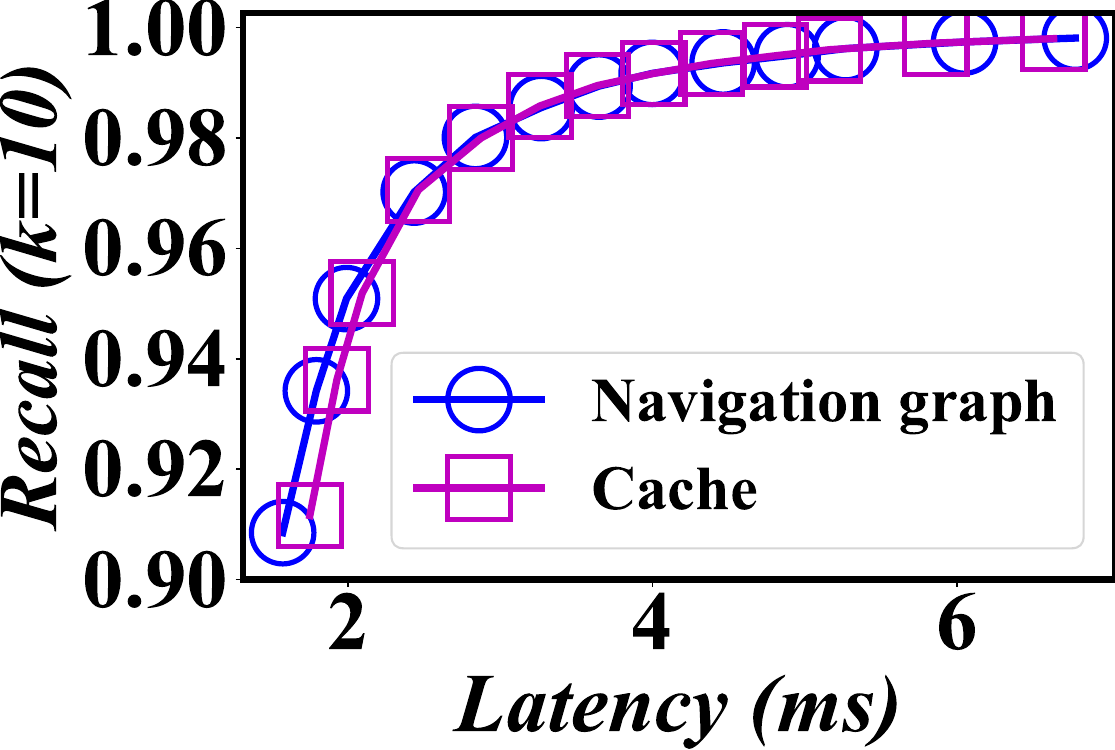}}{(b) $\mu=0.01$}
  \newline
  \stackunder[0.5pt]{\includegraphics[scale=0.21]{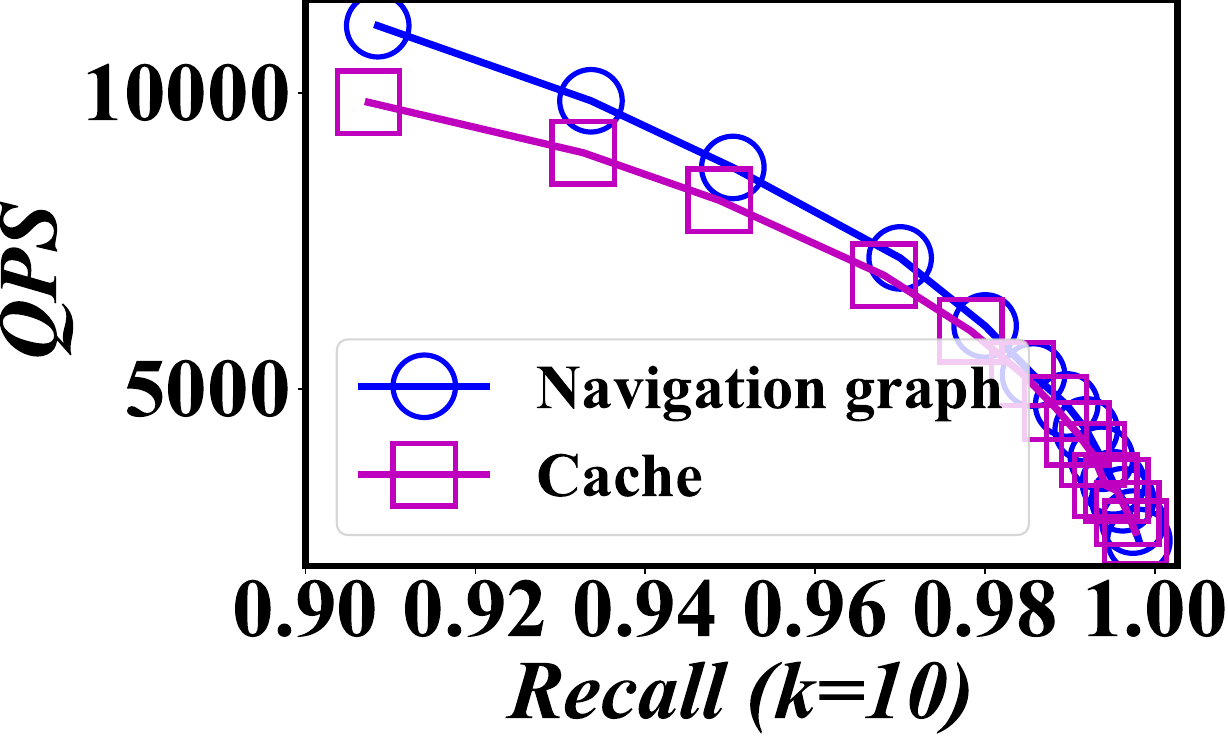}}{(a) $\mu=0.1$}
  \stackunder[0.5pt]{\includegraphics[scale=0.21]{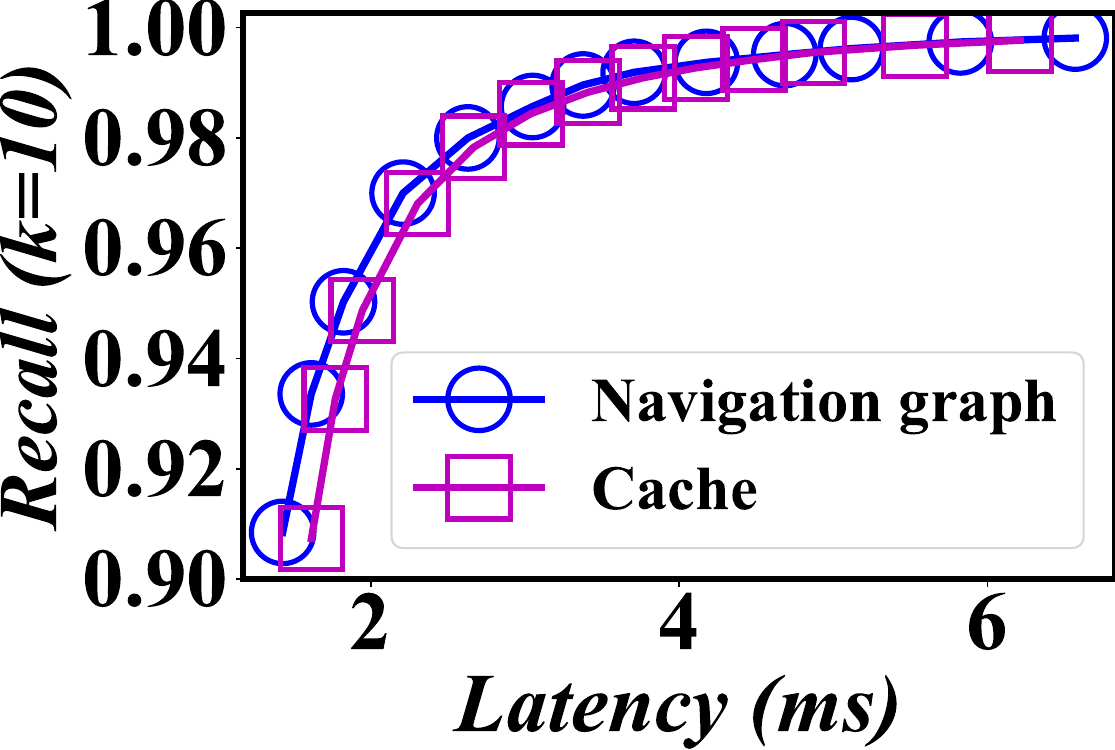}}{(b) $\mu=0.1$}
  \newline
  \caption{Search performance comparison under different $\mu$ values for two data layouts in memory.}
  \label{fig: memory graph vs cache, qps and latency}
\end{figure}

DiskANN \cite{DiskANN} (the baseline framework) caches some hot vertices and their neighbor IDs of disk-based graph in memory. We evaluate the effect of in-memory navigation graph on the BIGANN10M dataset, compared to DiskANN's cache strategy. We implement DiskANN's cache strategy on {\name} by replacing the in-memory navigation graph with it. We keep the sample ratio $\mu$ of two in-memory data layouts the same. We use 64 threads for the search procedure.

\noindent\textbf{Efficiency and accuracy.} Fig. \ref{fig: memory graph vs cache, qps and latency} shows the \textit{QPS} vs. \textit{Recall} and \textit{Recall} vs. \textit{Latency} comparison with different sample ratios $\mu$ on BIGANN10M. We can see that the in-memory navigation graph is always better than the cache strategy, especially when \textit{Recall} is between 0.90 and 0.97. We also see that when \textit{Recall} is above 0.97, two strategies are similar. The reason is that a higher \textit{Recall} needs more disk I/Os, and the I/Os reduction with the cache data or in-memory navigation graph is small compared to the total I/Os. So, the difference between the cache data and in-memory navigation graph becomes insignificant.

\begin{table}[!tb]
  \centering
  \setlength{\abovecaptionskip}{0.1cm}
  \setlength{\belowcaptionskip}{0.3cm}
  \setstretch{0.8}
  \fontsize{6.5pt}{3.3mm}\selectfont
  \caption{Memory overhead (MB) under different $\mu$ values for two data layouts in memory.}
  \label{tab: memory overhead two data layouts}
  \setlength{\tabcolsep}{.015\linewidth}{
  \begin{tabular}{|l|l|l|}
    \hline
    \textbf{Strategy$\rightarrow$} & \textbf{Navigation graph} & \textbf{Cache} \\
    \hline
    $\mu=0.001$ & 521 & 661 \\
    \hline
    $\mu=0.01$ & 558 & 1,039 \\
    \hline
    $\mu=0.1$ & 918 & 5,122 \\
    \hline
  \end{tabular}
  }\vspace{0.3cm}
\end{table}

\noindent\textbf{Memory overhead.} Tab. \ref{tab: memory overhead two data layouts} shows the memory overheads of two in-memory data layouts with the same \textit{Recall} (\textit{Recall=0.99}). We can see that the in-memory navigation graph has a lower memory cost than the cache strategy of DiskANN. This is because our in-memory graph keeps a light-weight graph index built on sample data in memory, but the cache strategy loads the vector data of vertices and their neighbor IDs of disk-based graph (they are always larger than in-memory graph index) into memory. We note that the memory overhead gap between navigation graph and cache strategy is larger as the sample ratio increases.

\section{Pruning ratio for Block Pruning}
\label{appendix: block pruning evaluation}
Fig. \ref{fig: block pruning evaluation} shows how different pruning ratios $\sigma$ affect the search performance on BIGANN10M. We use 8 threads to execute a batch of queries. The results show that increasing $\sigma$ from 0 to 0.3 improves the \textit{QPS} under the same \textit{Recall}. This is because a higher $\sigma$ allows us to check more neighbors and reach the vertices closer to the query with a high probability. This exploits the data locality and reduces the disk I/Os. When $\sigma=0$, we only visit the target vertex, which is equivalent to the vertex search strategy of the baseline framework. This does not benefit from the data locality from block shuffling, so it has the worst performance. However, when $\sigma$ exceeds a threshold (from 0.3 to 0.5), the \textit{QPS} under the same \textit{Recall} declines. This is because a larger $\sigma$ causes redundant computation. The mean I/Os decrease as $\sigma$ increases, because visiting more neighbors helps us find a vertex closer to the query and reduce disk I/Os. Therefore, there is an optimal pruning ratio that we can obtain by grid search. For example, the optimal $\sigma$ for the BIGANN dataset is 0.3. Tab. \ref{tab: parameters for search pruning ratio} shows the optimal $\sigma$ for other datasets.

\begin{figure}
  % \vspace{0.2cm}
  \setlength{\abovecaptionskip}{0cm}
  \setlength{\belowcaptionskip}{0cm}
  \centering
  \footnotesize
  \stackunder[0.5pt]{\includegraphics[scale=0.21]{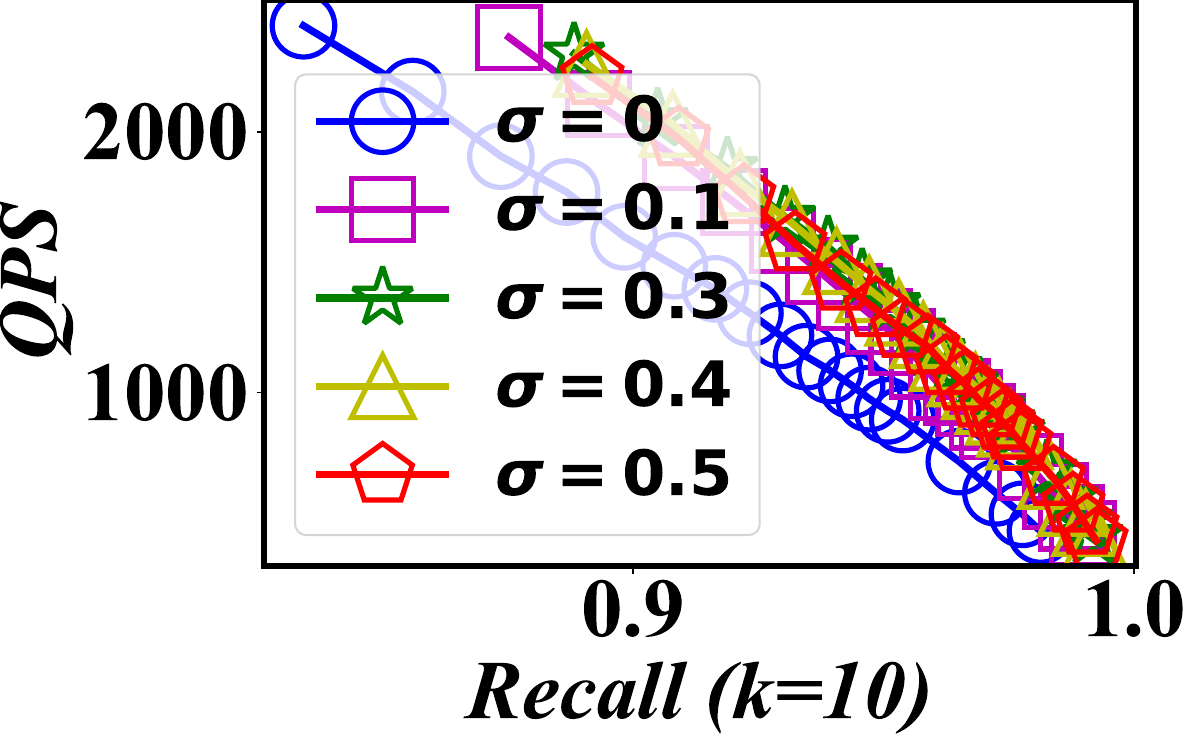}}{(a) \textit{QPS}}
  \stackunder[0.5pt]{\includegraphics[scale=0.21]{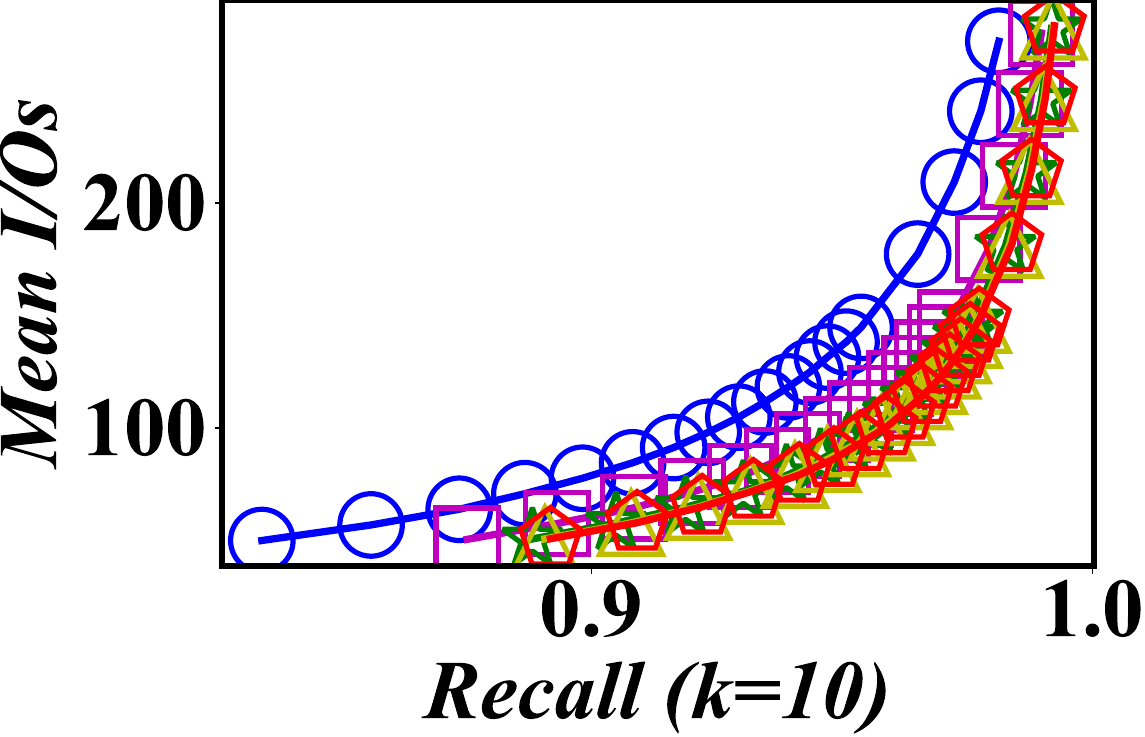}}{(b) \textit{Mean I/Os}}
  \newline
  \caption{\textit{QPS} and \textit{Mean I/Os} under different pruning ratios.}
  \label{fig: block pruning evaluation}
\end{figure}

\begin{table}[!tb]
  \centering
  \setlength{\abovecaptionskip}{0.1cm}
  \setlength{\belowcaptionskip}{0.3cm}
  \setstretch{0.8}
  \fontsize{6.5pt}{3.3mm}\selectfont
  \caption{Parameter values of the disk-based graph used in our experiments.}
  \label{tab: parameters for disk-based graph}
  \setlength{\tabcolsep}{.015\linewidth}{
  \begin{tabular}{|l|l|l|l|l|}
    \hline
    \textbf{Parameter$\downarrow$} & \textbf{BIGANN} & \textbf{DEEP} & \textbf{SSNPP} & \textbf{Text2image} \\
    \hline
    $\Lambda$ & 31 & 48 & 48 & 54 \\
    \hline
    $L$ & 128 & 128 & 128 & 128 \\
    \hline
    $B$ & 0.500 & 0.256 & 0.960 & 0.470 \\
    \hline
    $\gamma$ & 0.250 & 0.566 & 0.441 & 0.996 \\
    \hline
    $\eta$ & 4 & 4 & 4 & 4 \\
    \hline
    $\varepsilon$ & 16 & 7 & 9 & 4 \\
    \hline
    $\rho$ & 2,062,500 & 1,571,429 & 1,777,778 & 1,250,000 \\
    \hline
  \end{tabular}
  }
\end{table}

\begin{table}[!tb]
  \centering
  \setlength{\abovecaptionskip}{0.1cm}
  \setlength{\belowcaptionskip}{0.3cm}
  \setstretch{0.8}
  \fontsize{6.5pt}{3.3mm}\selectfont
  \caption{Parameter values of in-memory navigation graph used in our experiments.}
  \label{tab: parameters for memory graph}
  \setlength{\tabcolsep}{.015\linewidth}{
  \begin{tabular}{|l|l|l|l|l|}
    \hline
    \textbf{Parameter$\downarrow$} & \textbf{BIGANN} & \textbf{DEEP} & \textbf{SSNPP} & \textbf{Text2image} \\
    \hline
    $\Lambda$ & 20 & 10 & 48 & 100 \\
    \hline
    $L$ & 128 & 128 & 128 & 128  \\
    \hline
    $\mu$ & 0.09 & 0.10 & 0.10 & 0.09 \\
    \hline
  \end{tabular}
  }
\end{table}

\section{Parameters for Index Process}
\label{appendix: parameters index process}
Please refer to the DiskANN paper \cite{DiskANN} for details on disk-based graph index construction. We explain the main parameters for building the index. $\Lambda$ is the maximum degree of the disk-based graph. A higher $\Lambda$ leads to larger indexes and longer construction time, but also better search accuracy. $L$ is the size of candidate neighbors during index construction. A larger $L$ takes more time to construct but produces indexes with higher \textit{Recall} for the same search efficiency. $L$ should be at least as large as $\Lambda$. $B$ is the limit on the memory footprint of PQ short codes. This determines how much we compress the vector data to maintain in memory. A smaller $B$ reduces the memory overhead for the PQ short codes. Tab. \ref{tab: parameters for disk-based graph} shows the parameter values we used in our experiments (see {\S \ref{subsec: graph_reorder}} for more parameter description). In-memory navigation graph uses the same proximity graph, but it has smaller $\Lambda$ and $L$. Tab. \ref{tab: parameters for memory graph} provides the parameter values of in-memory navigation graph.

\begin{table}[!tb]
  \centering
  \setlength{\abovecaptionskip}{0.1cm}
  \setlength{\belowcaptionskip}{0.3cm}
  \setstretch{0.8}
  \fontsize{6.5pt}{3.3mm}\selectfont
  \caption{Pruning ratio used in our experiments.}
  \label{tab: parameters for search pruning ratio}
  \setlength{\tabcolsep}{.015\linewidth}{
  \begin{tabular}{|l|l|l|l|l|}
    \hline
    \textbf{Parameter$\downarrow$} & \textbf{BIGANN} & \textbf{DEEP} & \textbf{SSNPP} & \textbf{Text2image} \\
    \hline
    $\sigma$ & 0.3 & 0.3 & 0.3 & 1.0 \\
    \hline
  \end{tabular}
  }\vspace{0.3cm}
\end{table}

\begin{figure}
%   \vspace{0.2cm}
  \setlength{\abovecaptionskip}{0cm}
  \setlength{\belowcaptionskip}{0cm}
  \centering
  \footnotesize
  \stackunder[0.5pt]{\includegraphics[scale=0.21]{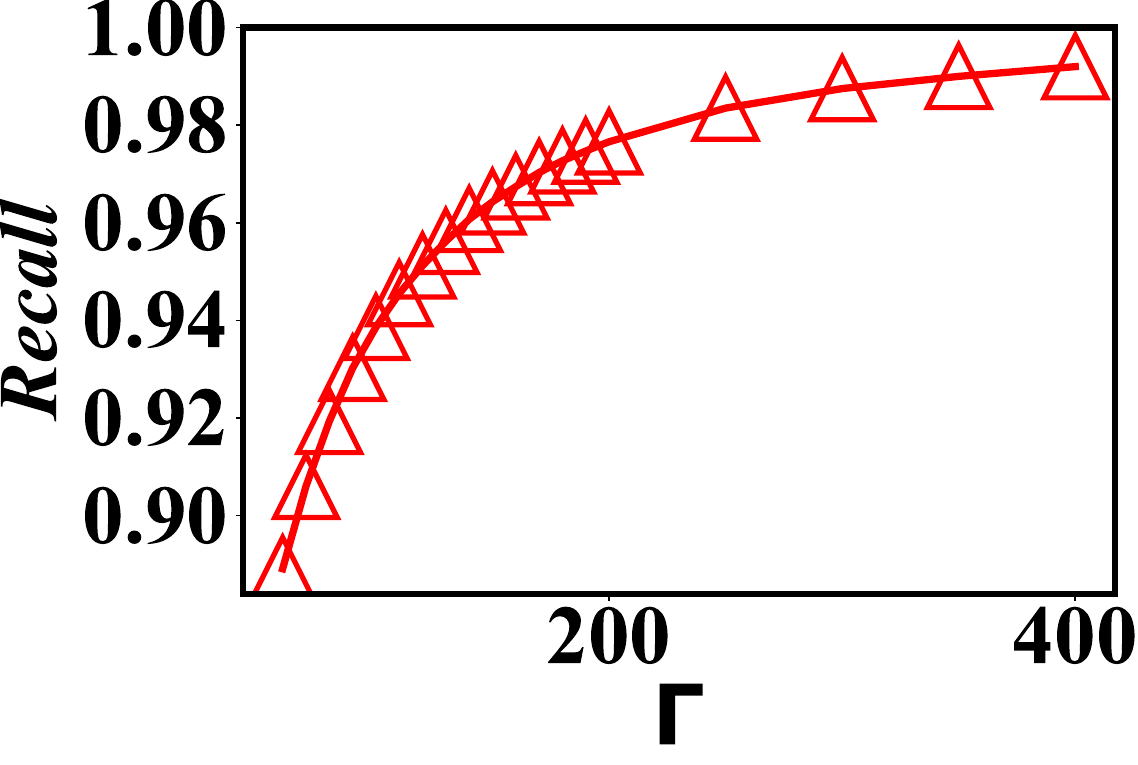}}{(a) \textit{Recall}}
  \stackunder[0.5pt]{\includegraphics[scale=0.21]{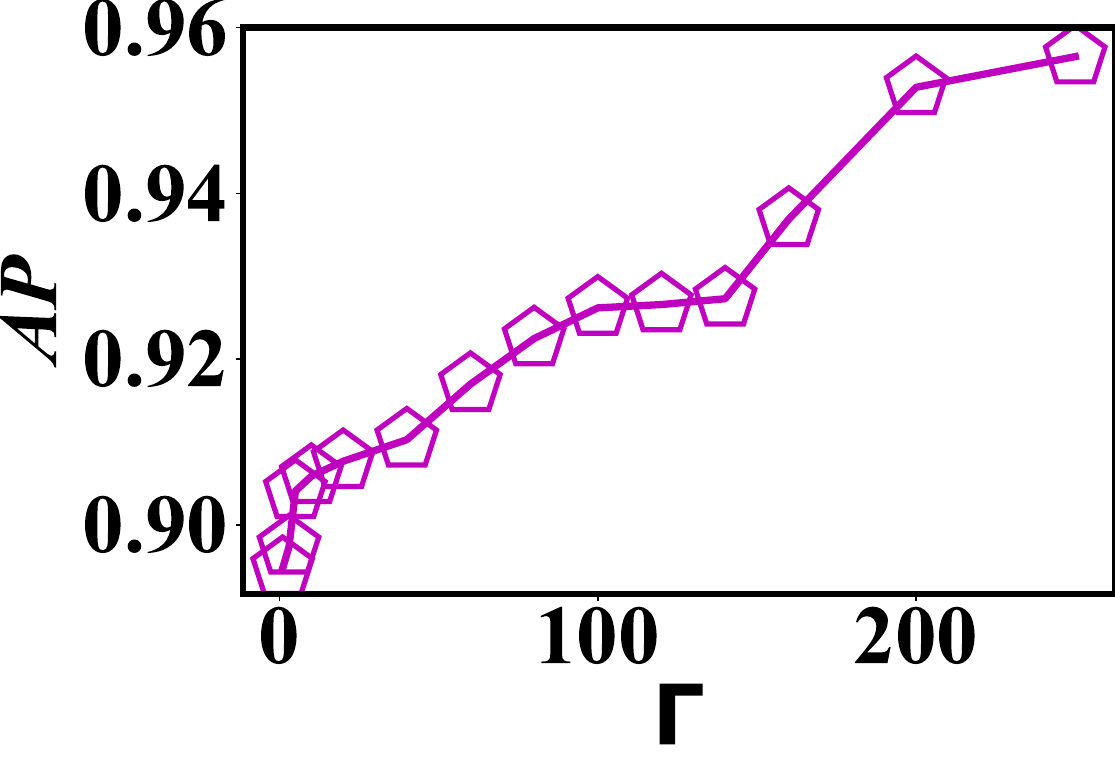}}{(b) \textit{AP}}
  \newline
  \stackunder[0.5pt]{\includegraphics[scale=0.21]{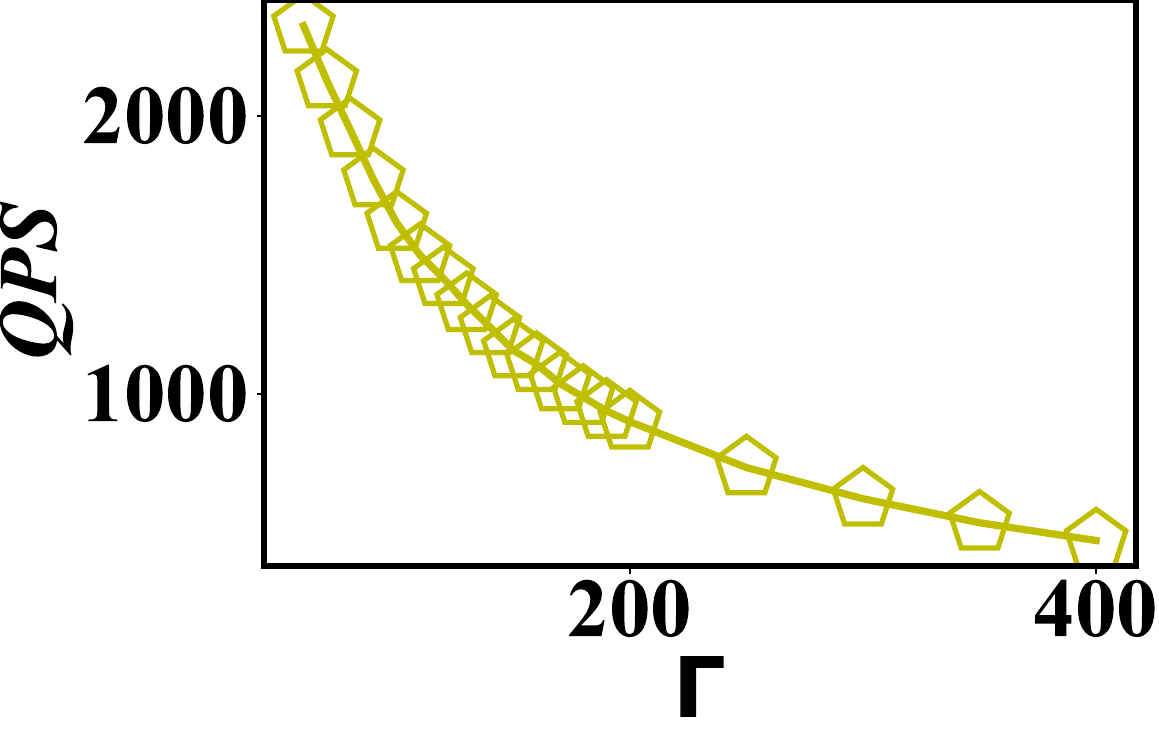}}{(c) \textit{QPS}}\hspace{0.2cm}
  \stackunder[0.5pt]{\includegraphics[scale=0.21]{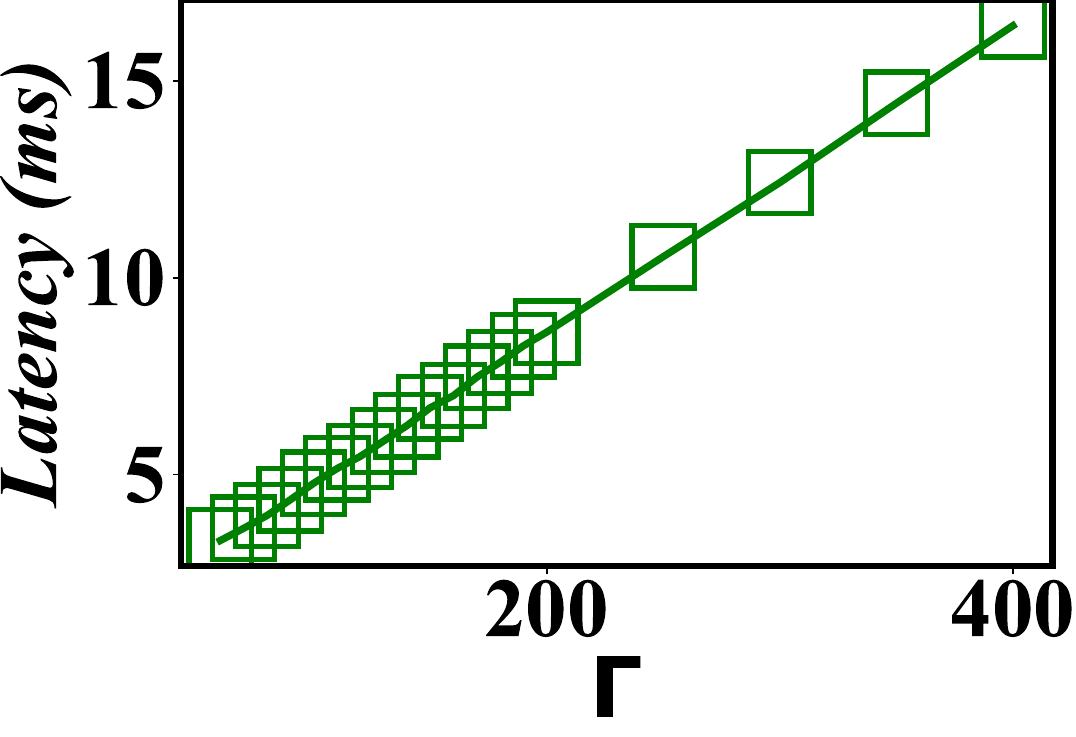}}{(d) \textit{Latency}}
  \newline
  \stackunder[0.5pt]{\includegraphics[scale=0.21]{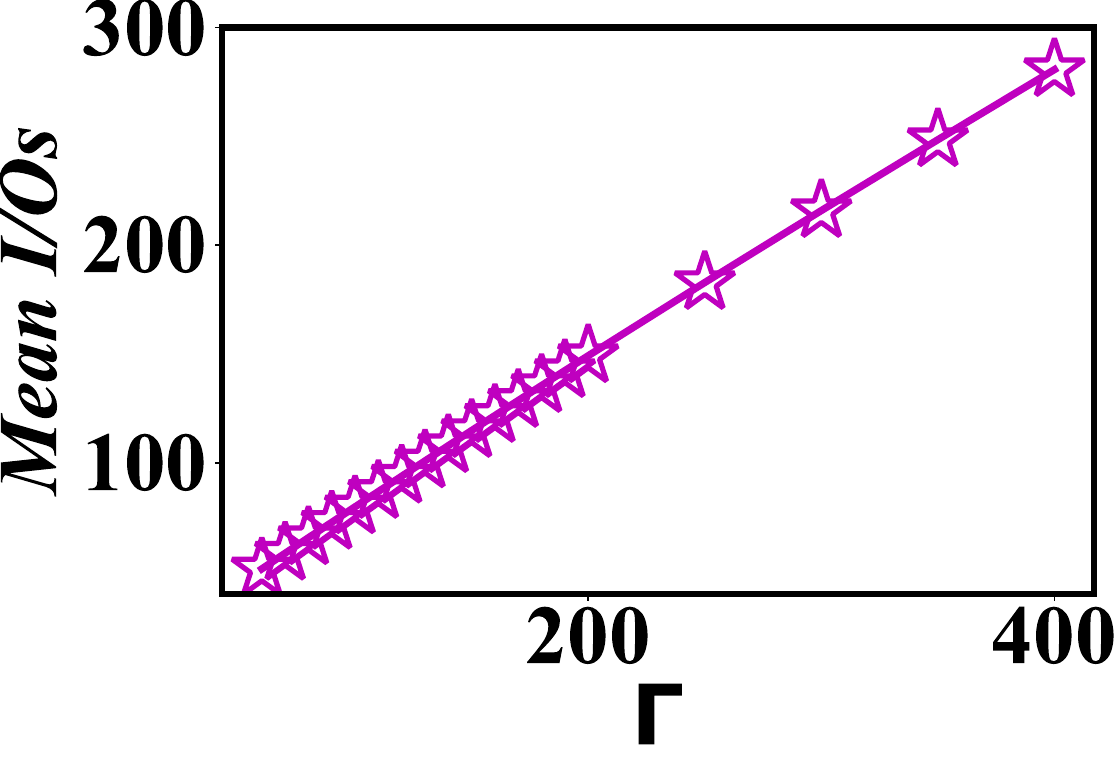}}{(e) Mean I/Os}
  \stackunder[0.5pt]{\includegraphics[scale=0.21]{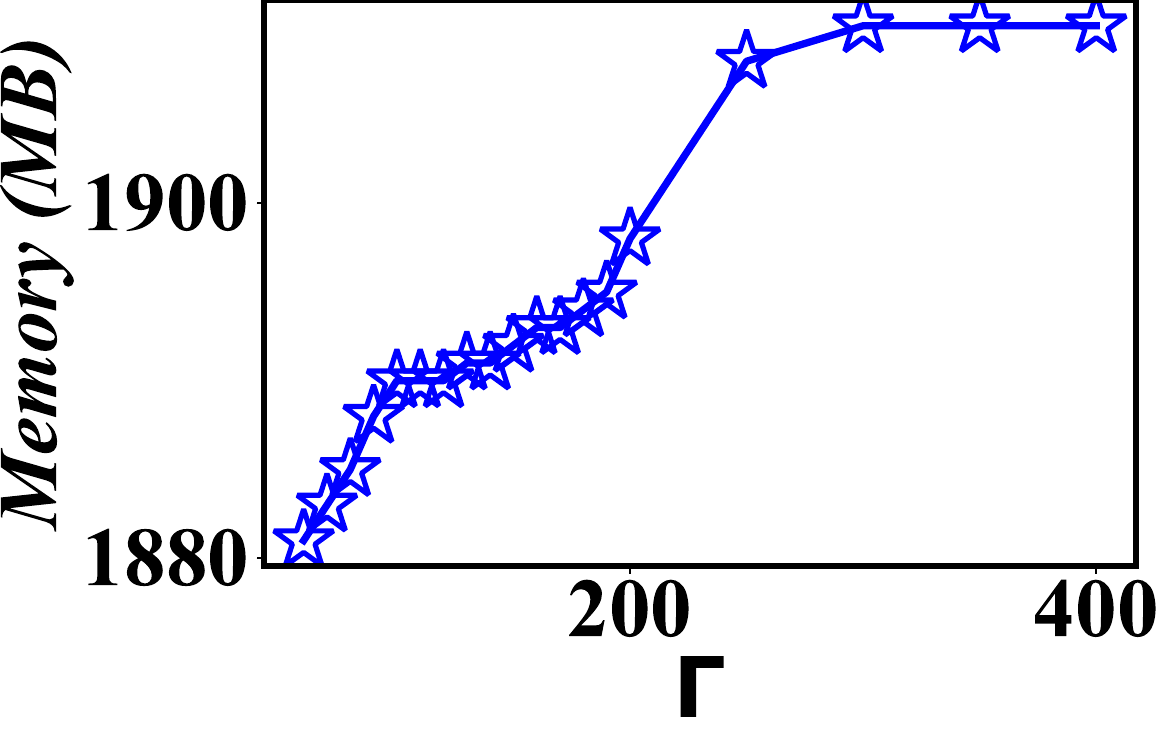}}{(f) Memory overhead}
  \newline
  \caption{Effect of candidate set size $\Gamma$ on search performance}
  \label{fig: search list size}
\end{figure}

\section{Parameters for searching}
\label{appendix: parameters search}
We modify the candidate set size $\Gamma$ to obtain different \textit{Recall} or \textit{AP} when searching. Fig. \ref{fig: search list size} shows the effect of $\Gamma$ on the BIGANN dataset. The results indicate that a larger $\Gamma$ increases the accuracy but also slows down the search and consumes more memory. This parameter applies to both in-memory navigation graph and disk-based graph. Tab. \ref{tab: parameters for search pruning ratio} lists the pruning ratio on different datasets we used in our experiments.

\setlength{\textfloatsep}{0cm}
\setlength{\floatsep}{0cm}
\begin{table}[!tb]
  \centering
  \setlength{\abovecaptionskip}{0.1cm}
  \setlength{\belowcaptionskip}{0.3cm}
  \setstretch{0.8}
  \fontsize{6.5pt}{3.3mm}\selectfont
  \caption{Memory overhead and search performance under different datasets.}
  \label{tab: memory overhead segment}
  \setlength{\tabcolsep}{.02\linewidth}{
  \begin{tabular}{|l|l|l|l|l|}
    \hline
    \textbf{Dataset}$\downarrow$ & \textbf{Method} & \textbf{\textit{Recall}}($k=10$) & \textbf{Memory overhead (MB)} & \textbf{\textit{QPS}} \\
    \hline
    \multirow{2}*{\textbf{BIGANN}} & {\name} & \textbf{0.9920} & \textbf{1,910} & \textbf{475} \\
    \cline{2-5}
    ~ & DiskANN & 0.9827 & 1,925 & 340 \\
    \hline
    \multirow{2}*{\textbf{DEEP}} & {\name} & \textbf{0.9950} & \textbf{1,213} & \textbf{504} \\
    \cline{2-5}
    ~ & DiskANN & 0.9927 & 1,244 & 379 \\
    \hline
  \end{tabular}
  }\vspace{0.2cm}
\end{table}

\section{Memory Cost and Search performance}
\label{appendix: memory cost and search performance}
DiskANN keeps PQ short codes, hot vertices and their neighbor IDs in memory. For {\name}, we mainly store in-memory navigation graph, PQ short codes, and the mapping of vertex IDs to block IDs. Our evaluation in {Appendix \ref{appendix: memory graph vs cache}} shows that navigation graph uses less memory than caching hot vertices and achieves better search performance. In one segment, the memory limit is 2 GB. We adjust the memory overhead by sample ratio for both {\name} and DiskANN (see {Appendix \ref{appendix: parameters of memory graph} and \ref{appendix: memory graph vs cache}}). Tab. \ref{tab: memory overhead segment} shows the memory overhead for similar \textit{Recall}. {\name} has lower memory overhead and higher \textit{QPS} for higher \textit{Recall}.

\setlength{\textfloatsep}{0cm}
\setlength{\floatsep}{0cm}
\begin{figure*}[!th]
\setlength{\abovecaptionskip}{0cm}
\setstretch{0.9}
\fontsize{8pt}{4mm}\selectfont
\begin{minipage}{0.665\textwidth}
  \setlength{\abovecaptionskip}{0.1cm}
  \setlength{\belowcaptionskip}{0cm}
  \centering
  \footnotesize
  \stackunder[0.7pt]{\includegraphics[scale=0.17]{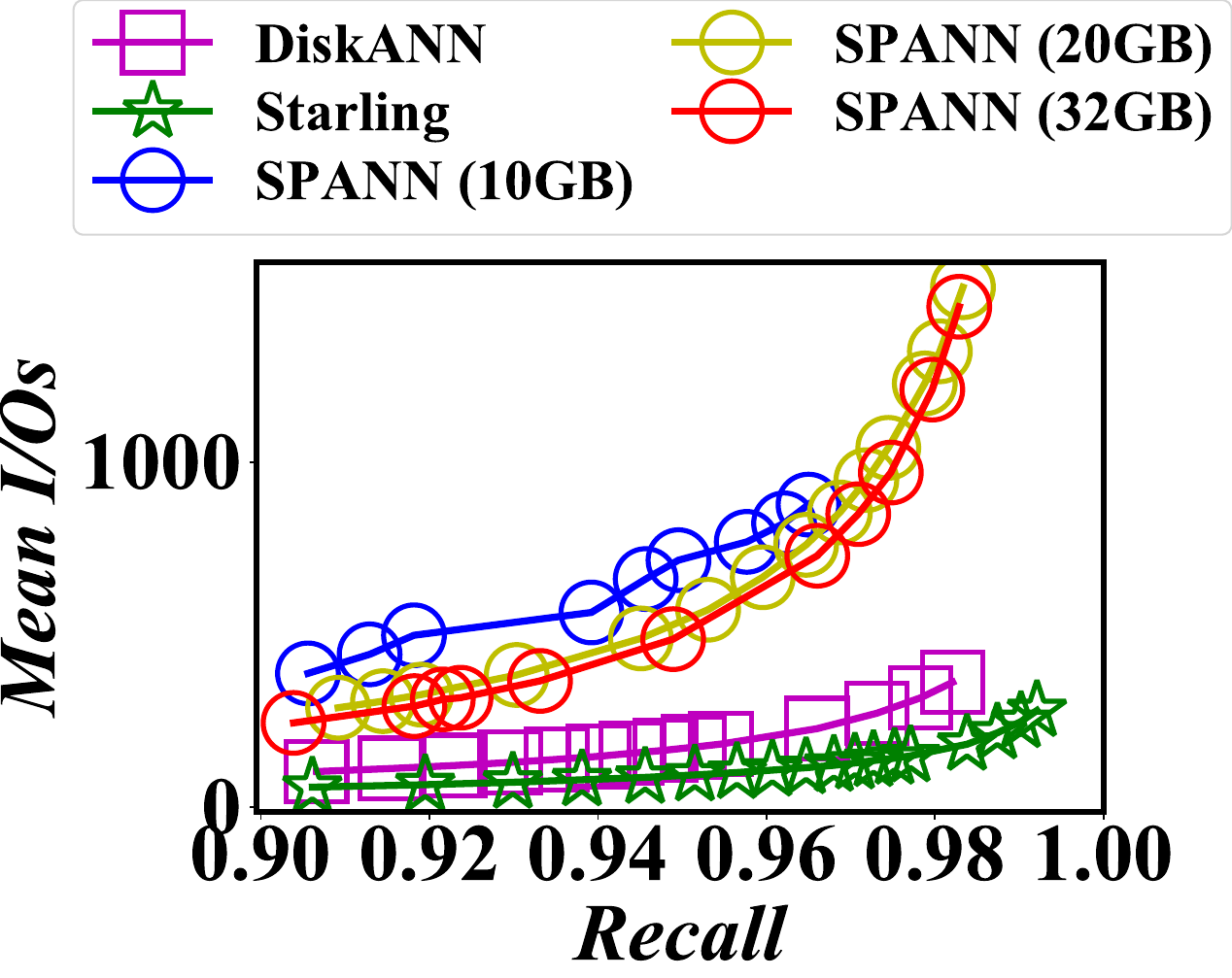}}{(a) BIGANN (4GB)}
  \hspace{0cm}
  \stackunder[0.7pt]{\includegraphics[scale=0.19]{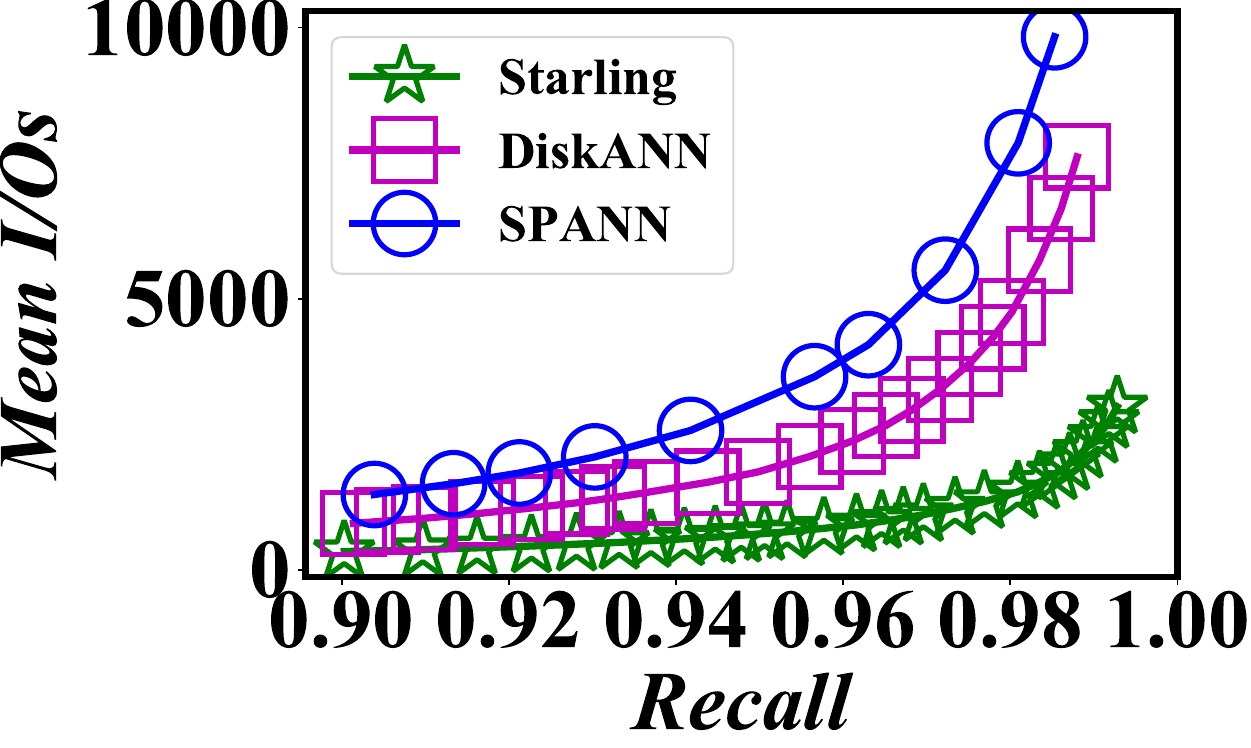}}{(b) BIGANN (8GB)}
  \hspace{0cm}
  \stackunder[0.7pt]{\includegraphics[scale=0.19]{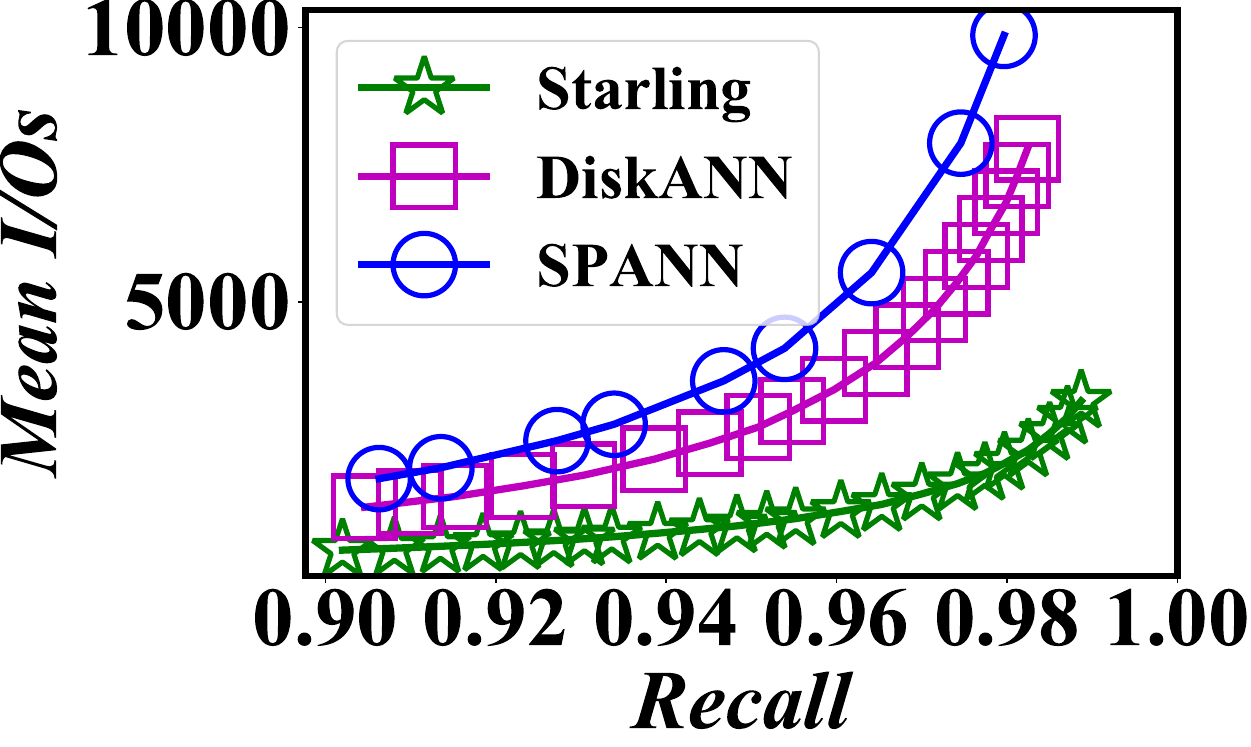}}{(c) BIGANN (16GB)}
  \caption{\textit{Mean I/Os} of search on different segments and data sizes.}
  \label{fig: disk io under different segment configurations}
\end{minipage}
\begin{minipage}{0.33\textwidth}
  \setlength{\abovecaptionskip}{0.1cm}
  \setlength{\belowcaptionskip}{0cm}
  \centering
  \footnotesize
  \stackunder[0.9pt]{\includegraphics[scale=0.19]{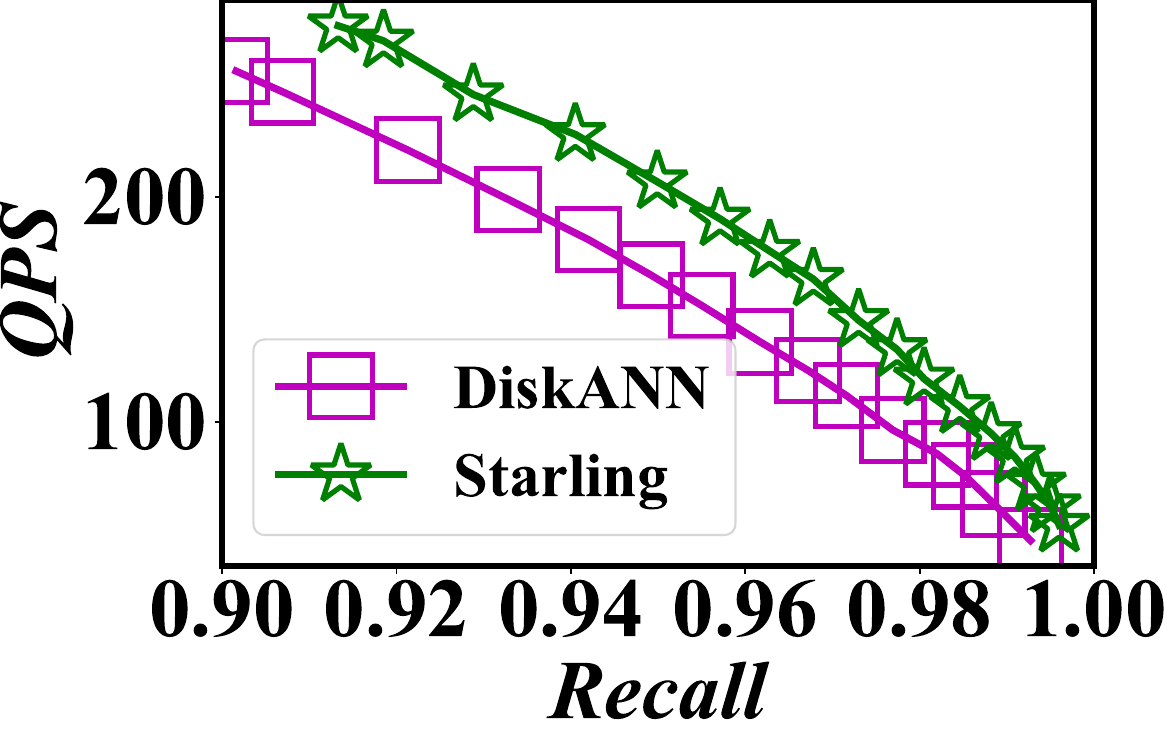}}{DEEP (310GB)}
  \caption{Effect of big data.}
  \label{fig: 341m deep data}
\end{minipage}
\vspace{-0.3cm}
\end{figure*}

\section{Related work}
\label{appendix: related work}
\noindent\textbf{Vector search algorithms.} High-dimensional vector similarity search (HVSS) is a well-studied topic, and most works aim to find an algorithm that balances efficiency and accuracy by preprocessing data \cite{graph_survey_vldb2021}. We can roughly categorize current algorithms into four common groups based on how they process the data: tree-based \cite{DasguptaF08, LuWWK20, MujaL14}; quantization-based \cite{PQ, ScaNN, AndreKS15}; hashing-based \cite{HuangFZFN15, GongWOX20, LiZSWT020}; and graph-based \cite{HNSW,NSG,NSSG}. Graph-based methods have shown to be very effective for HVSS in many empirical studies \cite{DPG,NSG,LiZAH20}. However, existing graph-based algorithms require both raw vector data and graph index to be stored in main memory. This leads to high memory consumption and limits the scalability to hundreds of millions of vectors \cite{SPANN}. {\name} builds a disk-based graph index and optimizes the block-level graph layout by block shuffling. It also maintains an in-memory navigation graph based on sampled data. Thus, {\name} solves the high memory overhead problem while achieving the state-of-the-art efficiency and accuracy trade-off.

Most of the research has focused on million-scale datasets in main memory. For web-scale vector search scenarios, existing vector search algorithms are limited by the main memory-served indexes \cite{DiskANN}. Current solutions include large-scale methods and distributed vector databases.

\vspace{0.5em}
\noindent\textbf{Large-scale solutions.} Large-scale solutions aim to reduce memory consumption from two perspectives: vector quantization and hybrid storage. On one hand, FAISS \cite{Faiss} uses vector quantization techniques (e.g., IVFPQ \cite{PQ}, OPQ \cite{OPQ}) to compress the vector data and store them in memory. However, the compression introduces quantization error, which lowers the search quality. On the other hand, some methods working on disk and heterogeneous memory (HM), such as GRIP \cite{zhang2019grip}, DiskANN \cite{DiskANN}, HM-ANN \cite{HM_ANN}, SPANN \cite{SPANN}, BBAnn \cite{Manu_zilliz}, allocate vector data or index to persistent storage. They use some hardware-oriented optimizations to alleviate the I/O bottleneck and support billion-scale vectors on a single machine. However, as data volume grows, it becomes harder to scale to distributed vector database systems for system features purpose \cite{deng2019pyramid}. Moreover, applying current disk-based solutions to the data segment of vector databases causes huge disk overhead or high I/O complexity. {\name} solves this problem by optimizing the data layout and search strategy following the state-of-the-art disk-based graph index paradigm.

\vspace{0.5em}
\noindent\textbf{Vector database systems.} Vector database research has gained more interest recently \cite{Milvus_sigmod2021,Manu_zilliz,Jingdong_paper}. AnalyticDB-V \cite{ADBV}, PASE \cite{PASE}, and VBase \cite{zhang2023vbase} extend relational databases to support vector data with the one-size-fits-all strategy. However, those systems require optimization for either transaction or analytical workloads and are not specialized for vector data management. Some purpose-built vector data management systems, such as Vearch \cite{Jingdong_paper} and Milvus \cite{Milvus_sigmod2021,Manu_zilliz}, store and search large-scale vector data efficiently by organizing data with segments for industrial applications. There are also some distributed solutions, e.g., Pyramid \cite{deng2019pyramid}, LANNS \cite{lanns}, LEQAT \cite{zhang2022leqat}. They focus on data segmentation strategies or query dispatching optimizations. They assign the vectors that are close to each other to the same segment and probe only a few segments when searching based on the query dispatching optimizations. However, they ignore the search optimization opportunities within the data segment of vector databases. In contrast, {\name} addresses the problem of how to search efficiently in the data segment by optimizing the data layout and search strategy.

\section{Main Parameters of SPANN and DiskANN}
\label{appendix: main parameters of SPANN and DiskANN}

Here, we discuss the key parameters in SPANN and DiskANN.

\begin{table}[!tb]
  \centering
  \setlength{\abovecaptionskip}{0.1cm}
  \setlength{\belowcaptionskip}{0.3cm}
  \setstretch{0.8}
  \fontsize{6.5pt}{3.3mm}\selectfont
  \caption{Parameter values of SPANN in our experiments.}
  \label{tab: parameters for SPANN}
  \setlength{\tabcolsep}{.015\linewidth}{
  \begin{tabular}{|l|l|l|l|}
    \hline
    \textbf{Parameter$\downarrow$} & \textbf{BIGANN} & \textbf{DEEP} & \textbf{Text2image} \\
    \hline
    $\epsilon$ & 2 & 3 & 5 \\
    \hline
    $\alpha$ (KB) & 12 & 48 & 48 \\
    \hline
    $\epsilon_1$ & 10.0 & 10.0 & 10.0 \\
    \hline
    $\epsilon_2$ & 8.0 & 8.0 & 8.0 \\
    \hline
  \end{tabular}
  }
\end{table}

\noindent\textbf{SPANN.} SPANN assigns a vector to multiple closest clusters if the distance between the vector and these clusters is nearly the same. The closure replica parameter ($\epsilon$) indicates the number of replicas for each vector in the closure clustering assignment. The posting list size ($\alpha$) reflects the length of each posting list (each posting list corresponds to a cluster). SPANN uses the RNG rule to reduce the similarity of two close posting lists, and the distance coefficient $\epsilon_1$ adjusts the RNG rule. In the search procedure, SPANN employs the query-aware dynamic pruning technique to reduce the number of posting lists to be searched based on the distance between the query and centroids, where $\epsilon_2$ is the pruning coefficient. In our experiments, we increase the maximum number of posting lists to be searched to obtain different recall quality. The eventual parameters used by SPANN are provided in Tab. \ref{tab: parameters for SPANN}.

\begin{table}[!tb]
  \centering
  \setlength{\abovecaptionskip}{0.1cm}
  \setlength{\belowcaptionskip}{0.3cm}
  \setstretch{0.8}
  \fontsize{6.5pt}{3.3mm}\selectfont
  \caption{Parameter values of DiskANN in our experiments.}
  \label{tab: parameters for DiskANN}
  \setlength{\tabcolsep}{.015\linewidth}{
  \begin{tabular}{|l|l|l|l|l|}
    \hline
    \textbf{Parameter$\downarrow$} & \textbf{BIGANN} & \textbf{DEEP} & \textbf{SSNPP} & \textbf{Text2image} \\
    \hline
    $\Lambda$ & 31 & 48 & 48 & 54 \\
    \hline
    $L$ & 128 & 128 & 128 & 128 \\
    \hline
    $B$ & 0.500 & 0.256 & 0.960 & 0.470 \\
    \hline
    $\pi$ & 0.06 & 0.07 & 0.12 & 0.10 \\
    \hline
  \end{tabular}
  }\vspace{0.2cm}
\end{table}

\noindent\textbf{DiskANN.}
We explain the main parameters used in building the index and executing a query for DiskANN. The parameter $\Lambda$ represents the maximum degree of the disk-based graph. Increasing $\Lambda$ results in larger indexes and longer construction time, but it also improves search accuracy. Another parameter, $L$, denotes the size of candidate neighbors during index construction. A larger value of $L$ requires more time for construction, but it produces indexes with higher \textit{Recall} while maintaining the same search efficiency. It is important to note that $L$ should be at least as large as $\Lambda$. The parameter $B$ sets the limit on the memory footprint of PQ short codes. This determines the level of compression applied to the vector data stored in memory. A smaller value of $B$ reduces the memory overhead for the PQ short codes. In DiskANN, certain hot vertices and their corresponding neighbor IDs from the disk-based graph are loaded into memory. The ratio of loaded vertices is denoted by $\pi$. The specific parameters used by DiskANN in our experimental datasets are provided in Tab. \ref{tab: parameters for DiskANN}. We adjust the size of the candidate set to achieve different levels of recall quality when searching.

\section{Different Segment and Data Size Setups}
\label{appendix: segment and data size setup}

In Fig. \ref{fig: disk io under different segment configurations}, we present an evaluation of the \textit{Mean I/Os} for search operations using different segment and data size configurations.

\noindent\textbf{Varying disk capacity under the same memory space.}
In this experiment, we maintain a fixed dataset size of 4GB and evaluate different segment setups while keeping the memory space at 2GB. The disk capacities vary as 10GB, 20GB, and 32GB, respectively. Fig. \ref{fig: disk io under different segment configurations}(a) illustrates the \textit{Mean I/Os} for the BIGANN dataset. For DiskANN and {\name}, we choose parameter configurations that provide the best trade-off between \textit{Mean I/Os} and \textit{Recall}, resulting in an index size of less than 10GB for the 4GB BIGANN dataset. Therefore, we only plot one \textit{Mean I/Os}-\textit{Recall} curve for DiskANN and {\name} each. In the case of SPANN, the \textit{Mean I/Os} decreases as the disk capacity increases. This is due to the larger disk space allowing for more duplication of boundary data points (which are shared by multiple clusters), thereby reducing disk I/Os required to access these boundary points. We observe that the \textit{Mean I/Os} of SPANN gradually approaches that of DiskANN as the disk space increases. However, {\name} consistently maintains significantly lower disk I/Os compared to DiskANN and SPANN. Tab. \ref{tab:spann index size} presents the index size of SPANN for the 4GB BIGANN dataset at different $\epsilon$ values.

\begin{table}[!tb]
  \centering
  \setlength{\abovecaptionskip}{0cm}
  \setlength{\belowcaptionskip}{0cm}
  \setstretch{0.8}
  \fontsize{6.5pt}{3.3mm}\selectfont
  \caption{Index size of SPANN on 4GB BIGANN dataset.}
  \label{tab:spann index size}
  \setlength{\tabcolsep}{.015\linewidth}{
  \begin{tabular}{|l|l|l|l|l|l|l|l|l|}
    \hline
    \textbf{$\epsilon$} & 1 & 2 & 3 & 4 & 5 & 6 & 7 & 8 \\
    \hline
    Index size (GB) & 5.9 & 8.4 & 12 & 16 & 19 & 23 & 26 & 28 \\
    \hline
  \end{tabular}
  }
\end{table}

\noindent\textbf{Varying dataset size under the same memory and disk space.}
In this experiment, we maintain a fixed memory space of 2GB and disk capacity of 32GB. We evaluate different dataset sizes, specifically 4GB, 8GB, and 16GB. Fig. \ref{fig: disk io under different segment configurations}(a) and (c) depict the search performance of different methods on the 8GB and 16GB BIGANN datasets, respectively (refer to Fig. \ref{fig: disk io under different segment configurations}(a) for the 4GB dataset size).
We adjust the parameters of all methods to achieve the best trade-off between \textit{Mean I/Os} and \textit{Recall} within the limitation of the segment space. From the results, we observe that {\name} consistently outperforms the competing methods in terms of \textit{Mean I/Os}. Furthermore, the performance gain of {\name} compared to the competitors increases as the dataset size grows. For SPANN, a larger dataset size limits its ability to duplicate data, resulting in numerous \textit{Mean I/Os}.

\section{Evaluation on 341M DEEP dataset}
\label{appendix: evaluation on 341M deep dataset}
We divide the 341M DEEP dataset into 31 data segments. Each segment was configured with 2GB of memory and 10GB of disk space. Since each query node has only 32GB of memory available, we assign 31 segments to two query nodes. The final search results are obtained by continuously merging candidates from each segment. Both {\name} and DiskANN adopt this setting to ensure fairness in the comparison.
According to the results shown in Fig. \ref{fig: 341m deep data}, {\name} achieves up to 56\% higher \textit{QPS} than DiskANN in the same high recall regime, such as \textit{Recall} $> 0.96$.

\section{Discussion on Data Locality}
\label{appendix: discussion on data locality}
Data locality, a coveted trait in various domains, presents a unique challenge for graph indexes constructed on high-dimensional vectors. Relational databases naturally achieve data locality through diverse partitioning strategies grouping relevant data into the same cluster \cite{zamanian2015locality,neumann2011efficiently} and storing them in the same disk sector. Graph engines like Pregel \cite{Pregel} and Ligra \cite{Ligra} exhibit inherent data locality for real-world graphs (e.g., social networks \cite{pacaci2019experimental,WeiYLL16}) since their degree follows a power-law distribution, and neighbors tend to cluster \cite{abbas2018streaming}. However, the nature of graph indexes introduces a distinct scenario where neighbors exhibit similarity and navigation traits (long links), and the degree is uniform (constant for all vertices) \cite{graph_survey_vldb2021}. Hence, the neighbors of a vertex may scatter across clusters, as reported in the HNSW paper \cite{HNSW}, making data locality more challenging.